\documentclass[12pt]{article}

\usepackage{amsmath,amssymb,pifont}

\usepackage{rotating}
\usepackage{amstext}
\usepackage{amsthm}
\usepackage{multirow}
\usepackage{booktabs}
\usepackage[skip=0pt]{subcaption}
\usepackage{times}

\usepackage{paralist}

\usepackage[shortlabels]{enumitem}

\usepackage{array}
\usepackage{siunitx}

\usepackage{bbm, dsfont}
\usepackage{mathtools}\usepackage[letterpaper, margin=1in]{geometry}
\usepackage{enumitem}

\usepackage{graphicx}
\usepackage{color}
\usepackage{array}
\usepackage{xspace}
\usepackage{fancyhdr}

\usepackage{hyperref}
\hypersetup{
    colorlinks=true,
    linkcolor=blue,
    filecolor=magenta,      
    urlcolor=cyan}
\usepackage[capitalise]{cleveref}

\usepackage{algorithm}
\usepackage[noend]{algpseudocode}

\usepackage[style=alphabetic,backend=bibtex,maxalphanames=10,maxbibnames=20,maxcitenames=10,giveninits=true,doi=false,url=true,backref=true]{biblatex}

\newcommand{\citet}[1]{\AtNextCite{\AtEachCitekey{\defcounter{maxnames}{999}}}\textcite{#1}}
\newcommand{\citep}[1]{\parencite{#1}}

\newtheorem{lem}{Lemma}[section]

\newtheorem{thm}[lem]{Theorem}
\newtheorem{theorem}[lem]{Theorem}
\newtheorem{cor}[lem]{Corollary}

\newtheorem{rem}[lem]{Remark}
\newtheorem{prop}[lem]{Proposition}

\newtheorem{defn}[lem]{Definition}

\newtheorem{conjecture}[lem]{Conjecture}

\numberwithin{equation}{section}

\iftrue  
    \Crefname{thm}{Theorem}{Theorems}
    \Crefname{cor}{Corollary}{Corollaries}
    \Crefname{lem}{Lemma}{Lemmas}
    \Crefname{prop}{Proposition}{Propositions}
    \Crefname{assumption}{Assumption}{Assumptions}
    \Crefname{definition}{Definition}{Definitions}
    \Crefname{claim}{Claim}{Claims}
    \Crefname{section}{Section}{Sections}
    \Crefname{appendix}{Appendix}{Appendices}
    \Crefname{figure}{Figure}{Figures}
    \crefformat{equation}{Equation~\textup{#2#1#3}}
\else  
    \Crefname{thm}{Thm.}{Thms.}
    \Crefname{cor}{Cor.}{Cors.}
    \Crefname{lem}{Lem.}{Lems.}
    \Crefname{prop}{Prop.}{Props.}
    \Crefname{assumption}{Assumption}{Assumptions}  
    \Crefname{definition}{Def.}{Defs.}
    \Crefname{claim}{Claim}{Claims}
    \Crefname{section}{Sec.}{Secs.}
    \Crefname{appendix}{App.}{Apps.}
    \crefformat{equation}{Eq.~\textup{#2#1#3}}
    \Crefname{figure}{Fig.}{Figs.}
\fi

\newcommand{\nbuf}{d}  
\newcommand{\dgr}{d}  

\newcommand{\outp}{\omega} 
\newcommand{\houtp}{\hat{\outp}} 

\DeclareMathOperator{\poly}{poly}
\newcommand{\maxerrtime}{\ensuremath{\calO(\poly(\nbuf)\log \nd)}}

\newcommand{\rc}{r}
\newcommand{\hth}{\hat{\theta}}
\newcommand{\rinv}{s}

\newcommand{\BLT}{\texttt{BLT}\xspace}
\newcommand{\BLTs}{{\BLT}s\xspace}
\newcommand{\RABLT}{\texttt{RA-BLT}\xspace}
\newcommand{\OptBLT}{\texttt{Opt-BLT}\xspace}
\newcommand{\OptBLTd}[1]{\mbox{$\text{\OptBLT}(\nbuf\!=\!#1)$}\xspace}
\newcommand{\RABLTd}[1]{\mbox{$\text{\RABLT}(\nbuf\!=\!#1)$}\xspace}

\DeclareMathOperator{\MaxErr}{MaxErr}

\newcommand{\nd}{\ensuremath{n}} 
\newcommand{\dimdim}{{\nd \times \nd}}
\newcommand{\mdim}{m}  

\newcommand{\mfid}{n'} 

\newcommand{\seqk}[1]{(#1)_{k=0}^\infty}

\DeclareMathOperator{\LTT}{M}

\newcommand{\bfA}{\ensuremath{A}}
\newcommand{\bfB}{\ensuremath{B}}
\newcommand{\bfC}{\ensuremath{C}}

\renewcommand{\epsilon}{\varepsilon}

\newcommand{\inv}{^{-1}}

\newcommand{\tp}{^T}  

\DeclareMathOperator{\sens}{sens}

\newcommand{\bftheta}{\ensuremath{\theta}}  

\newcommand{\idx}[2]{_{#1, #2}}

\newcommand{\mech}{\mathcal{M}}

\newcommand{\abs}[1]{|#1|}
\newcommand{\norm}[1]{\left\| #1 \right\|}

\newcommand{\ltwo}[1]{\left\|#1\right\|_2}
\newcommand{\lfrob}[1]{\left\|#1\right\|_F}

\newcommand{\ip}[2]{\langle #1, #2\rangle}

\newcommand{\bfb}{\ensuremath{{b}}}
\newcommand{\bfc}{\ensuremath{{c}}}

\newcommand{\bfw}{\ensuremath{{w}}}
\newcommand{\bfx}{\ensuremath{{x}}}

\newcommand{\bfz}{\ensuremath{{z}}}

\newcommand{\calM}{\ensuremath{\mathcal{M}}}
\newcommand{\calN}{\ensuremath{\mathcal{N}}}
\newcommand{\calO}{\ensuremath{\mathcal{O}}}

\newcommand{\R}{\mathbb{R}}
\newcommand{\N}{\mathbb{N}}

\newcommand{\C}{\mathbb{C}}

\newcommand{\comb}[2]{#1 \mathbin{\raisebox{.2ex}{ \hspace{-.4em}$\bigcirc$\hspace{-.75em}{\rm b}\hspace{.15em}}} #2}
\newcommand{\comc}[2]{#1 \mathbin{\raisebox{.2ex}{ \hspace{-.4em}$\bigcirc$\hspace{-.75em}{\rm c}\hspace{.15em}}} #2}

\newcommand{\ex}[2]{{\ifx&#1& \mathbb{E} \else \underset{#1}{\mathbb{E}} \fi \left[#2\right]}}
\newcommand{\pr}[2]{{\ifx&#1& \mathbb{P} \else \underset{#1}{\mathbb{P}} \fi \left[#2\right]}}
\newcommand{\exc}[3]{{\ifx&#1& \mathbb{E} \else \underset{#1}{\mathbb{E}} \fi \left[ #2 \middle| #3 \right]}}
\newcommand{\prc}[3]{{\ifx&#1& \mathbb{P} \else \underset{#1}{\mathbb{P}} \fi \left[ #2 \middle| #3 \right]}}
\newcommand{\var}[2]{{\ifx&#1& \mathbb{V} \else \underset{#1}{\mathbb{V}} \fi \left[#2\right]}}

\newcommand{\Jmat}{\mathbb{J}}
\newcommand{\br}[1]{\left({#1}\right)}
\newcommand{\One}{\mathbf{1}}
\newcommand{\evec}{\mathbf{e}}
\newcommand{\tran}[1]{{#1}^\top}
\newcommand{\invb}[1]{{\br{#1}^{-1}}}
\newcommand{\Sym}{\mathbb{S}}
\newcommand{\inner}[2]{\left\langle{#1},{#2}\right\rangle}
\newcommand{\Amatrix}{\mathbb{A}}

\begin{document}

\title{Efficient and Near-Optimal Noise Generation\\for Streaming Differential Privacy}
    \author{
    Krishnamurthy (Dj) Dvijotham\thanks{Google DeepMind \dotfill \texttt{dvij@google.com}}
    \and
    H.~Brendan McMahan\thanks{Google Research \dotfill \texttt{mcmahan@google.com}}
    \and
    Krishna Pillutla\thanks{IIT Madras. Work done at Google. \dotfill \texttt{pillutla@cs.washington.edu}}
    \and
    Thomas Steinke\thanks{Google DeepMind \dotfill \texttt{steinke@google.com}}
    \and
    Abhradeep Thakurta\thanks{Google DeepMind \dotfill \texttt{athakurta@google.com}}
    }
    \footnotetext{Alphabetical author order.}

\maketitle

\begin{abstract}
    In the task of differentially private (DP) continual counting, we receive a stream of increments and our goal is to output an approximate running total of these increments, without revealing too much about any specific increment.
    Despite its simplicity, differentially private continual counting has attracted significant attention both in theory and in practice.
    Existing algorithms for differentially private continual counting are either inefficient in terms of their space usage or add an excessive amount of noise, inducing suboptimal utility.
    
    The most practical DP continual counting algorithms add carefully correlated Gaussian noise to the values. The task of choosing the covariance for this noise can be expressed in terms of factoring the lower-triangular matrix of ones (which computes prefix sums). We present two approaches from this class (for different parameter regimes) that achieve near-optimal utility for DP continual counting and only require logarithmic or polylogarithmic space (and time).
    
    Our first approach is based on a space-efficient streaming matrix multiplication algorithm for a class of Toeplitz matrices. 
    We show that to instantiate this algorithm for DP continual counting, it is sufficient to find a low-degree rational function that approximates the square root on a circle in the complex plane. We then apply and extend tools from  approximation theory to achieve this.
    We also derive efficient closed-forms for the objective function for arbitrarily many steps, and show direct numerical optimization yields a highly practical solution to the problem.
    Our second approach combines our first approach with a recursive construction similar to the binary tree mechanism.
\end{abstract}

\clearpage 

\tableofcontents

\newpage
\section{Introduction}\label{sec:intro}

The simplest task in differentially private data analysis is \emph{counting}.
This task has an equally simple (and optimal) algorithm: we simply add Laplace or Gaussian noise to the final count \cite{DMNS}. 

\emph{Continual counting} \cite{Dwork-continual,CSS11-continual} asks that we not only release the final count, but that we provide a running total, i.e., we release all the partial sums.
This task and its variants have proved to be a surprisingly complex problem in the theory of differential privacy \cite{beimel2013private,bun2015differentially,bun2018composable,alon2019private,kaplan2020privately,gillenwater2021differentially,cohen2023optimal}.

Continual counting is also critically important in practice. 
State-of-the-art methods for differentially private machine learning rely on (high-dimensional extensions of) continual counting algorithms to maintain their state in a differentially private manner \cite{kairouz2021practical,Cutosky21,kairouz2021practical, denisov22matfact, choquette22multiepoch,choquette2024amplified,xu2023federated}. (See \cref{sec:dpftrl} for discussion of this application.)
This has motivated significant work on making differentially private continual counting as accurate as possible \cite{mathias1993hadamard,matouvsek2020factorization,fichtenberger2022constant,henzinger2023almost,henzinger2024unifying,andersson2024smooth}.

Practical algorithms for continual counting add Gaussian noise to the partial sums. However, this noise is not independent: it must be carefully correlated across steps to achieve good utility. Choosing an appropriate multivariate Gaussian can be formulated as a \emph{matrix factorization} problem, as we now describe.

To convey the key ideas in this introduction, we consider summing $\nd$ scalars; this discussion readily extends to summing vectors.
Let $x \in \mathbb{R}^\nd$ contain the $\nd$ terms we wish to sum (over $\nd$ iterations or steps).
Let $A \in \{0,1\}^{\nd \times \nd}$ be the lower-triangular all-ones matrix so that $Ax$ is the vector of partial sums -- i.e., $(Ax)_k = \sum_{i \le k} x_i$ for each $k$. We factor $A=BC$, where $B,C^T \in \mathbb{R}^{\nd \times \nd'}$. 
The corresponding differentially private continual counting algorithm $\mech$ is given by
\begin{equation}
    \mech(x) \coloneqq B(Cx+z) = Ax + Bz = A(x+C^\dagger z), \label{eq:mfviews}
\end{equation}
where $z \gets \mathcal{N}(0,\sigma^2I)$ is a vector of independent Gaussian noise.\footnote{$C^\dagger=A^{-1}B$ is the inverse of $C$ whenever $C$ is a square matrix. However, we will also consider settings where $C$ is not a square matrix and so $C^\dagger$ is a suitable pseudo-inverse.}
Equivalently, the output of the mechanism can be written as $\mech(x)=Ax+\widehat{z}$ where $\widehat{z} \gets \mathcal{N}(0,\Sigma)$ with $\Sigma = \sigma^2 B B^T$.

\paragraph{Mechanism privacy:} We assume that one person can change only one entry of the input $x$ by at most $1$, which means $Cx$ can change by at most the maximum column norm of $C$, denoted  
\[
\|C\|_{1 \to 2} \coloneqq \max_j \sqrt{\sum_i C\idx{i}{j}^2}.
\]
This is the $L_2$ sensitivity of $Cx$. Then, following the standard privacy analysis of Gaussian noise addition \citep{dwork2014algorithmic,steinke2022composition}, to ensure that $Cx+z$ is $(\varepsilon,\delta)$-differentially private, we scale the noise $z \gets \mathcal{N}(0,\sigma^2I)$ to have standard deviation  $\sigma = \zeta \|C\|_{1\to2}$, where $\zeta = O(\frac1\varepsilon\sqrt{\log(1/\delta)})$ is a noise multiplier depending only on the privacy parameters $(\varepsilon, \delta)$.\footnote{Importantly, particularly for machine learning applications, this privacy analysis can be extended to the case where each $x_k$ is adaptively chosen depending on the prefix sums $0, \dots, k-1$ already released by the mechanism \citep{denisov22matfact}.} By postprocessing, this implies $\mech(x)=B(Cx+z)$ is $(\varepsilon,\delta)$-differentially private.

\paragraph{Mechanism utility:}
The error of the mechanism $\mech$ is given by $\mech(x) - Ax = Bz$. Thus, the root mean squared error of the $i$\textsuperscript{th} partial sum can be calculated as 
\[
\displaystyle
    \sqrt{\mathbb{E}_{\calM}{(\calM(x)-Ax)_i^2}}=\sqrt{\mathbb{E}_{z}\left[(Bz)_i^2\right]} = \sqrt{\sigma^2  \sum_j B_{i,j}^2},
\]
which scales with the norm of the corresponding row of $B$.
We take our notion of utility to be the maximum\footnote{We can also consider other measures of error such as sum of variances (rather than the max). This corresponds to a different norm: $\ex{}{\|Bz\|_2^2} = \sigma^2 \, \lfrob{B}^2 $.} such error over all $\nd$ partial sums released, which scales with the maximum row norm 
\[
\|B\|_{2 \to \infty} \coloneqq \max_i \sqrt{\sum_j B_{i,j}^2}.
\]
%
Thus, the maximum root mean squared error of the mechanism's answers is
\begin{align} \label{eq:max-rms-error}
    \max_i \sqrt{\ex{\calM}{(\calM(x)-Ax)_i^2}} &= \max_i \sqrt{\ex{z \gets \mathcal{N}(0,\sigma^2I)}{(Bz)_i^2}} = \max_i \sqrt{ \sigma^2 \sum_j B_{i,j}^2} \notag\\&= \sigma \|B\|_{2\to\infty} = \zeta \|B\|_{2\to\infty} \|C\|_{1\to2} =: \zeta \MaxErr(B,C),
\end{align}
where we take $\sigma=\zeta \|C\|_{1\to 2}$ to ensure $(\varepsilon, \delta)$-differential privacy, and we define
\begin{equation}\label{eq:maxerr}
\MaxErr(B, C) 
  \coloneqq  \|B\|_{2 \to \infty} \|C\|_{1\to2}
  = \sqrt{ \max_i \sum_j B_{i,j}^2 } \sqrt{ \max_j \sum_i C_{i,j}^2 }.
\end{equation}
Since the error in \cref{eq:max-rms-error}  can be written as a product of the noise multiplier $\zeta$ (which does not depend on the factorization $B,C$) and $\MaxErr(B,C)$ (which does not depend on the differential privacy parameters), the same factorization will minimize error for all settings of the privacy parameters. Thus, we suppress $\zeta$ for the remainder of this paper.  

\paragraph{Matrix factorization:}
To summarize, our goal is to solve the matrix factorization problem
\begin{equation}
    \text{minimize}_{B,C} ~\MaxErr(B, C) ~~~~~ \text{ subject to } B C = A.
    \label{eq:mfobj}
\end{equation}
The optimal value of this objective is denoted $\gamma_2(A)$ and is known as the \emph{gamma-two factorization norm} of $A$.
Prior work has obtained near-optimal factorizations of $A$ (and related matrices) \cite{mathias1993hadamard,matouvsek2020factorization}. The optimal value is 
\begin{equation}
    \gamma_2(A) \coloneqq \inf \big\{ \MaxErr(B,C) : B, C \in \mathbb{R}^{\nd \times \nd}, BC=A \big\} = \frac{\log(\nd)}{\pi} \pm O(1).
\end{equation}
However, the existing factorizations are either far from optimal in terms of their utility or are impractical due to their high computational cost, which we turn to next. 

\paragraph{Computationally efficient noise generation:}
The computational cost of implementing this matrix factorization approach for differentially private continual counting is dominated by the matrix-vector multiplication $B z$, where $z\in\mathbb{R}^n$ consists of i.i.d.~Gaussian noise $z \gets \mathcal{N}(0,\sigma^2 I)$.
The main computational challenge is that our algorithm cannot store the increments $x$ or the seed noise $z$ (or even the matrices $B$ or $C$) in memory; this would require memory that is linear in the number of steps $\nd$, which can be prohibitively large.

To give a sense of scale, in private machine learning applications, we often need to compute $\nd>10^6$ partial sums (each coming from an iteration of stochastic gradient descent or a related algorithm), over $\mdim>10^{10}$ dimensions (the size of the model and its gradients).\footnote{In this introductory discussion, we have considered sums over $\mdim=1$ dimension. Extending to $\mdim>1$ is essentially a matter of running $\mdim$ parallel copies of the algorithm. In our technical sections, we take $x$ to be a matrix whose rows we wish to sum, rather than a vector.}
Thus we need to be able to compute $\mech(x)$ without storing all of the data or noise in memory, as this would take $\nd\mdim>10^{16}$ units of memory -- i.e., over 40 petabytes if we use 32-bit floating point numbers.\footnote{Throughout this paper we treat real computation as atomic; e.g., a real number takes up one unit of memory. In practice, floating point implementations can compromise privacy \cite{mironov2012significance}, but techniques exist to ensure differential privacy in discrete settings \cite{canonne2020discrete,kairouz2021distributed}.}

In this work, we address the computationally efficiency of sampling the correlated noise. Specifically,
\emph{we provide a matrix factorization mechanism for differentially private continual counting attaining near-optimal error with only polylogarithmic memory overhead.}

\paragraph{Streaming setting:} Our algorithm receives the coordinates of the input $x$ and i.i.d.~noise $z$ one at a time and outputs the approximate partial sums $\mech(x)$ one at a time. In particular, we must output each partial sum $\br{\mech(x)}_k \approx \sum_{i \le k} x_i$ before receiving the next input $x_{k+1}$.

Ideally, the memory usage should not grow with the stream length $\nd$ at all. 
If we use the trivial factorization given by $B=I$ and $C=A$, then this corresponds to adding independent noise to each output. This would only require constant memory and constant time per output, but the noise scale $\MaxErr(B,C)=\sqrt{\nd}$ is far from optimal. 
Thus, the challenge is to simultaneously obtain near-optimal error and computational efficiency.

To overcome this challenge we must impose some structure on the factors $B$ and $C$ that allows efficient noise generation without significantly increasing the matrix factorization objective $\MaxErr(B,C)$.
The three structures that have been considered in the literature are (i) lower triangular matrices, (ii) sparse matrices, and (iii) Toeplitz matrices -- i.e., constant diagonals, $\forall i,j ~~ C_{i,j} = c_{i-j}$ for some vector $c$ -- and combinations of these structures.
We focus primarily on lower triangular Toeplitz matrices. 
The streaming setting naturally corresponds to lower triangular matrices, as it ensures that the $k$-th output is only a function of the first $k$ inputs.
Adding the Toeplitz constraint has a minimal impact on the matrix factorization objective \eqref{eq:mfobj}. Furthermore, Toeplitz matrices (and certain specializations we consider) have convenient mathematical and algorithmic properties. See \cref{fig:bounds} and \cref{sec:lttoe} for further discussion about lower triangular Toeplitz factorizations versus other factorizations. 

\subsection{Our Contributions}

Our main theoretical result is a lower triangular Toeplitz matrix factorization that is nearly optimal accompanied by an efficient streaming algorithm for generating the corresponding noise.

\begin{thm}[Main Result -- Informal version of \Cref{thm:main-f}]\label{thm:main-inf}
    For each integer $\nd \ge 1$ and error parameter $\mu \in (0,1)$
    , there exists a lower triangular Toeplitz matrix factorization $B, C \in \mathbb{R}^{\nd \times \nd}$ with the following properties.
    \begin{compactitem}
        \item \textbf{Validity:}
        $BC=A$, where $A$ is the $\nd \times \nd$ lower-triangular all-ones matrix.
        \item \textbf{Near-optimality:}
        $
            \MaxErr(B,C) \le \mathsf{OptLTToe}(\nd) + \mu, 
        $
        where 
        \[
        \mathsf{OptLTToe}(\nd) = 1+\sum_{k=1}^{\nd-1} \left( 2^{-2k} {2k \choose k} \right)^2 \le 1 + \frac{0.57722+\log(n)}{\pi}
        \] is the optimal value of $\MaxErr(B,C)$ over all lower triangular Toeplitz factorizations $A=BC$.
        \item \textbf{Efficiency:}
        There exists a streaming algorithm that at each step $k$ takes as input $z_k$ and outputs $(Bz)_k$ (or, equivalently, $(C^{-1}z)_k$) and runs in time and space $O(\log^2(\nd/\mu))$.
    \end{compactitem}
\end{thm}

In order to ensure that the error term $\mu$ in our result is $o(1)$, we have space and time complexity $d=\Theta(\log^2 \nd)$.
It is natural to wonder whether this computational complexity can be improved.
We show that it can be improved to $\widetilde{O}(\log \nd)$ at the expense of a weaker multiplicative near-optimality guarantee for the matrix factorization objective.\footnote{
In this problem, the multiplicative constants in the error have a larger practical impact on the mechanism's utility than an additive error.
For instance, the binary tree mechanism is suboptimal by a multiplicative factor of $\pi / \log 2 \approx 4.5$, which yields a much worse $\MaxErr$ than the factorization of \citet{fichtenberger2022constant} (cf. \cref{fig:bounds}, left).
}

\begin{thm}[Secondary Result -- Informal version of \Cref{thm:main2-f}]\label{thm:main2-inf}
    For each integer $\nd \ge 1$, there exists an integer $\nd'=O(\nd)$ and a matrix factorization $B, C^T \in \mathbb{R}^{\nd \times \nd'}$ with the following properties.
    \begin{compactitem}
        \item \textbf{Validity:}
        $BC=A$, where $A$ is the $\nd \times \nd$ lower-triangular all-ones matrix.
        \item \textbf{Near-optimality:}
        $
            \MaxErr(B,C) \le (1+o(1)) \cdot \mathsf{Opt}(\nd) , 
        $
        where $\mathsf{Opt}(\nd) = \frac{\log(n)}{\pi} \pm O(1)$ is the optimal value of $\MaxErr(B,C)$ over all factorizations.
        \item \textbf{Efficiency:}
        There exists a streaming algorithm that at each step $k$ takes as input some coordinates of $z$ (but never reads the same coordinate more than once) and outputs $(Bz)_k$ and runs in space (and amortized time per iteration) $\widetilde{O}(\log \nd)$.
    \end{compactitem}
\end{thm}

The optimal matrix factorization objective value over lower-triangular Toeplitz factorizations $\mathsf{OptLTToe}(\nd)$ is a small additive constant (specifically, $\le 0.365$) away from the optimal over all factorizations $\mathsf{Opt}(\nd)$ (cf. \cref{cor:toeplitz_gap}).
Thus, \Cref{thm:main-inf}'s approximation bound on the class of lower triangular Toeplitz matrices can be directly compared to \Cref{thm:main2-inf}'s approximation bound over all possible factorizations.

More generally, we can smoothly trade off between the matrix factorization objective and computational efficiency. That is, we can interpolate between \Cref{thm:main-inf,thm:main2-inf}; see \Cref{prop:instantiate-recursive} for a general statement.

\Cref{thm:main2-inf} does not generate a lower triangular Toeplitz factorization; in fact it does not even produce a square factorization. It also gives a weaker multiplicative near-optimality guarantee. We leave it as an interesting open problem whether it is possible to improve on $\log^2 \nd$ space complexity with a Toeplitz factorization or with $\MaxErr(B,C) \le \mathsf{Opt}(\nd)+O(1)$.


\begin{figure}[h!]
    \begin{center}
    \includegraphics[width=4in]{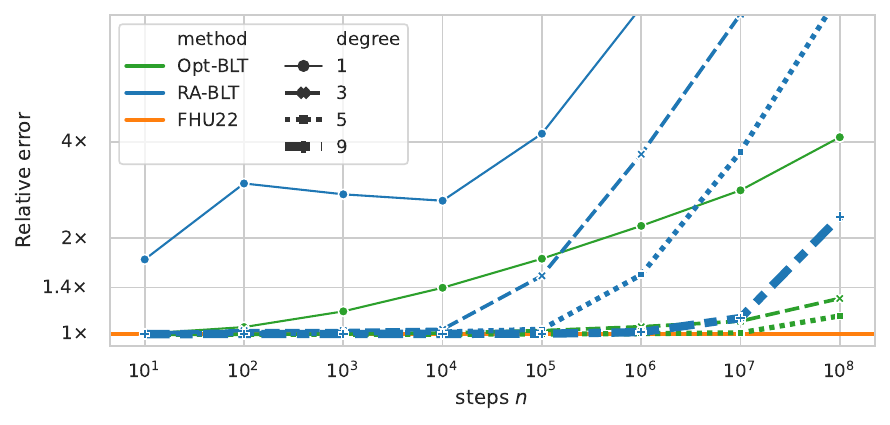}
    \caption{\label{fig:main_blt_results}
     Ratio of $\MaxErr(B,C)$ of our \RABLT and \OptBLT mechanisms for different numbers of steps $\nd$ and degrees $\nbuf$ (which corresponds directly to the space complexity) over that of the optimal Toeplitz mechanism of \citet{fichtenberger2022constant}.
     This illustrates that even with modest degree $d$, we obtain very good $\MaxErr(B,C)$ even for large numbers of steps $\nd$. For example, \OptBLT with $\nbuf=5$  is within $1\%$ of optimal for $n=10^7$ (we do not plot \OptBLT for $\nbuf=9$).
     }
     \end{center}
\end{figure}
\paragraph{Practical algorithms:}
Since our work is directly motivated by practical considerations, we also study the problem numerically.
Specifically, we show that the method behind \Cref{thm:main-inf} yields a factorization that is indistinguishable from optimal Toeplitz mechanism for practical purposes, but which has a highly efficient procedure for generating the noise. 
\Cref{fig:main_blt_results} shows the ratio of $\MaxErr(B,C)$ for our algorithm over the optimal $\mathsf{OptLTToe}(\nd)$.

To prove \cref{thm:main-inf} we provide closed-form parameters for the matrix factorization. This is already reasonably practical, but we further optimize the parameters numerically to obtain even better algorithms.
In order to optimize the matrix factorization we (i) choose an appropriate parameterization for the class of factorizations $B,C$ that we consider, (ii) give an efficiently computable expression for the objective $\MaxErr(B,C)$ in terms of this parameterization, and (iii) show that this expression is differentiable. This allows us to numerically optimize our factorization using gradient-based methods. While we do not prove that this optimization procedure converges, in practice it yields significantly better solutions than the closed-form parameters used to prove \cref{thm:main-inf}. 

To further illustrate the power of our approach, we show in \cref{sec:warmup} that even with constant memory, we can asymptotically attain $\MaxErr(B,C)=O(\nd^{1/6})$. This is already an improvement over $\MaxErr(I,A)=\sqrt{\nd}$ obtained by the trivial factorization.

\paragraph{Lower bounds:}
In \cref{sec:lttoe} we prove a lower bound that exactly characterizes $\mathsf{OptLTToe}(\nd)$. Specifically, we show that the lower triangular Toeplitz factorization of \citet{fichtenberger2022constant} is precisely optimal for this class.

Further, in \cref{sec:sdp_lower}, we develop \emph{numerical} lower bounds  on the objective of any matrix factorization for various classes of matrices for any fixed size $n$. In particular, we give lower bounds for arbitrary lower triangular matrices, arbitrary Toeplitz matrices, and Toeplitz matrices that correspond to our algorithm with a specific constant memory constraint.

\subsection{Our Algorithms \& Techniques}

The starting point for our main result is the lower triangular Toeplitz factorization of \citet{fichtenberger2022constant}, which we show is in fact the optimal lower triangular Toeplitz factorization in \Cref{sec:lttoe}.
This factorization is given by $B_{i,j}=C_{i,j}=f_{i-j}$ where the sequence $f_0, f_1, \cdots$ is the coefficients of a Taylor series:\footnote{We overload notation and use $x$ both for the input to the algorithm ($x \in \R^\nd$), and the indeterminate in polynomials and generating functions; the meaning should be clear from context.}
\begin{equation}
    \frac{1}{\sqrt{1-x}} = f_0 + f_1 x + f_2 x^2 + f_3 x^3 + \cdots. \label{eq:taylor1oversqrt}
\end{equation}
This sequence satisfies the recurrence $f_k=(1-1/2k) \cdot f_{k-1}$ for $k \ge 1$ with $f_0=1$. Our task in this case is to compute the correlated noise
$
(Bz)_k = \sum_{j=0}^{k} B_{k,j} \cdot z_j = \sum_{j=0}^{k} f_{k-j} \cdot z_j
$ 
in a streaming fashion for $k \in [\nd] \coloneqq \{0,1,\cdots,\nd-1\}$.\footnote{We zero-index sequences, vectors, and matrices throughout. See \cref{tab:notation} for a summary of symbols and notation.} Unfortunately, the factorization of \cref{eq:taylor1oversqrt} does not seem to admit an efficient sampling or noise generation procedure. 

\paragraph{\BLTs: Buffered Linear Toeplitz Matrices (\S\ref{sec:genfunc-framework}):}
Suppose the sequence instead satisfied a linear recurrence $f_k = q \cdot f_{k-1}$ for all $k \ge 1$.
Then
\begin{equation}
    (Bz)_k = \sum_{j=0}^{k} f_{k-j} \cdot z_j = f_0 \cdot z_k + \sum_{j=0}^{k-1} (q \cdot f_{k-j-1}) \cdot z_j = f_0 \cdot z_k + q \cdot (Bz)_{k-1}. \label{eq:singlerecurrencealg}
\end{equation}
This equation gives us an efficient algorithm: given the previous output $(Bz)_{k-1}$ and the current input $z_k$, we can compute the current output $(Bz)_k$. The memory requirement of this algorithm is simply to store the previous output in a single memory buffer. 

Next, suppose the sequence instead satisfied a linear recurrence of the form
\begin{equation}
    f_k = q_1 \cdot f_{k-1} + q_2 \cdot f_{k-2} + \cdots + q_d \cdot f_{k-d}. \label{eq:finiterecurrence-intro}
\end{equation}
As we explain next, this recurrence gives an algorithm where the memory requirement is to only store $\nbuf$ memory buffers. 

The recurrence $f_k = q f_{k-1}$ implies the closed form $f_k = q^k f_0$.
Similarly, the recurrence in \Cref{eq:finiterecurrence-intro} implies a closed form expression $f_k = u^T W^k v$ where $W \in \mathbb{R}^{d \times d}$ is a matrix and $u,v \in \mathbb{R}^d$ are vectors.
This closed form is what we use in \Cref{sec:algorithm} for our \Cref{alg:streamingmatrixpower}.
Specifically, we can extend \Cref{eq:singlerecurrencealg} to this matrix-power closed form
\begin{align*}
    (Bz)_k &= \sum_{j=0}^{k} f_{k-j} \cdot z_j 
    = \sum_{j=0}^{k} u^T W^{k-j} v \cdot z_j
    = u^T S_{k+1}
\end{align*}
for a suitable state vector $S_k \in \R^\nbuf$, stored in memory. Namely, we initialize $S_0=0$ and, at each iteration, our algorithm updates \[S_{k+1} = v \cdot z_k + W S_k\] and then outputs $(Bz)_k = u^T S_{k+1}$.
We refer to the entries of $S_k$ as the $\nbuf$ \emph{buffers} of our algorithm.\footnote{Recall that while we treat these as scalars here, in practical ML applications for example, each of these has size equal to the number of parameters of the model being trained, e.g. possibly $\mdim > 10^9$, so keeping $\nbuf$ to a small constant is critical.} The updates to the buffers on each step are an arbitrary linear function of the previous step's buffers ($WS_k$) and the current input ($v z_k$), and the output on each step is an arbitrary linear combination  of the buffers ($u^T S_{k+1}$). Hence, we term this class ``\textbf{Buffered Linear Toeplitz} matrices'' (\BLTs). We overload the acronym and write \BLTs as a shorthand encompassing matrices, factorizations, and mechanisms.

Unfortunately, the optimal factorization \cite{fichtenberger2022constant} does \emph{not} satisfy a recurrence like \cref{eq:finiterecurrence-intro} and cannot be expressed as a \BLT. Hence, our approach is to \emph{approximate} the optimal factorization using \BLTs.

\paragraph{Designing \BLTs via rational function approximation (\S\ref{sec:rationalfunctionapprox}):}
If we view the Toeplitz sequence $f_0, f_1, \dots$ as being defined by an ordinary generating function\footnote{We use the terms generating function, ordinary generating function, and Taylor series interchangeably; the generating function view of $f$ emphasizes the sequence being generated and requires (only) a formal power series, while the Taylor series view emphasizes $f$ is a real or complex function. As long as $f$ is analytic in a non-empty open neighborhood of zero, no ambiguity is introduced by these two views, see e.g. Thm. 2.8 (Transfer principle) of \citet{kauers11concrete}.}
 $f$ as in \Cref{eq:taylor1oversqrt}, it turns out that satisfying a linear recurrence as in \Cref{eq:finiterecurrence-intro} is equivalent to the function being rational with degree at most $\nbuf$, i.e., $f(x) = p(x)/q(x)$ for polynomials $p$ and $q$ of degree $\le \nbuf$.
(This equivalence is analogous to the fact that a real number is rational if and only if its decimal representation is repeating.)

This equivalence also suggests our first approach to designing \BLTs: we need a low-degree rational \emph{approximation} to the function $f(x) = {1}/{\sqrt{1-x}}$ from \cref{eq:taylor1oversqrt} that underlies the optimal factorization of \citet{fichtenberger2022constant}.

We appeal to known results in approximation theory. 
Specifically, it is known that the function $x\mapsto\sqrt{1-x}$ can be uniformly approximated on the unit complex disc $\{x \in \mathbb{C} : |x| \le 1\}$ with error $\eta>0$ by a rational function of degree $d=O(\log^2(1/\eta))$~\cite{Newman64,GopalT19}. That is, there exists a rational function $r$ of degree $\le d$, such that $|r(x)-\sqrt{1-x}|\le \eta$ for all $x\in \mathbb{C}$ with $|x|\le 1$.

It ``only'' remains to translate this approximation guarantee back to the matrix factorization objective.
Parseval's identity allows us to bound the difference between the sequences of Taylor coefficients in terms of an integral:
Suppose $f(x)=\sum_{k=0}^\infty f_k x^k$ and $\widetilde{f}(x)=\sum_{k=0}^\infty \widetilde{f}_k x^k$. Then
\begin{equation}
    \sum_{k=0}^\infty |f_k-\widetilde{f}_k|^2 = \frac{1}{2\pi} \int_{-\pi}^\pi \left|f(x(\theta))-\widetilde{f}(x(\theta))\right|^2 \, \mathrm{d}\theta \,~\text{ where }~ x(\theta) = \exp(\sqrt{-1} \theta). \label{eq:unweightedparseval}
\end{equation}
In our case, $f(x)={1}/{\sqrt{1-x}}$ corresponds to the optimal lower triangular Toeplitz factorization while $\widetilde{f}(x)$ is our rational approximation -- either $\widetilde{f}(x)=1/r(x)$ or $\widetilde{f}(x)=r(x)/(1-x)$, where $r(x) \approx \sqrt{1-x}$.
We are interested in a finite sum $\sum_{k=0}^{\nd-1} |f_k-\widetilde{f}_k|^2$, rather than the infinite sum in \Cref{eq:unweightedparseval} (which does not converge in our setting).
Thus we consider a weighted version of Parseval's identity where the integral goes around a circle in the complex plane centered at $0$ with radius $e^{-1/2n}$. 
Note that we require the approximation guarantee to hold on the complex plane, not just the real line.
Once we have this bound on $\sum_{k=0}^{\nd-1} |f_k-\widetilde{f}_k|^2$, the near optimality guarantee of \Cref{thm:main-inf} follows from the triangle inequality.

The time and space requirement of \Cref{thm:main-inf} is $O(\log^2 n)$. This dependence arises from the degree of the rational approximation.
A degree at least $\Omega(\log^2(1/\eta))$ is necessary for approximating $\sqrt{1-x}$ with error $\le\eta$ even for real values $x \in [-1,1]$ \cite{Newman64}.
Thus, unless we can exploit some slack in our analysis, it seems we require different techniques to bring the space down to $O(\log n)$.

The proof of \Cref{thm:main-inf} gives an explicit rational function approximation, which we can directly convert into a matrix factorization and feed into our algorithm. We term the \BLTs coming from this approach {\RABLT}s, with each choice of the degree $\nbuf$ leading to a different (and successively better) approximation to the optimal Toeplitz factorization. 

\paragraph{Designing \BLTs via direct optimization (\S\ref{sec:practical-optimization}):}
Recall that our goal is to minimize $\MaxErr(B,C)=\MaxErr(AC^{-1},C)$.
In order to ensure computational efficiency, we restrict the matrix $C$ to the class of \BLTs. While approximation theory lets us directly construct a near-optimal $C$ as outlined above, we can also approach this as an optimization problem and optimize the \BLT parameters that define $C$ numerically.

This optimization is far from straightforward. The class of \BLT matrices can be parameterized in multiple ways. Converting the parameters for $C$ into $\MaxErr(AC^{-1},C)$ is nontrivial to compute -- much less optimize -- when the size $\nd$ is large.
Nevertheless, the class of \BLTs is algebraically closed under multiplication and addition, and this structure combined with the connection to rational generating functions, provides powerful tools for reasoning about them.

We give a parameterization for the class of \BLTs that allows us to efficiently compute $\MaxErr$ in time practically independent of the size $\nd$, specifically \maxerrtime, and also to compute gradients.
Being able to compute gradients allows us to optimize $\MaxErr(AC^{-1},C)$ numerically.

The first challenge is that we need to be able to effectively parameterize both $C$ and $A C^{-1}$.
If $C$ is a \BLT, then, as discussed above, its Toeplitz coefficients are given by the Taylor series for a rational generating function $c(x) = p(x)/q(x)$ for polynomials $p$ and $q$. Further, these coefficients have a simple closed-form expression given by \cref{lem:rationaltoconstrec} (see also \cref{eq:closedcoefs}). 
Of course $1 / c(x) = q(x)/p(x)$ is also a rational function, and it in fact generates the Toeplitz coefficients of $C\inv$ (\cref{lem:mgfmult}), and so $C\inv$ is also a \BLT. The lower-triangular matrix of ones, $A$, is trivially a \BLT, and so \cref{lem:mgfmult} also implies $B = A C\inv$ is a \BLT.

In \cref{sec:practical-optimization}, \cref{lem:invogfs} shows that given a $\BLT$ $C$, we can (explicitly and in closed-form) derive the \BLT parameters of $C\inv$, and hence a closed form for its Toeplitz coefficients.\footnote{More precisely, \cref{lem:invogfs} uses a parameterization of $p$ and $q$ which leads to closed forms for both $C$ and $C\inv$.}
Using these closed-form expressions for the Toeplitz coefficients enables us to directly compute $\MaxErr(B, C)$ for $B, C \in \R^\dimdim$ in time \maxerrtime, see \cref{lem:closedsensitivity,lem:closederr}.

This is immediately useful, for example in \cref{fig:bounds} allowing us to plot the performance of our mechanisms for $\nd$ up to $10^8$ using only a few seconds of computer time.\footnote{We could have in fact scaled the $\MaxErr$ calculations for our mechanisms to arbitrary $\nd$; the bottleneck is in the computation of the exact $\MaxErr$ for \citep{fichtenberger2022constant}, which requires time $\calO(\nd)$.}

More importantly, however, \cref{sec:practical-optimization} shows computing $\MaxErr(B, C)$ is a differentiable function of the parameters of the \BLTs $B$ and $C$, and hence we can use a gradient-based optimization method such as L-BFGS to directly minimize $\MaxErr$ targeting a specific number of steps $\nd$, where $B$ and $C$ are $\BLTs$ with $\nbuf$ buffers defined by $2\nbuf$ parameters. We term these \OptBLT mechanisms, and they perform extremely well in practice. For example, for $\nd = 10^7$, a \OptBLTd{4} has $\MaxErr$ that is $1.032\!\times$ that of the optimal Toeplitz factorization \citep{fichtenberger2022constant}, and \OptBLTd{7} is  $1.001\!\times$ optimal; for smaller $\nd$ the results are even better, for example for $\nd = 10^4$, \OptBLTd{4} achieves $1.001\!\times$ optimal. \cref{fig:bounds} gives more complete results.

\paragraph{Generalizations of the binary tree mechanism (\S\ref{sec:generalized-binary-tree}):}
The starting point for \Cref{thm:main2-inf}  is the binary tree mechanism of \citet{Dwork-continual,CSS11-continual}.
The binary tree mechanism can be viewed as a recursive construction of a matrix factorization. 
A recursion of depth $\ell$ yields a matrix factorization of size $\nd=2^\ell$ and an algorithm running in time and space $O(\ell)$.
The binary tree mechanism does not produce a Toeplitz or square matrix factorization; the structure that it relies on for computational efficiency is sparsity.
The matrix factorization objective $\MaxErr(B,C)$ for this construction is $O(\log \nd)$ -- that is, it is within a constant factor of optimal. Specifically, the binary tree mechanism is asymptotically a factor of $\frac{\pi}{\log 2} \approx 4.5$ from optimal. This factor is significant in practice, as shown in \cref{fig:bounds}. 

We combine the binary tree mechanism's recursive approach with \Cref{thm:main-inf} to get the best of both worlds -- near-optimal constants and $\widetilde{O}(\log n)$ space. This proves \Cref{thm:main2-inf}.

We illustrate one step of the recursive construction using the following example for size $n=6$. We can decompose the $6 \times 6$ all-ones lower triangular matrix $A^{(6)}$ into a sum of expressions involving a $2 \times 2$ all-ones lower triangular matrix $A^{(2)}$ and a $3 \times 3$ all-ones lower triangular matrix $A^{(3)}$:
\begin{align}
A^{(6)} =
    \left(\begin{array}{cccccc}
    1 & 0 & 0 & 0 & 0 & 0\\
    1 & 1 & 0 & 0 & 0 & 0\\
    1 & 1 & 1 & 0 & 0 & 0\\
    1 & 1 & 1 & 1 & 0 & 0\\
    1 & 1 & 1 & 1 & 1 & 0\\
    1 & 1 & 1 & 1 & 1 & 1
    \end{array}\right)
    &=
    \left(\begin{array}{cccccc}
    1 & 0 & 0 & 0 & 0 & 0\\
    1 & 1 & 0 & 0 & 0 & 0\\
    0 & 0 & 1 & 0 & 0 & 0\\
    0 & 0 & 1 & 1 & 0 & 0\\
    0 & 0 & 0 & 0 & 1 & 0\\
    0 & 0 & 0 & 0 & 1 & 1
    \end{array}\right)
    + 
    \left(\begin{array}{cccccc}
    0 & 0 & 0 & 0 & 0 & 0\\
    0 & 0 & 0 & 0 & 0 & 0\\
    1 & 1 & 0 & 0 & 0 & 0\\
    1 & 1 & 0 & 0 & 0 & 0\\
    1 & 1 & 1 & 1 & 0 & 0\\
    1 & 1 & 1 & 1 & 0 & 0\\
    \end{array}\right) \label{eq:recursive6}\\
    &= 
    \left(\begin{array}{ccc}
    1 & 0 & 0 \\
    0 & 1 & 0 \\
    0 & 0 & 1
    \end{array}\right) 
    \otimes 
    \left(\begin{array}{cc}
    1 & 0 \\
    1 & 1
    \end{array}\right)
    +
    \left(\begin{array}{ccc}
    0 & 0 & 0 \\
    1 & 0 & 0 \\
    1 & 1 & 0
    \end{array}\right) 
    \otimes 
    \left(\begin{array}{cc}
    1 & 1 \\
    1 & 1
    \end{array}\right) \notag\\
    &= I \otimes A^{(2)} + \left( S^{(3)} A^{(3)} \right) \otimes \left( \mathbf{1} \mathbf{1}^T \right), \notag
\end{align}
where $\otimes$ denotes the Kronecker product,\footnote{A key property of the Kronecker product is that $(A \cdot B) \otimes (C \cdot D) = (A \otimes C) \cdot (B \otimes D)$.} $I$ is the identity matrix, $\mathbf{1}$ is the all-ones vector, and $S^{(3)} = \left(\begin{array}{ccc} 0 & 0 & 0 \\ 1 & 0 & 0 \\ 0 & 1 & 0 \end{array}\right)$ is a non-cyclic shift matrix.
\Cref{eq:recursive6} can be used to take factorizations of size $2$ and $3$ and combine them into a factorization of size $6$. 
Namely, if $A^{(2)}=B^{(2)}C^{(2)}$ and $A^{(3)}=B^{(3)}C^{(3)}$, then 
\begin{align}
    A^{(6)} &= I \otimes A^{(2)} + \left( S^{(3)} A^{(3)} \right) \otimes \left( \mathbf{1} \mathbf{1}^T \right) \tag{\cref{eq:recursive6}} \\
    &= (I \cdot I) \otimes (B^{(2)} \cdot C^{(2)}) + \left( S^{(3)} B^{(3)} \cdot C^{(3)} \right) \otimes \left( \mathbf{1} \cdot \mathbf{1}^T \right) \notag \\
    &= (I \otimes B^{(2)}) \cdot ( I \otimes C^{(2)}) + \left( S^{(3)} B^{(3)} \otimes \mathbf{1} \right) \cdot \left(C^{(3)} \otimes  \mathbf{1}^T \right) \notag \\
    &= \underbrace{\left( I \otimes B^{(2)} ~\mid~ S^{(3)} B^{(3)} \otimes \mathbf{1} \right)}_{=B^{(6)}} \cdot \underbrace{\left(\begin{array}{c} I \otimes C^{(2)} \\ C^{(3)} \otimes \mathbf{1}^T \end{array}\right)}_{=C^{(6)}}. 
\end{align}
The factors $B^{(6)}$ and $C^{(6)}$ are non-square matrices represented as block matrices.
The reason we move to non-square matrices is that \Cref{eq:recursive6} decomposes $A^{(6)}$ as a \emph{sum} of two matrix products and we must re-express this as a single matrix product.

Roughly speaking, the binary tree mechanism corresponds to starting with a factorization of size $2$ and applying a recursive step similar to the above $\ell-1$ times to obtain a factorization of size $2^\ell$. 
Rather than starting with a factorization of size $2$, we can start with a larger factorization of size $\nd_1$ given by \Cref{thm:main-inf} and then repeat the recursive step above $\ell-1$ times to obtain a factorization of size $\nd_1^\ell$.
Intuitively, by picking a larger factorization as the starting point we get closer to the optimal constant.
With careful analysis and the right choice of parameters, this yields \Cref{thm:main2-inf}.
 
This recursive construction attains excellent asymptotics, but, for practical parameter regimes, we find that the \BLT approach is practically indistinguishable from the optimal Toeplitz mechanism.

\subsection{An Empirical Comparison of Mechanisms}
In this section, we provide an empirical comparison of the primary DP mechanisms discussed, demonstrating their effectiveness in practical regimes.

\cref{fig:bounds} compares mechanisms (and lower bounds) in terms of $\MaxErr$ from $n=1$ to $n=10^6$ iterations. The sub-optimality of the binary tree mechanism is immediately clear. 
This plot also shows that little is lost by the restriction from general matrix mechanisms to Toeplitz mechanisms, where \citet{fichtenberger2022constant} provide the optimal (but inefficient) construction. Our \BLT mechanisms essentially match this performance, while requiring time and memory $\widetilde{\calO}(1)$ instead of $\calO(\nd)$.

\cref{fig:main_blt_results} provides a detailed comparison of our \BLT mechanisms and the optimal Toeplitz mechanism. Several important conclusions can immediately be drawn: (1) For both \RABLT and \OptBLT, increasing the degree (number of allowed memory buffers) increases performance. (2)  A larger number of steps $\nd$ requires a higher number of buffers $\nbuf$ for both of our approaches; this is expected and necessary, as shown by our theory, see \cref{rem:growth}. (3) A key point to emphasize is that each blue line for \RABLT corresponds to a single mechanism (a fixed rational approximation); for \OptBLT, we compute an optimized matrix factorization of the given degree for each different $\nd$. This specialization of the mechanism to the specific anticipated number of steps $\nd$ is critical to the advantage enjoyed by this approach. For example, we see \OptBLT with only 5 buffers outperforms \RABLT with 9 buffers across the full range of $\nd$.

It is of course possible to run \OptBLT mechanisms for a different number of steps than the optimization targeted. \cref{fig:bltcompare} (Left column) explores this. We construct three fixed \OptBLT factorizations, optimized for $\nd \in \{100, 1000, 10000\}$, and compare their performance (in terms of $\MaxErr$ relative to the optimal Toeplitz mechanism) across a range of steps, from $10^1$ to $10^5$. As expected, the mechanisms work best for the $\nd$s for which they were optimized; however, the excess error is highly asymetric; a mechanism optimized for $\nd^*$ will generally perform well for $\nd < \nd^*$ steps, but can quickly perform very badly when $\nd > \nd^*$. This is expected when one considers that \OptBLT should intuitively be ensuring a good approximation of the optimal Toeplitz coefficients $r_0, \dots, r_{\nd^*-1}$, but the approximation of the optimal Toeplitz coefficients for larger $\nd$ can become arbitrarily bad.  We see this in \cref{fig:bltcompare} (Middle column), where we compare the Toeplitz coefficients defining $C$ and $B$ to the optimal coefficients corresponding to the generating function $1/ \sqrt{1 - x}$. To emphasize this issue, we optimize for $\nd^* = 100$, and consider degree $\nbuf=2$. The \OptBLT factorization provides a better approximation to the optimal coefficients for $\nd$ up to 100 compared to \RABLTd{2} (which does not depend on $\nd$), and a generally worse approximation beyond that.  \cref{fig:bltcompare} (Right column) shows that while both \OptBLT and \RABLT correspond to ``reasonable'' approximations to $\sqrt{1 -x }$ (top), they distribute their errors in the approximation very differently (bottom).

\begin{figure}[h!]

    \begin{subfigure}[b]{0.49\linewidth}
    \includegraphics[width=\textwidth]{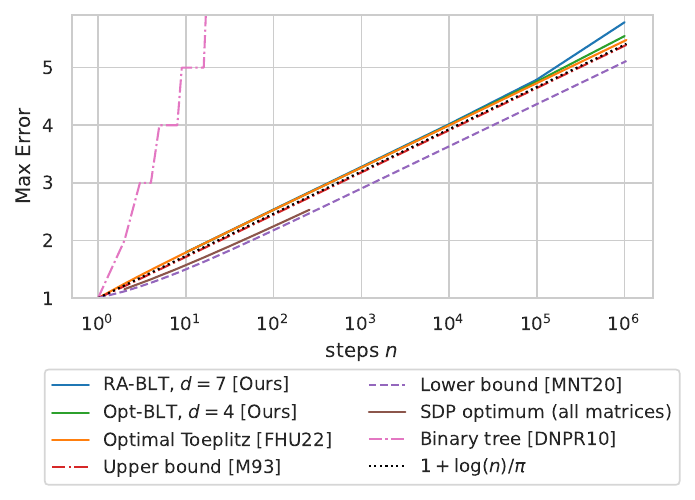}
    \end{subfigure}
    \hfill
    \begin{subfigure}[b]{0.49\linewidth}
    \includegraphics[width=\textwidth]{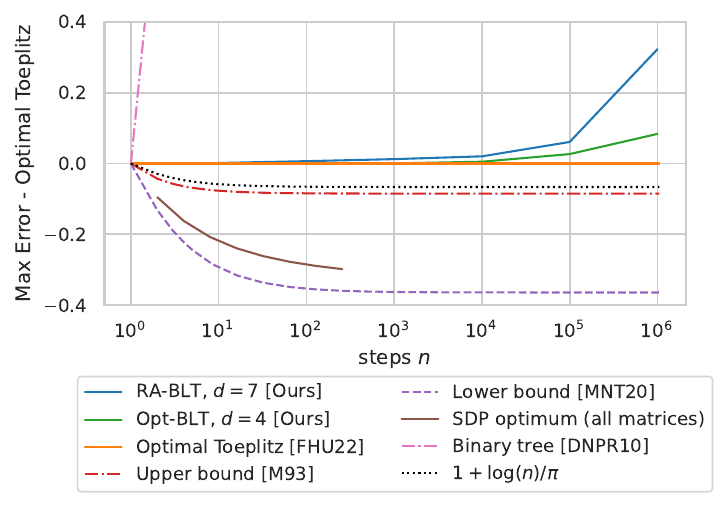}
    \end{subfigure}
    
    \caption{\label{fig:bounds}
        Comparison of known upper and lower bounds for factorizations $A=BC$ of the all-ones lower triangular matrix $A_{i,j}=\mathbb{I}[i\ge j]$. Note that this includes non-Toeplitz factorizations. This illustrates that there is a small gap between lower triangular Toeplitz factorizations and general factorizations; furthermore this gap is asymptotically constant.
        \textbf{Left:} Vertical axis is $\MaxErr(B,C)$.
        \textbf{Right:} Vertical axis is $\MaxErr(B,C)-\mathsf{OptLTToe}(\nd)$.
    }
\end{figure}

\begin{figure}[h!]
    \begin{center}
    \includegraphics[width=6in]{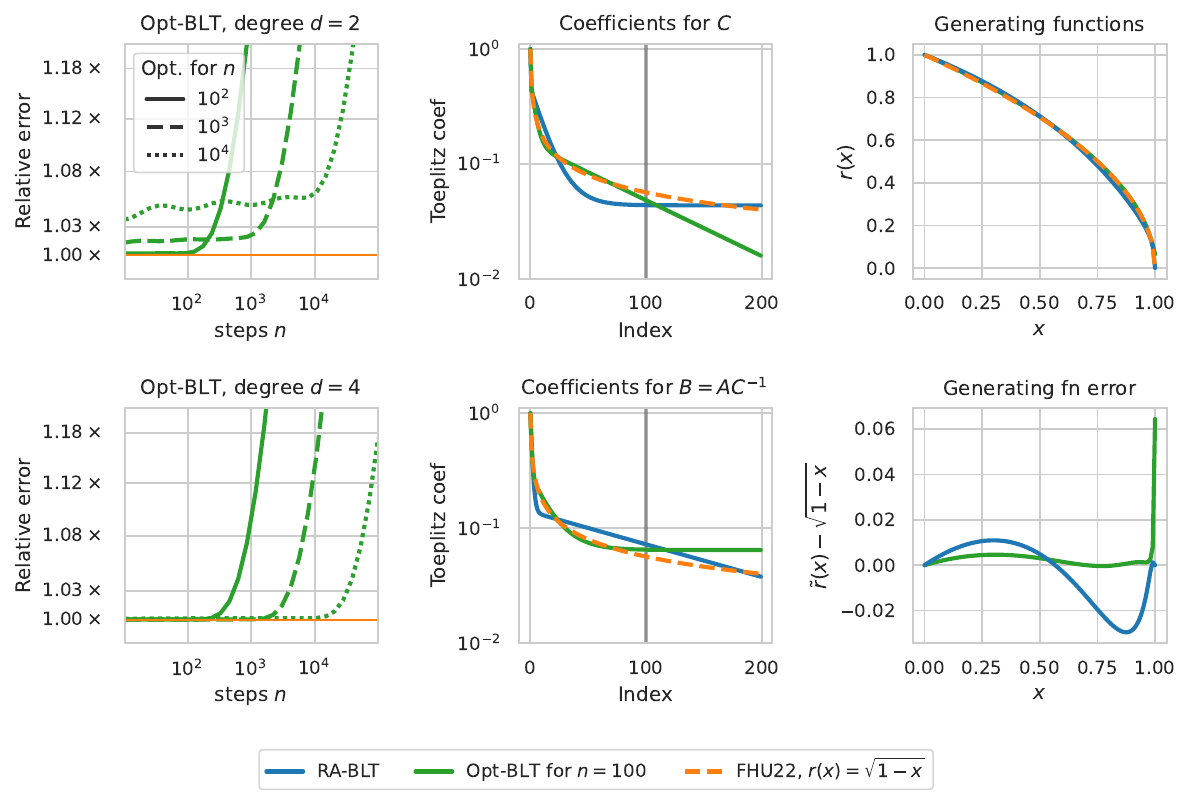}
    \caption{(Left column) Comparison of three fixed \OptBLT mechanisms across a range of $\nd$ (extending beyond the optimization targets).  (Center) Comparing the first 200 Toeplitz coefficients defining the $C$ and $B$ matrices for \OptBLT and \RABLT for degree $\nbuf=2$, with the \OptBLT mechanism optimized for $\nd^* = 100$.  (Right) Differences in the generating functions for the \OptBLT and \RABLT factorizations of the middle column.}\label{fig:bltcompare}
    \end{center}
\end{figure}

\clearpage
\section{Background}

We start by formally setting up the problem where each stream entry is a vector of dimension $\mdim$ (as opposed to a scalar in the introduction).
Fix a stream length $\nd \in \mathbb{N}$ and let 
\begin{equation}
    A = A^{(\nd)} \coloneqq \left( \begin{array}{ccccc} 1 & 0 & 0 & \cdots & 0 \\ 1 & 1 & 0 & \cdots & 0 \\ 1 & 1 & 1 & \cdots & 0 \\ \vdots & \vdots & \vdots & \ddots & \vdots \\ 1 & 1 & 1 & \cdots & 1 \end{array} \right) \in \{0,1\}^{\nd \times \nd} \label{eq:allonesM}
\end{equation} 
be the all-ones lower triangular matrix given by $A^{(\nd)}_{i,j} = 1 \iff i \ge j$ and $A_{i,j} = 0 \iff i < j$ for all $i,j \in [\nd]$.
Note that $A$ is invertible and its inverse is a lower triangular bi-diagonal matrix:
\begin{equation}
    A^{-1} = \left( \begin{array}{ccccc} 1 & 0 & 0 & \cdots & 0 \\ -1 & 1 & 0 & \cdots & 0 \\ 0 & -1 & 1 & \cdots & 0 \\ \vdots & \ddots & \ddots & \ddots & \vdots \\ 0 & \cdots & 0 & -1 & 1 \end{array} \right) \in \{-1,0,1\}^{\nd \times \nd}. \label{eq:allonesMinv}
\end{equation}

Our goal is to compute a matrix factorization $A=BC$ where $B,C^T \in \mathbb{R}^{\nd \times \mfid}$ that minimizes 
\begin{equation}
     \MaxErr(B,C) \coloneqq \|B\|_{2 \to \infty} \cdot \|C\|_{1 \to 2} \coloneqq \sqrt{ \max_{i \in [\nd]} \sum_{j \in [\mfid]} B_{i,j}^2 } \cdot \sqrt{ \max_{j \in [\nd]} \sum_{i \in [\mfid]} C_{i,j}^2 }. \label{eq:obj}
 \end{equation}

Simultaneously we want computational efficiency in the following sense. 
We need to generate samples from $BZ \in \mathbb{R}^{\nd \times \mdim}$ where $Z \in \mathbb{R}^{\mfid \times \mdim}$ is a matrix with independent standard Gaussian entries. 
We want to do this in a streaming setting where we output one row at a time and the memory is limited. Ideally the memory should be linear in $\mdim$ and constant or logarithmic in $\nd$.

We do not restrict $B$ and $C$ to be square matrices. 
Non-square factorizations may be advantageous from a computational perspective.
However, computational efficiency aside, we can assume without loss of generality that they are square (i.e., $\mfid=\nd$) by taking a singular value decomposition (SVD) of the factors and discarding the rows/columns that do not correspond to a non-zero singular value.

It is also natural to restrict $B$ and $C$ to be lower triangular matrices, like $A$.
Lower triangular structure implies each intermediate output $(Cx)_k$ is only a function of the inputs seen so far $x_0,x_1,\cdots,x_k$ and not future inputs $x_{k+1},\cdots,x_{n-1}$.
This is particularly valuable for the privacy analysis in the case where the input stream $x_1, \cdots, x_k$ is chosen adaptively, such as in machine learning applications. In the adaptive setting, each input $x_k$ may depend on the previous intermediate outputs $(Cx)_0,(Cx)_1, \cdots, (Cx)_{k-1}$. Lower triangular structure prevents a cyclic dependency.
Fortunately, we can assume the factors are lower triangular without loss of generality, see for example Proposition 2.2 of \citet{denisov22matfact}.

\subsection{Prior Work on Matrix Factorizations for Continual Counting}

There has been a \emph{lot} of work on factorizing the all-ones lower triangular matrix. 
However, most of that work optimizes the matrix factorization objective $\MaxErr(B,C)$ with little regard for the computational efficiency of sampling the correlated noise $Bz$ from i.i.d. seed noise $z \gets \mathcal{N}(0,\sigma^2I)$.

The binary tree mechanism \cite{Dwork-continual,CSS11-continual} implicitly constructs a factorization which optimizes $\MaxErr(B,C)$ up to constant factors and is efficiently computable. Prior to our work, this is the only known approach offering both efficient computation and some form of near-optimality.
The binary tree factorization can be expressed recursively as $B^{(1)}=C^{(1)}=(1) \in \{0, 1\}^{1 \times 1}$ and
\begin{equation}
    B^{(2\nd)} = \left(\begin{array}{ccc}
        B^{(\nd)} & 0 & 0 \\
        0 & B^{(\nd)} & \vec{1}
    \end{array} \right) \in \{0,1\}^{2\nd \times (4\nd-1)}
    ~~ \text{ and } ~~
    C^{(2\nd)} = \left(\begin{array}{cc}
        C^{(\nd)} & 0 \\
        0 & C^{(\nd)} \\
        \vec{1}^T & 0
    \end{array}\right) \in \{0,1\}^{(4\nd-1) \times 2\nd}. \label{eq:bintreeprior}
\end{equation}
Note that this only gives a factorization of $A^{(\nd)}$ when $\nd$ is a power of $2$. By discarding rows/columns we can extend this to any $\nd$. The objective value is given by \begin{equation}\MaxErr(B,C) = \|B^{(2^\ell)}\|_{2 \to \infty}^2 = \|C^{(2^\ell)}\|_{1 \to 2}^2 = \ell+1 = \left\lceil \frac{\log \nd}{\log 2} \right\rceil + 1.\label{eq:bintree}\end{equation}
The recursive formulation naturally leads to an efficient sampling algorithm for the correlated noise $Bz$ using only $O(\ell)$ space.
Subsequent work has attempted to improve the constants \cite{qardaji2013understanding,honaker2015efficient,andersson2024smooth}.
\citet{andersson2024smooth} improve the leading constant by a factor of $2$ compared to the standard bound in \cref{eq:bintree}. Their approach builds on the observation that the \emph{average} squared norm of a row of $B$ or a column of $C$ in the binary tree construction (as in \cref{eq:bintreeprior}) is roughly half the \emph{maximum} squared norm. Their construction starts with a binary tree factorization of size $O(n)$ (i.e., a constant factor larger than required) and discards large rows/columns of $B$/$C$, which leaves only rows/columns of average length. Their construction preserves sparsity and efficiency; however, the leading constant is still a constant factor far from optimal.

Observe that the binary tree factorization produces sparse matrices. Specifically, each row of $B^{(n)}$ and column of $C^{(n)}$ only has $O(\log n)$ nonzero entries.
Intuitively, this sparsity is what enables efficient computation; that is, $(Bz)_k$ only depends on logarithmically many elements of $z$.
However, our experience suggests that sparsity requires non-square matrices and cannot produce near-optimal factorizations like \BLTs. 

\citet{fichtenberger2022constant} give the following elegant explicit lower triangular Toeplitz factorization. 
Let $f_0 = 1$ and, for $k \in \mathbb{Z}_{\ge 1}$, let $f_k = f_{k-1} \cdot (1-1/2k)$ and $f_{-k}=0$. Equivalently, $f_k = 4^{-k} {2k \choose k} \le \frac{1}{\sqrt{\pi k}}$ for $k \ge 1$ \cite{CentralBinomialCoefficient}.
Define $B,C \in \mathbb{R}^{\nd \times \nd}$ by $B_{i,j} = C_{i,j} = f_{i-j}$ for all $i,j \in [\nd]$, which we denote $B = C =\LTT((f_i)_{i=0}^{n-1})$ or simply $\LTT(f, n)$.
Then, we have that $A=BC$ and\footnote{Unfortunately there appears to be an off-by-one error in the proof of the bound given by Theorem 1 of \citet{fichtenberger2022constant}; this has since been corrected.}
\begin{equation}
    \MaxErr(B,C) = \|B\|_{2 \to \infty}^2 = \|C\|_{1 \to 2}^2 = \sum_{k=0}^{\nd-1} f_k^2 \le 1+\frac{\log(\nd)+\gamma}{\pi},\label{eq:jalaj}
\end{equation}
where $\gamma \le 0.57722$ is the Euler-Mascheroni constant \cite{EulerMascheroniConstant}.
This factorization turns out to be optimal among the class of lower triangular Toeplitz factorizations; we prove this in \cref{lem:lbstruct}.

However, unlike the binary tree mechanism, these Toeplitz matrices are dense and we do not know how to efficiently generate noise according to this factorization in the streaming setting.

The $\gamma_2$ factorization norm is in fact a norm, i.e., it satisfies the triangle inequality \cite{tomczak1989banach}.
Furthermore, there is a dual characterization of the $\gamma_2$ factorization norm \cite[Theorem 9]{lee2008direct}. 
\begin{align}
    \gamma_2(A) &\coloneqq \inf\{ \MaxErr(B,C) : B,C \in \mathbb{R}^{\nd \times \nd}, A=BC \} \notag\\
    &= \sup\left\{ \| P^{1/2} \cdot A \cdot Q^{1/2} \|_{\text{trace}} : 
    \begin{aligned}
    P,Q \text{ non-negative diagonal matrices} \\
    \text{ with } \mathsf{trace}(P)=\mathsf{trace}(Q)=1
    \end{aligned}
    \right\},
    \label{eq:dual-characterization}
\end{align}
where $\| M \|_{\text{trace}} = \mathsf{trace}(\sqrt{M^TM})$ is the sum of the absolute singular values.

Using the triangle inequality, Mathias \cite[Corollary 3.5]{mathias1993hadamard} gives a non-constructive upper bound 
\begin{equation}
    \gamma_2(A) \le \frac12 + \frac{1}{2\nd} \sum_{j=1}^\nd \frac{1}{\sin\left(\pi \frac{2j-1}{2\nd}\right)} \le 1 + \frac{\log (\nd)}{\pi}.\label{eq:mathias_ub}
\end{equation}
On the other hand, using the dual characterization (\cref{eq:dual-characterization} with $P=Q=\tfrac1n I$), \citet{matouvsek2020factorization} give a lower bound\footnote{Proposition 4.1 of \citet{matouvsek2020factorization} simply states an $\Omega(\log \nd)$ lower bound, but these sharper expressions can easily be extracted from the proof. \citet{mathias1993hadamard} also gives a lower bound which differs from the upper bound in \cref{eq:mathias_ub} by less than an additive $\frac12$, but this seems slightly weaker than Equation \ref{eq:sasho}.}
\begin{equation}
    \gamma_2(A) \ge \frac{1}{2\nd} \sum_{j=1}^\nd \frac{1}{\sin\left(\pi \frac{2j-1}{4\nd+2}\right)}  \ge \frac{4n+2}{2\pi\nd} \sum_{j=1}^\nd \frac{1}{2j-1} \ge  \frac{\log(2\nd+1)}{\pi}.\label{eq:sasho}
\end{equation}
The upper bound in \cref{eq:jalaj} or \cref{eq:mathias_ub} and the lower bound in \cref{eq:sasho} match up to a small \emph{additive} constant. Thus the optimal value of the matrix factorization objective is $\gamma_2(A) = \frac{\log \nd}{\pi} \pm \mathsf{constant}$.
Numerically, the maximum gap between the Toeplitz upper bound in \cref{eq:jalaj} and the lower bound in \cref{eq:sasho} is less than $0.365$.

Figure \ref{fig:bounds} shows how these upper and lower bounds compare.
The binary tree mechanism has the advantage of computational efficiency, but as we can see, it is far from optimal for the objective.
The leading term for the binary tree is $\frac{\log \nd}{\log 2}$, while the optimal leading term is $\frac{\log \nd}{\pi}$. Thus the binary tree is asymptotically suboptimal by a multiplicative factor of $\pi/\log 2 \approx 4.5$.

\subsection{Lower Triangular Toeplitz Factorizations versus General Factorizations} \label{sec:lttoe}

Our work mainly focuses on lower triangular Toeplitz factorizations. In particular, \Cref{thm:main-inf} gives such a factorization and proves near-optimality with respect to this class. (Although we consider non-square factorizations in \Cref{sec:generalized-binary-tree}, where we prove \Cref{thm:main2-inf}.) 
This structure is essential to our algorithms in \Cref{sec:genfunc-framework,sec:rationalfunctionapprox,sec:practical-optimization}.
However, it is natural to wonder how much we lose in terms of the objective $\MaxErr(B,C)$ by restricting $B$ and $C$ to be lower triangular Toeplitz matrices.

In this subsection we discuss the gap between lower triangular Toeplitz factorizations and general factorizations.
We do not have an explicit construction for the optimal general matrix factorization or even a formula for optimal value $\gamma_2(A)$. 
For lower triangular Toeplitz matrix factorizations the best factorization we have is that of \citet{fichtenberger2022constant}.
We show that this is in fact optimal among the class of lower triangular Toeplitz factorizations.
This is -- to the best of our knowledge -- a novel result which may be of independent interest.

We begin with a lemma stating some basic properties of the factorization of \citet{fichtenberger2022constant} and we give a proof for completeness.

\begin{lem}\label{lem:fhu-properties}
    Define a sequence $f_0, f_1, \cdots$ by $f_0=1$ and $f_k = (1-1/2k)f_{k-1}$ for all $k \ge 1$.
    Then, for all $k \ge 1$, we have 
    \begin{equation}
        f_k = 4^{-k} {2k \choose k} \in \left[ \frac{1}{\sqrt{\pi(k+1)}}, \frac{1}{\sqrt{\pi k}}\right]. \label{eq:lem:fhu-properties1}
    \end{equation}
    For all integers $n \ge 0$, we have 
    \begin{equation}
        \sum_{k=0}^{n} f_k f_{n-k} = 1 \label{eq:lem:fhu-properties2}
    \end{equation}
    and
    \begin{equation}
        \frac{\gamma+\log(n) - 1}{\pi} \le \sum_{k=1}^\infty f_k^2 \le \frac{\gamma + \log(n)}{\pi}, \label{eq:lem:fhu-properties3}
    \end{equation}
    where $0.57721 \le \gamma \le 0.57722$ is the Euler-Mascheroni constant.
    Furthermore, for all $x\in\mathbb{C}$ with $|x|<1$, we have
    \begin{equation}
        \sum_{k=0}^\infty f_k x^k = \frac{1}{\sqrt{1-x}}. \label{eq:lem:fhu-properties4}
    \end{equation}
\end{lem}
\begin{proof}
The first part of \cref{eq:lem:fhu-properties1} can be shown by induction:
For $k \ge 1$,
\[4^{-k} {2k \choose k} = \frac14 \cdot 4^{-(k-1)} {2k-2 \choose k-1} \frac{2k (2k-1)}{k^2} = \frac{2k(2k-1)}{4k^2} f_{k-1} = \left( 1 - \frac{1}{2k} \right) f_{k-1} = f_k.\]
Next consider the derivatives of $f(x) \coloneqq \frac{1}{\sqrt{1-x}}$: For $k \ge 1$, we have \begin{equation}f^{(k)}(x) = (1-x)^{-k-1/2} \prod_{\ell=0}^{k-1} \left(\ell+\frac12\right).\end{equation}
Thus \[\frac{f^{(k)}(0)}{k!} = \prod_{\ell=0}^{k-1} \frac{1}{\ell+1}\left(\ell+\frac12\right) = \prod_{\ell=0}^{k-1} \frac{(2\ell+1)}{2(\ell+1)} \frac{(2\ell+2)}{2(\ell+1)} = \frac{(2k)!}{(k!)^2 2^{2k}} = 4^{-k} {2k \choose k} = f_k.\]
By Taylor's theorem, for all $x \in \mathbb{C}$ with $|x|<1$, \[\frac{1}{\sqrt{1-x}} = f(x) = \sum_{k=0}^\infty \frac{f^{(k)}(0)}{k!} x^k = \sum_{k=0}^\infty f_k x^k, \] as required to prove \cref{eq:lem:fhu-properties4}.
Furthermore, for all $x \in \mathbb{C}$ with $|x|<1$, we have \[\sum_{n=0}^\infty x^n = \frac{1}{1-x} = f(x)^2 = \left( \sum_{k=0}^\infty f_k x^k \right)^2 = \sum_{n=0}^\infty x^n \sum_{k=0}^n f_k f_{n-k}.\]
Matching coefficients proves \cref{eq:lem:fhu-properties2}

To prove the second part of \cref{eq:lem:fhu-properties1}, we use standard bounds on the central binomial coefficient \cite{CentralBinomialCoefficient,SpeyerMO}: For all integers $k \ge 1$, \begin{equation} \frac{4^k}{\sqrt{\pi(k+1)}} \le {2k \choose k} \le \frac{4^k}{\sqrt{\pi k}},\end{equation} whence $\frac{1}{\sqrt{\pi (k+1)}} \le f_k \le \frac{1}{\sqrt{\pi k}}$ . 
Next we use standard bounds on the harmonic numbers to prove \cref{eq:lem:fhu-properties3}: For all integers $n \ge 1$,
\begin{equation} \log(n) + \gamma + \frac{1}{2(n+1)} \le \sum_{k=1}^n \frac{1}{k} \le \log(n) + \gamma + \frac{1}{2n}\end{equation}
where $0.57721 \le \gamma \le 0.57722$ is the Euler-Mascheroni constant \cite{EulerMascheroniConstant}.
It follows that, for all $n \ge 2$,
\begin{equation} \sum_{k=1}^{n-1} f_k^2 \le \sum_{k=1}^{n-1} \frac{1}{\pi k} \le  \frac{\gamma + \log(n-1) + 1/2(n-1)}{\pi} \le \frac{\gamma + \log(n)}{\pi}.\end{equation}
Similarly, for all $n \ge 2$
\begin{equation}
    \sum_{k=1}^{n-1} f_k^2 \ge \sum_{k=1}^{n-1} \frac{1}{\pi(k+1)} =  \sum_{k=2}^n \frac{1}{\pi k} \ge \frac{\gamma + \log(n) + 1/(2n+2)-1}{\pi}.
\end{equation}
\end{proof}

Most importantly, we show that the factorization of \citet{fichtenberger2022constant} is optimal among the class of lower triangular Toeplitz factorizations:

\begin{prop}[Optimal lower triangular Toeplitz factorization]\label{lem:opt-ltt}
For any integer $n \ge 1$, consider the optimization problem
\begin{equation}
\begin{array}{rrclcl}
\displaystyle \min_{\bfb, \bfc \, \in \, \mathbb{R}^n} & \multicolumn{3}{l}{\ltwo{\bfb}\ltwo{\bfc}}\\
\textrm{s.t.} & \forall k  <n ~~ \sum_{i=0}^k b_ic_{k-i}=1.\\
\end{array}
\label{eq:optimization}
\end{equation}
The minimum is achieved at $\bfb=\bfc=(f_k)_{k=0}^{n-1}$, where $f_0, f_1, \cdots$ is given by $f_0 = 1$ and $f_{k}=\br{1-\frac{1}{2k}}f_{k-1}$ for all $k \ge 1$, which are the coefficients derived by \citet{fichtenberger2022constant}.
\label{lem:lbstruct}
\end{prop}
\begin{proof}
The proof proceeds in 3 steps: we first reformulate the optimization as a quadratically constrained quadratic program, then derive a natural Lagrangian relaxation of it, and finally construct a primal dual solution to the Lagrangian relaxation that is globally optimal with the primal solution $\bfb=\bfc=(f_k)_{k=0}^{n-1}$.

We use the following properties of $f_k$'s established by \citet{fichtenberger2022constant} and also shown in \cref{lem:fhu-properties}:
\begin{itemize}
    \item $f_k > 0 \quad \forall k \in \{0, \ldots, n-1\}$.
    \item Let $F \in \mathbb{R}^{n \times n}$ denote the lower triangular Toeplitz matrix whose first column is $f_0, f_1, \ldots, f_{n-1}$. Then, we have $F^2=A$.
    \item $\sum_{k=0}^{\infty}f_k x_k = \frac{1}{\sqrt{1-x}} \quad \forall x \in [0, 1]$.
\end{itemize}

\paragraph{Reformulation as a quadratically constrained quadratic program:}
We begin by observing that the optimization problem can be rewritten as a quadratic optimization as follows:
\begin{align*}
\ltwo{\bfb}\ltwo{\bfc} = \inf \left\{\frac\nu2 \ltwo{b}^2 + \frac{1}{2\nu} \ltwo{c}^2 : \nu>0 \right\} = \min \left\{\frac12 \ltwo{\nu b}^2 + \frac12 \ltwo{\frac{c}{\nu}}^2 : \nu \in \mathbb{R}\setminus\{0\} \right\}. 
\end{align*}
Furthermore, for any feasible $(b, c)$ that satisfy the constraints of \cref{eq:optimization} and any $\nu \ne 0$, the pair $(\nu b, \frac{c}{\nu})$ also satisfies the constraints. If $S$ denotes the feasible set of \cref{eq:optimization}, we have 
\[\min_{\bfb,\bfc \in S} \ltwo{\bfb}\ltwo{\bfc} = \min_{\bfb,\bfc \in S}\min_{\nu \in \mathbb{R}\setminus\{0\}} \frac12 \ltwo{\nu b}^2 + \frac12\ltwo{\frac{c}{\nu}}^2 = \min_{\bfb, \bfc \in S} \frac12 \ltwo{\bfb}^2 + \frac12 \ltwo{\bfc}^2. \]

Thus, we can simply write the overall optimization problem as 
\begin{equation}
\begin{array}{rrclcl}
\displaystyle \min_{\bfb,\bfc\in\mathbb{R}^n} & \multicolumn{3}{l}{\frac{1}{2}{\br{\ltwo{\bfb}^2 + \ltwo{\bfc}^2}}}\\
\textrm{s.t.} & \forall k  <n~~ \sum\limits_{i=0}^k b_ic_{k-i}=1.\\
\end{array}
\label{eq:optimization_alt}
\end{equation}

\paragraph{Lagrangian relaxation:}
We write the Lagrangian relaxation of this optimization problem introducing dual variables $\lambda_0, \cdots, \lambda_{n-1} \in \mathbb{R}$ corresponding to the constraints:
\begin{align}
&L\br{b,c,\lambda} = \frac{1}{2}\br{\ltwo{\bfb}^2 + \ltwo{\bfc}^2}    - \sum_{k=0}^{n-1} \lambda_k \br{\sum_{i=0}^k b_ic_{k-i}-1} \\
&=\frac{1}{2}\br{\ltwo{\bfb}^2 + \ltwo{\bfc}^2} - \tran{\bfb}\Gamma\br{\lambda}\bfc + \tran{\One}\lambda \notag\\
& = \frac{1}{2}\tran{\begin{pmatrix}b \\ c\end{pmatrix}} \begin{pmatrix} I & - \Gamma\br{\lambda} \\ -\Gamma\br{\lambda} & I\end{pmatrix}\begin{pmatrix}b \\ c\end{pmatrix} + \tran{\One}\lambda , \notag
\end{align}
where $\Gamma\br{\lambda}$ is the symmetric Hankel matrix
\[
\Gamma\br{\lambda} \coloneqq \begin{pmatrix}
\lambda_0 & \lambda_1 & \lambda_2 & \ldots & \lambda_{n-1}\\
\lambda_1 & \lambda_2 & \lambda_3 & \ldots & 0 \\
\lambda_2 & \lambda_3 & \lambda_4 & \ldots & 0\\
\vdots & \vdots & \vdots & \vdots & \vdots \\
\lambda_{n-1} & 0 & 0 & \ldots & 0
\end{pmatrix}.\]
The constrained optimization problem in \cref{eq:optimization_alt} is \emph{equivalent} to the unconstrained min-max problem \begin{equation}\min_{b,c\in\mathbb{R}^n} \max_{\lambda\in\mathbb{R}^n} L(b,c,\lambda).\end{equation}
Letting $\vec f = (f_0, \ldots, f_{n-1}) \in \mathbb{R}^n$, we will set $b=c=\vec f$ and exhibit a setting $\lambda=\lambda^*$ such that $\nabla_{b,c,\lambda}L(b,c,\lambda) = 0$ and $L(b,c,\lambda)$ is jointly convex in $b,c$ and concave in $\lambda$. It follows that $b=c= \vec f$ is the optimal solution to the original problem, as required.
Since $b=c=\vec f$ is feasible, $\nabla_\lambda L(b,c,\lambda) = \nabla_\lambda L(\vec f, \vec f,\lambda) = 0$.
The Lagrangian is linear in $\lambda$ and hence trivially concave.
The Lagrangian is convex in $\br{\bfb, \bfc}$ if and only if 
\[\begin{pmatrix} I & - \Gamma\br{\lambda} \\ -\Gamma\br{\lambda} & I\end{pmatrix} \succeq 0 \]
or equivalently (by Schur complements) if $\norm{\Gamma\br{\lambda}}_\star\leq 1$ where $\norm{Q}_\star \coloneqq \sup \{ \|Qx\|_2 : \|x\|_2 \le 1 \}$ denotes the operator norm. 

\paragraph{Construction of primal-dual optimal solution:}
Now it only remains to compute $\lambda=\lambda^*$ such that $\nabla_{b,c} L(b,c,\lambda) = 0$ for $b=c=\vec f$ and $\norm{\Gamma\br{\lambda}^2}_\star\leq 1$.

For all $k \in [n]$, we have $\frac{\partial}{\partial b_k} L(b,c,\lambda) = b_k - \sum_{i=k}^{n-1} \lambda_i c_{i-k}$ and $\frac{\partial}{\partial c_k} L(b,c,\lambda) = c_k - \sum_{i=k}^{n-1} \lambda_i b_{i-k}$. Equivalently, $\nabla_b L(b,c,\lambda) = b - \Gamma(\lambda) c$ and $\nabla_c L(b,c,\lambda) = c - \Gamma(\lambda) b$. 
Thus, in order to ensure $\nabla_{b,c} L(b,c,\lambda) = 0$  when $b=c=\vec f$, it suffices to set $\lambda = \lambda^\star$, where $\lambda^\star$ is chosen so as to solve 
\begin{align*}
    \vec f = \Gamma\br{\lambda^\star} \vec f &\iff P \vec f = P \Gamma\br{\lambda^\star} \vec f \\&\iff \begin{pmatrix}f_{n-1} \\ f_{n-2} \\ \vdots \\ f_0 \end{pmatrix} = \begin{pmatrix}
\lambda^\star_{n-1} & 0 & 0 & \ldots & 0\\
\lambda^\star_{n-2} & \lambda^\star_{n-1} & 0 & \ldots & 0 \\
    \lambda^\star_{n-3} & \lambda^\star_{n-2} & \lambda^\star_{n-1} & \ldots & 0\\
\vdots & \vdots & \vdots & \vdots & \vdots \\
\lambda^\star_{0} & \lambda^\star_1 & \lambda^\star_2 & \ldots & \lambda^\star_{n-1}
\end{pmatrix}\begin{pmatrix}f_0 \\ f_{1} \\ \vdots \\ f_{n-1} \end{pmatrix}\\
&\iff \begin{pmatrix}f_{n-1} \\ f_{n-2} \\ \vdots \\ f_0 \end{pmatrix} = \begin{pmatrix}
f_0 & 0 & 0 & \ldots & 0\\
f_{1} & f_{0} & 0 & \ldots & 0 \\
    f_{2} & f_{1} & f_{0} & \ldots & 0\\
\vdots & \vdots & \vdots & \vdots & \vdots \\
f_{n-1} & f_{n-2} & f_{n-3} & \ldots & f_{0}
\end{pmatrix}\begin{pmatrix}\lambda^\star_{n-1} \\ \lambda^\star_{n-2} \\ \vdots \\ \lambda^\star_{0} \end{pmatrix} ,
\end{align*}
where $P$ is the permutation matrix such that $Px$ is the reversal of the vector $x$. Thus, we can solve for $\lambda^\star$ to obtain
\[P \lambda^\star = \begin{pmatrix}\lambda^\star_{n-1} \\ \lambda^\star_{n-2} \\ \vdots \\ \lambda^\star_{0} \end{pmatrix} = \begin{pmatrix}
f_0 & 0 & 0 & \ldots & 0\\
f_{1}-f_0 & f_{0} & 0 & \ldots & 0 \\
    f_{2}-f_1 & f_{1} - f_{0} & f_0 & \ldots & 0\\
\vdots & \vdots & \vdots & \vdots & \vdots \\
f_{n-1}-f_{n-2} & f_{n-2}-f_{n-3} & f_{n-3}-f_{n-4} & \ldots & f_{0}
\end{pmatrix}\begin{pmatrix}f_{n-1} \\ f_{n-2} \\ \vdots \\ f_0 \end{pmatrix},\]
where we used the fact that $\vec f$ is a feasible solution:
\begin{align*}\begin{pmatrix}
f_0 & 0 & 0 & \ldots & 0\\
f_{1} & f_{0} & 0 & \ldots & 0 \\
    f_{2} & f_{1} & f_0 & \ldots & 0\\
\vdots & \vdots & \vdots & \vdots & \vdots \\
f_{n-1} & f_{n-2}& f_{n-3} & \ldots & f_{0}
\end{pmatrix}\begin{pmatrix}
f_0 & 0 & 0 & \ldots & 0\\
f_{1} & f_{0} & 0 & \ldots & 0 \\
    f_{2} & f_{1} & f_0 & \ldots & 0\\
\vdots & \vdots & \vdots & \vdots & \vdots \\
f_{n-1} & f_{n-2}& f_{n-3} & \ldots & f_{0}
\end{pmatrix} = \begin{pmatrix}
1 & 0 & 0 & \ldots & 0\\
1 & 1 & 0 & \ldots & 0 \\
    1 & 1 & 1 & \ldots & 0\\
\vdots & \vdots & \vdots & \vdots & \vdots \\
1 & 1 & 1 & \ldots & 1
\end{pmatrix} \\
\iff
{\begin{pmatrix}
f_0 & 0 & 0 & \ldots & 0\\
f_{1} & f_{0} & 0 & \ldots & 0 \\
    f_{2} & f_{1} & f_0 & \ldots & 0\\
\vdots & \vdots & \vdots & \vdots & \vdots \\
f_{n-1} & f_{n-2}& f_{n-3} & \ldots & f_{0}
\end{pmatrix}}^{-1} = {\begin{pmatrix}
1 & 0 & 0 & \ldots & 0\\
1 & 1 & 0 & \ldots & 0 \\
    1 & 1 & 1 & \ldots & 0\\
\vdots & \vdots & \vdots & \vdots & \vdots \\
1 & 1 & 1 & \ldots & 1
\end{pmatrix}}^{-1}\begin{pmatrix}
f_0 & 0 & 0 & \ldots & 0\\
f_{1} & f_{0} & 0 & \ldots & 0 \\
    f_{2} & f_{1} & f_0 & \ldots & 0\\
\vdots & \vdots & \vdots & \vdots & \vdots \\
f_{n-1} & f_{n-2}& f_{n-3} & \ldots & f_{0}
\end{pmatrix} \\
\iff {\begin{pmatrix}
f_0 & 0 & 0 & \ldots & 0\\
f_{1} & f_{0} & 0 & \ldots & 0 \\
    f_{2} & f_{1} & f_0 & \ldots & 0\\
\vdots & \vdots & \vdots & \vdots & \vdots \\
f_{n-1} & f_{n-2}& f_{n-3} & \ldots & f_{0}
\end{pmatrix}}^{-1} = {\begin{pmatrix}
1 & 0 & 0 & \ldots & 0\\
-1 & 1 & 0 & \ldots & 0 \\
    0 & -1 & 1 & \ldots & 0\\
\vdots & \vdots & \vdots & \vdots & \vdots \\
0 & 0 & 0 & \ldots & 1
\end{pmatrix}}\begin{pmatrix}
f_0 & 0 & 0 & \ldots & 0\\
f_{1} & f_{0} & 0 & \ldots & 0 \\
    f_{2} & f_{1} & f_0 & \ldots & 0\\
\vdots & \vdots & \vdots & \vdots & \vdots \\
f_{n-1} & f_{n-2}& f_{n-3} & \ldots & f_{0}
\end{pmatrix}.
\end{align*}

From the definition $f_0=1$ and $f_k = f_{k-1}(1-1/2k)$, for all $k \ge 1$, we have
\begin{align*}\lambda^\star_{n-1-i} &= f_0 f_{n-1-i} - \sum_{j=1}^i \br{f_j-f_{j-1}}f_{n-1-i+j} \\&= f_{n-1-i} - \sum_{j=1}^i f_{j-1}\br{\frac{1}{2j}}f_{n-1-i+j}\\&=f_{n-1-i}\br{1-\sum_{j=1}^if_{j-1}\br{\frac{1}{2j}}\frac{f_{n-1-i+j}}{f_{n-1-i}}}.\end{align*}
Since $\frac{f_{n-1-i+j}}{f_{n-1-i}} = \prod_{\ell=1}^j \br{1-\frac{1}{2(n-i+\ell)}} \le 1$, we have
\begin{align} \lambda^\star_{n-1-i} \ge f_{n-1-i}\br{1-\sum_{j=1}^if_{j-1}\br{\frac{1}{2j}}} > f_{n-1-i}\br{1-\sum_{j=1}^\infty  f_{j-1}\br{\frac{1}{2j}}} = f_{n-1-i}\br{1-\sum_{j=0}^\infty\frac{f_{j}}{2(j+1)}}\label{eq:lam_eqn}.\end{align}
Since $\sum_{j=0}^\infty f_j x^j = \frac{1}{\sqrt{1-x}}$ for all $x \in [0 , 1)$, integrating both sides between the limits $0$ and $t$, we have
\[\sum_{j=0}^\infty \frac{f_j}{j+1} t^{j+1}=2\br{1-\sqrt{1-t}}.\]
Taking limits as $t \to 1$, we obtain
\[\sum_{j=0}^\infty \frac{f_j}{j+1}=2.\]
Thus, from \cref{eq:lam_eqn}, we can conclude  that $\lambda^\star_{n-1-i} > 0$ for all $i \in [n]$ and hence $\Gamma\br{\lambda^\star}^2$ is a matrix with all entries $(\Gamma(\lambda^\star)^2)_{i,j} = \sum_{k=0}^{n-1-\max\{i,j\}} \lambda^\star_{i+k}\lambda^\star_{j+k}$ strictly positive. Hence, by the Perron-Frobenius theorem, we have that $\Gamma\br{\lambda^\star}^2$ has a unique eigenvector with all coordinates positive, and this corresponds to the maximum eigenvalue. Since $f$ is an eigenvector of $\Gamma\br{\lambda^\star}^2$ by construction, it must be the unique leading eigenvector corresponding to the eigenvalue $1$. Hence $\norm{\Gamma\br{\lambda^\star}^2}_* \leq 1$.

Thus, $L\br{\bfb, \bfc, \lambda^\star}$ is convex in $\br{\bfb, \bfc}$, since $\norm{\Gamma\br{\lambda}^2}_\star \leq 1 \iff \norm{\Gamma\br{\lambda}}_\star \leq 1$.
And $\lambda^\star$ was chosen to ensure $\nabla_{b,c} L(b,c,\lambda^\star)=0$.
This lets us conclude that $\br{\vec f, \vec f, \lambda^\star}$ is a stationary point of the Lagrangian that is convex in $\bfb, \bfc$ and concave in $\lambda$, hence it is the global optimum. 

\end{proof}

\Cref{lem:opt-ltt} fully characterizes the optimal lower triangular Toeplitz factorization.
We now compare the value of the matrix factorization objective for lower triangular Toeplitz factorizations versus general factorizations.
For convenience, we define a notation for these values:

\begin{defn}
Let $\bfA\in\{0,1\}^{n\times n}$ be the lower triangular matrix of all ones.
Recall that $C \in \mathbb{R}^{n \times n}$ is lower triangular Toeplitz if $C_{i,j}=0$ for all $i < j$ and $C_{i,j} = c_{i-j}$ for all $i \ge j$, where $c \in \mathbb{R}^n$ is a vector.
Define 
\begin{align}
    \mathsf{OptLTToe}(n) &\coloneqq \inf \left\{ \MaxErr(B,C) : B,C \in \mathbb{R}^{n \times n} \text{ lower triangular Toeplitz}, BC=A \right\}, \label{eq:optlttoe} \\
    \mathsf{Opt}(n) &\coloneqq \inf \left\{ \MaxErr(B,C) : B,C \in \mathbb{R}^{n \times n} , BC=A \right\}, \label{eq:optgen}
\end{align}
where $\MaxErr(B,C) \coloneqq \|B\|_{2\to\infty} \|C\|_{1 \to 2} = \sqrt{\max_i \sum_j B_{i,j}^2} \sqrt{\max_j \sum_i C_{i,j}^2}$.
\end{defn}
By \Cref{lem:lbstruct} and \Cref{lem:fhu-properties}, for all integers $n \ge 1$, we have
\begin{equation}
    1 + \frac{\gamma-1 + \log(n)}{\pi} \le \mathsf{OptLTToe}(n) = 1 + \sum_{k=1}^{n-1} \left( 2^{-2k} {2k \choose k} \right)^2 \le 1 + \frac{\gamma+\log(n)}{\pi}.
\end{equation}
The lower bound of \citet{matouvsek2020factorization} in \cref{eq:sasho} states the optimal over all factorizations is \[\mathsf{Opt}(n) = \gamma_2(A^{(n)}) \ge \frac{\log\left(2n+1\right)}{\pi}.\]
Combining these bounds gives a bound on the gap between the optimal general factorization and the optimal lower triangular Toeplitz factorization.
\begin{cor} \label{cor:toeplitz_gap}
    The gap between the lower triangular Toeplitz and general optimal factorizations is 
\[\mathsf{OptLTToe}(n)-\mathsf{Opt}(n) \le \left(1 + \frac{\gamma+\log(n)}{\pi}\right)-\left( \frac{\log(2n+1)}{\pi}\right) \le 1+\frac{\gamma-\log(2)}{\pi} \le 0.97.\]
\end{cor}
Numerically, we have $\mathsf{OptLTToe}(n)-\mathsf{Opt}(n) \le 0.365$ for all $n \ge 1$.
\Cref{fig:bounds} shows numerically how the various known upper and lower bounds compare for differing values of $n$.

We conclude this subsection by remarking that, without the Toeplitz constraint, we can assume that the factors are lower triangular without loss of generality.
This is a consequence of the fact that two factorizations $A=BC$ and $A=\hat B \hat C$ are functionally equivalent as long as $BB^T = \hat B \hat B^T$. To see why this fact is true, note that for $z \gets \mathcal{N}(0,\sigma^2I)$, the distribution of $Bz$ is $\mathcal{N}(0,\sigma^2 BB^T)$. So if $BB^T = \hat B \hat B^T$, then $Bz$ and $\hat B z$ have the exact same distribution. Note that this fact is specific to the Gaussian distribution.
Given this fact, we can take an arbitrary factorization $A=BC$ and perform a Cholesky decomposition of $BB^T$ to obtain a lower triangular $\hat B$ satisfying $BB^T = \hat B \hat B^T$.
However, if we restrict to Toeplitz matrices, then we cannot assume they are lower triangular without loss of generality, since the Cholesky decomposition may not preserve this structure.

\subsection{Applications of Matrix Factorizations in Machine Learning} \label{sec:dpftrl}

In the recent years, Matrix Factorization (MF) have found extensive use both in obtaining optimal private regret/stochastic convex optimization (SCO) guarantees~\cite{thakurta2013nearly, agarwal2017price, kairouz2021practical,asi2023near, Cutosky21}, and in practical deployments~\cite{kairouz2021practical, denisov22matfact, choquette22multiepoch,choquette2024amplified,xu2023federated}. The main observation that maps MF to the (convex) optimization setup is that (unconstrained) stochastic graident descent update with constant learning rate can be viewed as computing prefix sums over adaptively chosen gradient vectors.\footnote{There are more general mappings for constrained optimization, and adaptive optimizers like momentum methods, but a detailed discussion on those is tangential to this work.} Differentially Private Follow-the-regularized-leader (DP-FTRL)~\cite{kairouz2021practical,thakurta2013nearly} is an optimization algorithm that arises naturally from this mapping. DP-FTRL has found extensive usage in production grade deployment of DP learning. For example, all next-word-prediction models on Android Gboard is trained with DP-FTRL~\cite{xu2023federated}. 

For a given set of $\ell_2$-Lipschitz loss functions $\{f_1(\theta),\ldots,f_\nd(\theta)\}$ (where each loss function $f_i$ can be thought of as the optimization loss on a disjoint minibatch of data samples), the simplest version of DP-FTRL can be described as follows:
\vspace{-0.2cm}
\begin{equation}
    \texttt{DP-FTRL:}\ \ \ \  \theta_{t}\leftarrow \theta_0 -\eta \left( \sum_{k=0}^{t-1} \nabla f_k(\theta_k) + \hat{z}_k \right).\label{eq:DP-FTRL}
\end{equation}
In~\cref{eq:DP-FTRL}, $\eta$ is the learning rate, and $\{\hat{z}_t\}$ are the correlated Gaussian noise added to ensure that the computation of the sequence of $\{\theta_t\}$'s are differentially private. Mapping to the MF view in~\cref{eq:mfviews}, the adaptively chosen data set is set to $x=\{\nabla f_t(\theta_t): t\in[\nd]\}$, and the noise sequence $\{\hat{z}_t\}$ is set to $B z$, where $z \leftarrow  \calN(0,\sigma^2I)$. DP-FTRL captures a large class of noise addition mechanisms (parameterized by the specific factorization scheme $A=BC$ chosen). In particular, it captures DP-SGD~\cite{song2013stochastic,BST14,DP-DL} (another heavily used private learning algorithm) as a special case~\cite{choquette2024amplified}, and has much better privacy/utility trade-offs in general~\cite{choquette2023correlated}. \cite{kairouz2021practical,asi2023near} showed that for obtaining the best regret/SCO guarantees from DP-FTRL style algorithms, it suffices to optimize for $\MaxErr(B, C)$ to instantiate the noise mechanism in~\cref{eq:DP-FTRL}.

One major challenge in large deployments of DP-FTRL is that noise generation for an arbitrary matrix mechanism can be prohibitively expensive, as discussed in \cref{sec:intro}. Our work provides a practical approach that only requires a constant factor blow-up (usually at most 4-6 times) in memory as compared to that of DP-SGD, and at privacy/utility trade-off comparable to DP-FTRL with an optimal matrix-factorization mechanism. (DP-SGD adds independent Gaussian noise to the gradients at each state update.)


\subsection{Other Related Work}

\paragraph{Beyond additive noise:}
Our algorithms are based on adding Gaussian noise (with a carefully chosen covariance). Since the noise is additive, the error is independent of the number of data points.\footnote{The error of our algorithms depends on the number of iterations $\nd$. This is usually closely related to the number of data points $|x|$, but, in general, these two quantities can be unrelated.} The expected maximum error is \begin{equation}
    \forall x ~~~ \ex{\calM}{\|\calM(x)-Ax\|_\infty} = \ex{\calM}{\max_k |(\calM(x)-Ax)_k|} \le O(\log(n))^{3/2}, \label{eq:Emaxerr}   
\end{equation}
while the maximum expected error is 
\begin{equation}
    \forall x ~~~ \max_k \ex{\calM}{|(\calM(x)-Ax)_k|} \le O(\log(n)). \label{eq:maxEerr}
\end{equation}
The additional $\sqrt{\log(n)}$ term in \Cref{eq:Emaxerr} compared to \Cref{eq:maxEerr} comes from taking a union bound over $n$ Gaussians.
Note that, for simplicity, we suppress the privacy parameters; to achieve $(\varepsilon,\delta)$-differential privacy, the error above scales with $\frac1\varepsilon\log(1/\delta)$; to achieve $\rho$-zCDP \cite{bun2016concentrated}, the error above scales with $1/\sqrt{\rho}$.

The bound in \cref{eq:maxEerr} on the maximum expected error is tight up to constants in the streaming setting -- in fact the lower bound holds for \emph{average} expected error: 
\begin{equation}
    \forall \calM ~ \exists x ~~~ \frac1n \sum_k \ex{\calM}{|(\calM(x)-Ax)_k|} \ge \Omega(\log(n)),
\end{equation}
assuming $\calM$ is differentially private \cite[Theorem 4]{henzinger2023almost}.
However, the bound in \cref{eq:Emaxerr} on expected maximum error is \emph{not} tight if we move beyond additive noise mechanisms \cite{dwork2015pure}.
\citet{dwork2015pure} show that we can obtain an improved error guarantee:\footnote{\citet{dwork2015pure} stated this result for pure differential privacy, with $\log(|x|)^{2}$ instead of $\log(|x|)^{3/2}$.}
\begin{equation}
    \exists \calM ~ \forall x ~~~ \ex{\calM}{\|\calM(x)-Ax\|_\infty} \le O(\log(n) + \log(|x|)^{3/2}), \label{eq:dwork2015pure}   
\end{equation}
where $|x|$ is the number of data points.\footnote{In most of the differential privacy literature $n$ is the number of data points / people, but we use $n$ for the number of iterations; $|x|$ should be read as the size of the dataset $x$.}
The high level idea is to ``freeze'' the output of $\calM$ until the partial sum has changed significantly and only then update the output.
This works because deciding whether the partial sum has changed significantly is easier than evaluating it. Specifically, we can use the so-called ``sparse vector'' technique to decide when to update.
This algorithm is a reduction -- when it does decide to update, it uses additive noise -- the ``win'' is that we need to use additive noise for $\le |x|$ updates, rather than $n$ updates. Thus our improvements for additive noise mechanisms can be applied in this setting.

\cref{eq:dwork2015pure} is only an improvement over \cref{eq:Emaxerr} in the sparse setting ($|x| \ll n$). Machine learning applications are typically dense ($|x| \ge n$).

\paragraph{Offline setting:}
Our work is in the streaming setting, where we must output an approximation to each partial sum as soon as it is complete.
It is also natural to consider the offline setting, where we are given the entire input at once.
In the offline setting, the memory constraints that we focus on are not as relevant.

The offline version of continual counting is known as threshold queries or quantiles and this task has been studied extensively \cite{beimel2013private,bun2015differentially,bun2018composable,alon2019private,kaplan2020privately,gillenwater2021differentially,cohen2023optimal}.
The offline setting permits more sophisticated algorithms, which achieve asymptotically better error guarantees.
\citet{cohen2023optimal} give an $(\varepsilon,\delta)$-differentially private algorithm $\calM$ with error
\begin{equation}
    \forall x ~~~ \ex{\calM}{\|\calM(x)-Ax\|_\infty} \le \widetilde{O}\left(\frac1\varepsilon \log^2(1/\delta) \log^*(n)\right)
\end{equation}
Note that offline algorithms are not applicable in the machine learning applications discussed in \Cref{sec:dpftrl}. 

\paragraph{Pure differential privacy:}
The early work on continual counting \cite{Dwork-continual,CSS11-continual,xiao2010differential} works under pure $(\varepsilon,0)$-differential privacy, which rules out Gaussian noise.
Instead Laplace noise can be used. In this case, we must bound the $L_1$ sensitivity. So the objective for the matrix factorization is to minimize $\|B\|_{2\to\infty} \cdot \|C\|_{1 \to 1}$, where $\|C\|_{1 \to 1} = \max_j \sum_i |C_{i,j}|$ is the maximum $1$-norm of a column. 
This subtly changes the problem. The binary tree factorization still works well with this objective. However, Toeplitz factorizations do not work as well. Intuitively, the binary tree factorization relies on sparsity which controls both the $1$- and $2$-norms, whereas Toeplitz factorizations produce dense matrices for which the the $1$- and $2$-norms differ substantially.

\paragraph{Factoring other matrices:}
We exclusively look at factorizing the all-ones lower triangular matrix. This is a natural and well-motivated problem.
However, there is also work on factorizing other matrices.
Most closely related to the all-ones lower triangular matrix $A$, \citet{mathias1993hadamard} gives exact expressions for $\gamma_2(A+A^T)$ and $\gamma_2(A+A^T+I)$.
There is a different line of work \cite{li2015matrix,mckenna2021hdmm} considering factorizations of generic workload matrices, which have a completely different structure than the matrices we consider.

\subsection*{Notation Summary}
Before proceeding, we briefly summarize the key symbols and notation used throughout:
\begin{table}[H]
\renewcommand{\arraystretch}{1.4}%
\begin{tabular}{ll}
\toprule
\nd & Number of steps on which DP estimates are released.  \\
$\nbuf$ & Order of recurrence / degree / number of memory buffers. \\
$\mdim$ & Dimension of per-step user contributions (e.g., model size). \\
$[k]$ & $=\{0, \dots, k - 1\}$ for $k \ge 1$. \\
$\N$ & $= \{0, 1, 2, \dots \}$, the natural numbers including zero. \\ 
$f$ & Rational (generating) function where $f(x) = f_0 + f_1 x + f_2 x^2 \cdots$. \\
$p$ & Polynomial of degree $d$ with $p(x) = p_0 + p_1 x + p_2 x^2 + \cdots + p_d x^d$. \\
$\seqk{f_k}$ & Sequence $f_0, f_1, \dots$ with corresponding generating function $f$. \\
$\LTT(f, \nd) \in \R^\dimdim$ & Lower-triangular Toeplitz matrices defined by a generating function $f$. \\
$\cdot$ & Context-dependent multiplication, used where it improves readability. \\
$f * r$ & Cauchy product / convolution $(f * r)_n = \sum_{i=0}^{n} f_i \cdot r_{n-i}$. \\
$A\tp$  & Transpose of a matrix $\bfA$. \\
$A^\star$  & A matrix $A$ that is ``optimal'' in a context-dependent sense. \\
$A^\dagger$  & Moore-Penrose pseudoinverse of matrix $\bfA$. \\
$A\idx{i}{j}$ & The $(i, j)^{\text{th}}$ entry of matrix $\bfA$, zero-indexed.\\
$A\idx{i}{:}$ and $\bfA\idx{:}{j}$  & The $i^{\text{th}}$ row and $j^{\text{th}}$ column. \\
$\lfrob{A}$  & The Frobenius norm of a matrix $A$. \\
$\bfA \in \R^\dimdim$  & Specifically, the lower-triangular matrix of ones to be factorized as $\bfA = \bfB \bfC$. \\
$\log$ & Natural logarithm, $\log(2.718) \approx 1$.\\
\bottomrule
\end{tabular}
\caption{Summary of notation}\label{tab:notation}
\end{table}

\section{Efficiently-Sampleable Factorization via Rational Generating Functions} \label{sec:genfunc-framework}

In this section we describe our main algorithmic tools. 
First (in \S\ref{sec:genfuncmat}) we present the view of lower triangular Toeplitz matrices in terms of generating functions or sequences. This is a convenient mathematical formalism for analysis \cite{fichtenberger2022constant}.
Second (in \S\ref{sec:rational-constrec}) we discuss the special case of rational generating functions or, equivalently, constant-recurrent sequences. Our results focus on this special case, as the additional structure is useful for our algorithms; this structure is where we depart from the prior literature.
Third (in \S\ref{sec:algorithm}) we present our algorithm for sampling from lower triangular Toeplitz matrices with rational generating functions (that is, the multiplication algorithm for \BLT matrices).
In subsequent sections (\S\ref{sec:rationalfunctionapprox}, \S\ref{sec:practical-optimization}) we instantiate these rational generating functions.

\subsection{Sequences, Lower Triangular Toeplitz Matrices, \& Generating Functions} \label{sec:genfuncmat}

We begin by describing the generating function view of lower triangular Toeplitz matrices used by \citet{fichtenberger2022constant}. 
Our analysis moves fluidly between three different views of the same mathematical object:
\begin{itemize}
    \item The sequence $\seqk{f_k} = f_0, f_1, f_2, \dots$ with $f_k \in \R$ for each index $k$.
    \item The ordinary generating function of the sequence, $f: \C \rightarrow \C$,
      \[f(x) = \sum_{k=0}^\infty f_k x^k = \sum_{k=0}^\infty \frac{f^{(k)}(0)}{k!} x^k.
      \]
    (It will be clear from context if $f$ refers to the sequence entries or the generating function).
    \item The family of $\dimdim$ lower-triangular Toeplitz matrices $M(f,\nd) \in \mathbb{R}^{\nd \times \nd}$ for $\nd \in \N$ that are generated by $f$ as
    \begin{equation}
    \forall i,j \in [\nd] ~~~~~~~~~~
    M(f,\nd)_{i,j} \coloneqq \left\{\begin{array}{cl} f_{i-j} & \text{ if } i \ge j \\ 0 & \text{ if } i<j \end{array}\right\} 
    \label{eq:genfuncmatrix}
\end{equation} 
    When the specific $\nd$ is unimportant or clear from context, we write $M(f)$ to refer to any matrix in this family (or equivalently, the infinite-dimensional linear operator).
\end{itemize}
%
%
As a canonical example, we have the all-ones sequence,
\begin{equation}
\seqk{g_k} = 1, 1, 1, 1, \dots, 
\quad
 g(x) = \frac{1}{1-x} = \sum_{k=0}^\infty x^k, 
 \quad
A^{(4)} \coloneqq M(g, 4) = \begin{bmatrix}
1 & 0 & 0 & 0 \\
1 & 1 & 0 & 0 \\
1 & 1 & 1 & 0 \\
1 & 1 & 1 & 1 \\
\end{bmatrix} 
\label{eq:allonesgf}  
\end{equation}

It is straightforward to see that addition (under the usual definitions) is equivalent across all three representations, and (slightly less obviously) the same fact holds for suitable definitions of multiplication. We denote the Cauchy product or convolution $*$ for sequences $f$ and $g$ by
\begin{equation}\label{eq:cauchyprod}
    \seqk{f_k} * \seqk{g_k} = \seqk{h_k}
    \quad \text{where} \quad 
    h_k \coloneqq \sum_{i=0}^{k} f_i \cdot g_{k-i}, 
\end{equation}
and write $h = f * g$ as a shorthand.

Then, the mathematical structure summarized in the following lemma is key to our analysis:
\begin{lem}\label{lem:mgfmult}
  Let $f$, $g$, and $h$ be ordinary generating functions with corresponding sequences. Then the Cauchy product of sequences, multiplication of the generating functions, and matrix multiplication are all equivalent. That is,
         \[ 
         \left(h = f * g \right) 
         \quad \Longleftrightarrow \quad
         \big( h(x) = f(x) g(x) \big)
         \quad \Longleftrightarrow \quad
         \big( \forall \nd ~~ M(h,\nd) = M(f,\nd) \cdot M(g,\nd) \big).
      \]
      Similarly the usual definitions of addition are equivalent:
      \[
      \big( \forall i ~~ h_i = f_i + g_i \big)
         \quad \Longleftrightarrow \quad
         \big( h(x) = f(x) + g(x) \big)
         \quad \Longleftrightarrow \quad
         \big( \forall \nd ~~ M(h,\nd) = M(f,\nd) + M(g,\nd) \big).  
         \]
\end{lem}


\begin{proof}
    The results for sequences and their ordinary generating functions are well known, see e.g. \citet{kauers11concrete} or their equivalent results based on well-known properties of the Z-transform; cf.~\cite{oppenheim1997signals}. The results for addition are also immediate. Hence, it is sufficient to show $h(x) = f(x) g(x)$ iff 
    $\forall \nd \in \N, M(h,\nd) = M(f,\nd) \cdot M(g,\nd).$ 

    Fix $i,j \in [\nd]$. 
    If $i<j$, then $M(h,\nd)_{i,j} = 0$ and $(M(f,\nd) \cdot M(g,\nd))_{i,j} = \sum_{k \in [\nd]} M(f,\nd)_{i,k} \cdot M(g,\nd)_{k,j} = \sum_k 0$, since, for all $k \in [\nd]$, either $i<k$, whence $M(f,\nd)_{i,k}=0$, or $k<j$, whence $M(g,\nd)_{k,j}=0$, (or both). 
    Thus we can assume $i \ge j$.

    By the product rule and induction, we have the derivative $h^{(\ell)}(x) = \sum_{k = 0}^\ell {\ell \choose k} f^{(\ell-k)}(x) \cdot g^{(k)}(x)$ for all $\ell \ge 0$ and all applicable $x$.
    Thus
    \begin{align*}
        M(h,\nd)_{i,j} &= \frac{h^{(i-j)}(0)}{(i-j)!} \\
        &= \frac{1}{(i-j)!} \sum_{k = 0}^{i-j} {i-j \choose k} f^{(i-j-k)}(0) \cdot g^{(k)}(0) \\
        &= \sum_{k = 0}^{i-j} \frac{f^{(i-j-k)}(0)}{(i-j-k)!} \cdot \frac{g^{(k)}(0)}{k!} \\
        &= \sum_{k' = j}^{i} \frac{f^{(i-k')}(0)}{(i-k')!} \cdot \frac{g^{(k'-j)}(0)}{(k'-j)!} \tag{$k'=k+j$}\\
        &= \sum_{k' = j}^{i} M(f,\nd)_{i,k'} \cdot M(g,\nd)_{k',j} \\
        &= \sum_{k' \in [\nd]} M(f,\nd)_{i,k'} \cdot M(g,\nd)_{k',j} \tag{$i<k' \implies M(f,\nd)_{i,k'}=0$, $k'<j \implies M(g,\nd)_{k',j}=0$} \\
        &= \left( M(f,\nd) \cdot M(g,\nd) \right)_{i,j}.
    \end{align*}
\end{proof}

Given \cref{eq:allonesgf} and the above result for multiplication, in order to factorize $\bfA = \bfB \bfC$, it is natural to factorize by taking the square root of the generating function. Let $f(x) = \frac{1}{\sqrt{1-x}}$. Since $f(x) \cdot f(x) = \frac{1}{1-x} = g(x)$ is the generating function of the all ones sequence, it follows that $M(f,\nd) \cdot M(f,\nd) = M(g,\nd) = A^{(\nd)}$.
Indeed, this is the factorization given by \citet{fichtenberger2022constant}, which we show to be optimal among the class of lower triangular Toeplitz matrices in \Cref{lem:opt-ltt}.

Given any generating function $r(x)$, we can obtain a valid factorization $\bfA = \bfB \bfC$ where $\bfB = M(b,\nd)$ for $b(x)=r(x)/(1-x)$ and $\bfC = M(c,\nd)$ for $c(x) = 1/r(x)$.
Our approach is to choose a generating function such that $r(x) \approx \sqrt{1-x}$ but which also permits us to design an efficient algorithm.

\subsection{Rational Generating Functions and Constant-Recurrent Sequences} \label{sec:rational-constrec}

We will focus our attention on rational generating functions of the form
\begin{equation}
    r(x) = \frac{p(x)}{q(x)} = \frac{\sum_{i=0}^{d-1} p_i x^i}{1+\sum_{j=1}^d q_j x^j} .\label{eq:rationalpoly}
\end{equation}
These have convenient algorithmic properties. In this subsection, we discuss several equivalent representations of rational functions.
As a starting point, a rational function of degree $d$ can be represented as a ratio of polynomials, where the numerator $p(x)$ has degree $<d$ and the denominator has degree $\le d$.
Note that we normalize the denominator so that $q(0)=q_0=1$.\footnote{We can rescale both the numerator and denominator by a constant to make this assumption true. Thus this assumption is without loss of generality, unless the rational function has a pole at $0$ -- but we do not consider such functions, as they are not valid generating functions.}

\Cref{eq:rationalpoly} is a convenient mathematical representation, but for algorithmic purposes we need the sequence $\seqk{r_k}$ of Taylor series coefficients
\begin{equation}
    r(x) = \sum_{k=0}^\infty r_k x^k. \label{eq:r-series}
\end{equation}
Thus our first task is to map from \Cref{eq:rationalpoly} to \Cref{eq:r-series} with a convenient representation that is amenable to efficient computation. 

The generating function $r(x)$ being a rational function is equivalent to the sequence $\seqk{r_k}$ being a constant-recurrent sequence (a.k.a.~linear-recursive or C-finite sequence); see for example \citet{kauers11concrete}.
The sequence terms can be computed by taking powers of a matrix; this is the representation that we will use in our algorithm.
We summarize this result in \cref{lem:rationaltoconstrec} and provide a proof for completeness.

\begin{lem}[Constant-recurrent Taylor series representation of a rational function]\label{lem:rationaltoconstrec}
    Let 
    \[r(x) = \frac{p(x)}{q(x)} = \frac{\sum_{i=0}^{d-1} p_i x^i}{1+\sum_{j=1}^d q_j x^j}
    \]
    be a rational function of degree $ \le d$.    
    As in \cref{sec:genfuncmat}, let \[r(x) = \sum_{k=0}^\infty r_k x^k.\] That is, $r(x)$ generates the sequence $\seqk{r_k}$. 
    Then, for all $k \ge d$, this sequence satisfies the recurrence
    \begin{equation}
        r_k = - \sum_{j=1}^d q_j r_{k-j}.
    \end{equation}
    And, for $0 \le k < d$, we have $r_k = p_k - \sum_{j=1}^k q_j r_{k-j}$. 
    Furthermore, for all $k$, we have 
    \begin{equation} 
    r_k = u^T W^k v,
    \end{equation}
    where
    \begin{equation}
        u \coloneqq \left(\begin{array}{c} 1 \\ 0 \\ 0 \\ 0 \\ \vdots \\ 0 \\ 0 \end{array}\right) \in \mathbb{R}^{d \times 1}, \qquad
        v \coloneqq \left(\begin{array}{c} p_0 \\ p_1 \\ p_2 \\ p_3 \\ \vdots \\ p_{d-2} \\ p_{d-1} \end{array}\right) \in \mathbb{R}^{d \times 1},
    \end{equation}
    and
    \begin{equation}\label{eq:companionmatrix}
        W \coloneqq \left(\begin{array}{ccccccc} 
            -q_1 & 1 & 0 & 0 & \cdots & 0 & 0 \\
            -q_2 & 0 & 1 & 0 & \cdots & 0 & 0\\
            -q_3 & 0 & 0 & 1 & \cdots & 0 & 0\\
            -q_4 & 0 & 0 & 0 & \cdots & 0 & 0\\
            \vdots & \vdots & \vdots & \vdots & \ddots & \vdots & \vdots \\
            -q_{d-1} & 0 & 0 & 0 & \cdots & 0 & 1 \\
            -q_d & 0 & 0 & 0 & \cdots & 0 & 0
        \end{array}\right) \in \mathbb{R}^{d \times d}. 
    \end{equation}
\end{lem}
\begin{proof}
 
    We have 
    \[ r(x) = \frac{p(x)}{q(x)} = \frac{\sum_{i=0}^{d-1} p_i x^i}{1+\sum_{j=1}^d q_j x^j} = \sum_{k=0}^\infty r_k x^k ,\]
    which rearranges to 
    \[p(x) = \sum_{i=0}^{d-1} p_i x^i = q(x) \cdot r(x) = \left( 1+\sum_{j=1}^d q_j x^j \right) \cdot \left( \sum_{k=0}^\infty r_k x^k \right) = \sum_{k=0}^\infty \left( r_k + \sum_{j=1}^{\min\{d,k\}} q_j r_{k-j} \right) \cdot x^k. \]
    Matching coefficients gives $p_k = r_k + \sum_{j=1}^k q_j r_{k-j}$ for $1 \le k \le d-1$ and $0 = r_k + \sum_{j=1}^d q_j r_{k-j}$ for $k \ge d$, which rearranges to give the recurrence.

    Next we turn to the matrix representation.
    Define $p^{(0)}(x) = p(x)$ -- i.e., $p^{(0)}_i = p_i$ for $0 \le i \le d-1$.
    For all $k \ge 1$, let ${p}_j^{(k)} \coloneqq p_{j+1}^{(k-1)} - p_0^{(k-1)} q_{j+1}$ for $0 \le j \le d-2$ and let ${p}_{d-1}^{(k)} \coloneqq - p_0^{(k-1)} q_d$.
    We can check that these values satisfy
    \begin{equation}
        \frac{\sum_{i=0}^{d-1} p_i^{(k-1)} x^i}{1+\sum_{j=1}^d q_j x^j} = p_0^{(k-1)} + x \cdot \frac{\sum_{i=0}^{d-1} {p}_i^{(k)} x^i}{1+\sum_{j=1}^d q_j x^j}. \label{eq:polylongdivision}
    \end{equation}
    The left hand side of \Cref{eq:polylongdivision} with $k=1$ is simply $r(x)$; from this we can conclude that $r_0 = p_0^{(0)} = p_0$.
    Expanding \Cref{eq:polylongdivision} for $k=1$ and $k=2$ gives $r_1 = {p}_0^{(1)}$. Specifically, we have
    \[r(x) = \frac{\sum_{i=0}^{d-1} p_i^{(0)} x^i}{1+\sum_{j=1}^d q_j x^j} = p_0^{(0)} + x \cdot \frac{\sum_{i=0}^{d-1} {p}_i^{(1)} x^i}{1+\sum_{j=1}^d q_j x^j} = p_0^{(0)} + x \cdot \left( p_0^{(1)} + x \cdot \frac{\sum_{i=0}^{d-1} {p}_i^{(2)} x^i}{1+\sum_{j=1}^d q_j x^j} \right).\]
    We can iteratively repeat this expansion to extract the entire Taylor series.
    (This algorithm is known as polynomial long division.)
    That is, we can show by induction that, for all $k$, we have
    \[ r(x) = \sum_{\ell=0}^{k-1} p_0^{(\ell)} x^\ell + x^k \cdot \frac{\sum_{i=0}^{d-1} {p}_i^{(k)} x^i}{1+\sum_{j=1}^d q_j x^j}.\]
    Thus $r_k = p_0^{(k)}$ for all $k$.
    
    Now we can write the recursive definition of $p^{(k)}_j$ in matrix form: For all $k \ge 1$, we have
    \[
        \left(\begin{array}{c} {p}_0^{(k)} \\ {p}_1^{(k)} \\  \vdots \\ {p}_{d-1}^{(k)} \end{array}\right)
         = W \cdot 
         \left(\begin{array}{c} {p}_0^{(k-1)} \\ {p}_1^{(k-1)} \\  \vdots \\ {p}_{d-1}^{(k-1)} \end{array}\right),
    \]
    where $W \in \mathbb{R}^{d \times d}$ is given in Equation \ref{eq:companionmatrix}.
    Given this, induction shows that, for all $k \ge 0$,
    \[
        \left(\begin{array}{c} {p}_0^{(k)} \\ {p}_1^{(k)} \\  \vdots \\ {p}_{d-1}^{(k)} \end{array}\right)
         = W^k \cdot 
         \left(\begin{array}{c} {p}_0^{(0)} \\ {p}_1^{(0)} \\  \vdots \\ {p}_{d-1}^{(0)} \end{array}\right)
         = W^k \cdot v
    \]
    and, hence, $r_k = p_0^{(k)} = u^T \cdot W^k \cdot v$, as required.
\end{proof}

We make some remarks about \Cref{lem:rationaltoconstrec}:
\begin{enumerate}
    \item \cref{lem:rationaltoconstrec} assumes $\deg(p)<d$ and $\deg(q)\le d$, but it allows for the possibility that some of the coefficients $p_j,q_j$ (including $q_d$) may be zero. Thus by padding, we can accommodate arbitrary degree of both the numerator and denominator.
    
    \item 
    Lemma \ref{lem:rationaltoconstrec} is in fact an ``if and only if'' -- that is, any Taylor series satisfying either of the conclusions of the Lemma must be a rational function (see e.g. \citep{kauers11concrete}). However, we do not use the converse of the result.
    
    \item
    The eigenvalues of $W$ can be related to the poles of the corresponding rational function. Specifically,
    \begin{equation}
        \det(\lambda I - W) = 0 \iff \lambda^d + \sum_{j=1}^d q_j \lambda^{d-j} = 0 \iff \big( q(\tfrac{1}{\lambda}) = 0 \text{ or } \lambda=0=q_d \big). \label{eq:companion}
    \end{equation}
    
    \item
    Furthermore, the eigenvectors of $W$ can be written in terms of the poles: Suppose $q(\tfrac{1}{\lambda}) = 1 + \sum_{j=1}^d q_j (\tfrac{1}{\lambda})^j = 0$. Let $u^T = (\lambda^{d-1}, \lambda^{d-2}, \cdots, \lambda, 1) \in \mathbb{R}^{1 \times d}$ (or $\lambda=0$ and $q_d=0$). Then $u^T W = \lambda \cdot u^T$ or, equivalently, $W^T u = \lambda \cdot u$.
    
    \item
    If the eigenvalues of $W$ are distinct (and real), then the above remark produces a complete (real) eigenbasis, and so we can diagonalize the matrix $W$ (over $\R$). A diagonal matrix $W$ is particularly convenient for our streaming algorithm (\Cref{alg:streamingmatrixpower}).    
    
    \item
    If the eigenvalues of $W$ are contained in the unit circle (i.e., $|\lambda|\le 1$), then the computation of $W^k$ is numerically stable. By Equation \ref{eq:companion}, this is equivalent to the poles of the rational function having magnitude $\ge 1$.
    
    \item
    Suppose the rational function $r$ has distinct poles $\frac{1}{\theta_i}$, and can be presented as\footnote{Diagonalizing the matrix $W$ given by \Cref{lem:rationaltoconstrec} corresponds to finding a representation of this form.}
    \[
      r(x) = \sum_{i \in [\dgr]}  \frac{\omega_i}{1 - \theta_i \cdot x}
    \]
    via a partial-fraction decomposition (instead of as a ratio of polynomials as in \Cref{lem:rationaltoconstrec}).
    (Indeed, our rational function in \Cref{sec:rationalsqrt} is presented in this form in \cref{eq:explicitrationalsqrt}.)
    Then, we have a simpler closed form for the sequence,
    \[
      r_k  = \sum_{i \in [\dgr]}  \omega_i \theta_i^k.
     \]
    This immediately yields a matrix representation: $W$ is a diagonal matrix with $W_{i,i}=\theta_i$ for all $i$, and the vectors $u$ and $v$ simply need to satisfy $u_i v_i = \omega_i$ for all $i$.
    This matrix representation may be more convenient than the one given by \Cref{lem:rationaltoconstrec}; we use this approach extensively in \cref{sec:practical-optimization}.

    \item
    The form of the matrix $W$ in \Cref{lem:rationaltoconstrec} is known as a ``companion matrix'' \cite{enwiki:1201009369}. Namely, the matrix $W$ is a companion to the denominator polynomial $q(x)$ since the roots of the polynomial correspond to the eigenvalues of the matrix (per \cref{eq:companion}).
\end{enumerate}

Our choice in how to represent $r(x)$ has some ramifications for the design of the noise generation (multiplication) algorithm in the next section.

\begin{rem}\label{rem:degree}
    The representation in \Cref{eq:rationalpoly} assumes the numerator has lower degree than the denominator, i.e., $\deg(p)<d$ versus $\deg(q)\le d$. But, in this work, we typically consider rational functions with the numerator and denominator both having the same degree, i.e., $\deg(\overline{p})=\deg(q)=d$.
    This case can either be handled by padding (i.e., increment $d \mapsto d+1$ and set $q_{d+1}=0$) or by including an additive constant:
    \begin{equation}
        r(x) = \frac{\overline{p}(x)}{q(x)} = t + \frac{{p}(x)}{q(x)} 
          = \frac{{p}(x) + tq(x)}{q(x)},
         \quad \text{ where } \quad
          t = \frac{\overline{p}_d}{q_d} 
        \quad \text{and} \quad 
        {p}(x) = \overline{p}(x) - t q(x), 
    \end{equation}
    so that $\deg({p})\le d-1$.  We then have
    \begin{equation}\label{eq:rkwitht}
    r_k = u^T W^k v + t \mathbb{I}[k=0]
    \end{equation}
    with $(u, W, v)$ as given in \cref{lem:rationaltoconstrec} applied to $p(x)/q(x)$. Here $\mathbb{I}$ is the indicator function taking value $1$ if the condition holds and $0$ otherwise.        
    We prefer the representation with an additive constant, as it is more efficient to implement algorithmically.
\end{rem}
We note the indicator function can also be written as $\mathbb{I}[k=0] = 0^k$ for $k \ge 0$.
Thus we can incorporate the $t \mathbb{I}[k=0]$ into the matrix representation by appending a row/column of $0$s to $W$, appending $1$ to $u$ and $t$ to $v$:
\begin{equation}
    \left( \begin{array}{c} u \\ 1 \end{array} \right)^T \left( \begin{array}{cc} W & \mathbf{0} \\ \mathbf{0}^T & 0 \end{array} \right)^k \left( \begin{array}{c} v \\ t \end{array} \right)
    = \left( \begin{array}{c} u \\ 1 \end{array} \right)^T \left( \begin{array}{cc} W^k & \mathbf{0} \\ \mathbf{0}^T & 0^k \end{array} \right) \left( \begin{array}{c} v \\ t \end{array} \right)
    = u^T W^k v + t \cdot \mathbb{I}[k=0].\label{eq:consttopad}
\end{equation}
This increases the dimension $d$ by $1$, which is equivalent to padding. However, for efficient streaming algorithm we introduce next, we can save one memory buffer by working directly with the $r_k = u^T W^k v + t \mathbb{I}[k=0]$ representation rather than the pure matrix representation in \cref{eq:consttopad}.

\subsection{Efficient Sampling via \BLT Multiplication} \label{sec:algorithm}

The generating function view outlined in \Cref{sec:genfuncmat} gives us a matrix factorization from any function $r(x)$. Namely, $A^{(\nd)} = M(b,\nd) \cdot M(c,\nd)$ for $b(x)=r(x)/(1-x)$ and $c(x)=1/r(x)$. This turns out to be a good matrix factorization as long as $r(x) \approx \sqrt{1-x}$.
The other desideratum is being able to efficiently sample noise according to the matrix factorization.

In this subsection, we show that if $r(x)$ is a rational function of low degree, then there exists an efficient sampling algorithm or, equivalently, an efficient algorithm for streaming multiplication by the lower-triangular Toeplitz matrix $M(r(x),\nd)$.
This algorithm relies on the representation in \cref{sec:rational-constrec}.
The final missing ingredient, which we provide in \Cref{sec:rationalfunctionapprox,sec:practical-optimization}, is to instantiate $r(x)$.

Our algorithmic task is as follows. Fix a lower-triangular Toeplitz matrix $M(r,\nd) \in \mathbb{R}^{\nd \times \nd}$ with a generating function $r$. We are given as input a matrix $Z \in \mathbb{R}^{\nd \times \mdim}$ and must produce as output $\widehat{Z} \coloneqq M(r,\nd) \cdot Z \in \mathbb{R}^{\nd \times \mdim}$. 
The obvious algorithm for this task is to use standard matrix multiplication; this would take $O(\nd^2\mdim)$ time and $O(\nd\mdim)$ space. Another approach is to exploit the Toeplitz structure of $M(r,\nd)$ which makes the matrix multiplication a convolution; this means it can be accelerated using the fast Fourier transform to take $O(\nd \mdim \log \nd)$ time, but this is not a streaming algorithm and it still requires $O(\nd\mdim)$ space. In order to reduce the space required we need an algorithm that is tailored to the streaming setting.

For streaming prefix sum applications (and private deep learning applications like DP-SGD and DP-FTRL in particular), the input $Z$ consists of $\nd \mdim$ \emph{independent} samples from a Gaussian and $\widehat{Z} = M(r,\nd) \cdot Z$ is the \emph{correlated} noise we add to our private learning procedure -- each row corresponds to one training iteration and each column to one parameter of the model.
In particular, the rows $Z$ can be sampled as needed -- since they are independent noise -- and the output $\widehat{Z}$ can be returned one row at a time. 

This is our streaming setting: At each iteration $k$, our algorithm receives as input the next row $Z_{k,\cdot} \in \mathbb{R}^{\mdim}$ and must output the next row $\widehat{Z}_{k,\cdot} \in \mathbb{R}^{\mdim}$.
Our goal is to develop an algorithm that runs in time and space $O(\mdim)$ per iteration (i.e., $O(\nd\mdim)$ time in total).
Reducing the space usage relies on the fact that we do not need to store all of $Z$ or $\widehat{Z}$ in the streaming setting. Accomplishing this relies on the structure of the lower triangular Toeplitz matrix $M(r,\nd) \in \mathbb{R}^{d \times d}$ generated by a rational function 
\[
r(x) = t + \frac{\sum_{i=0}^{d-1} p_i x^i}{1+\sum_{j=1}^d q_j x^j} = r_0 + r_1 x + r_2 x^2 + r_3 x^3 + \cdots. 
\]
Following \cref{rem:degree} and \cref{lem:rationaltoconstrec}, we use the representation of the sequence $\seqk{r_k}$ in terms of matrix powers: $r_k = u^T W^k v + t \mathbb{I}[k=0]$.  This is the representation that we use for our algorithm.

\begin{algorithm}
    \begin{algorithmic}
    \State \textbf{Parameters:} Matrix $M(r,\nd)$ defined following \cref{lem:rationaltoconstrec,rem:degree} via column vectors $u,v \in \mathbb{R}^{d \times 1}$, a matrix $W \in \mathbb{R}^{d \times d}$, and a scalar $t \in \mathbb{R}$.
    \State \textbf{Streaming Input:} Row vectors $Z_{0,\cdot},Z_{1,\cdot}, \cdots, Z_{\nd-1,\cdot} \in \mathbb{R}^{1 \times \mdim}$.
    \State \textbf{Streaming Output:} Row vectors $\widehat{Z}_{0,\cdot},\widehat{Z}_{1,\cdot}, \cdots, \widehat{Z}_{\nd-1,\cdot} \in \mathbb{R}^{1 \times \mdim}$  satisfying $\widehat{Z}_k = (\LTT(r) Z)_k$.
    \State \textbf{Goal:} $\widehat{Z}_{k,\cdot} = \sum_{i=0}^{k} r_{i} Z_{k-i, \cdot}$, where $r_k = u\tp W^k v + t \mathbb{I}[k=0]$.
    \State Initialize $S_0 = 0 \in \mathbb{R}^{d \times \mdim}$. \Comment{State of the Algorithm}
    \For{$k = 0 , \cdots , \nd-1$}
        \State Receive input $Z_{k,\cdot} \in \mathbb{R}^{1 \times \mdim}$.
        \State Compute $S_{k+1} = v Z_{k,\cdot} + W \cdot S_{k} \in \mathbb{R}^{d \times \mdim}$.
        \State Return output $\widehat{Z}_{k,\cdot} = t Z_{k,\cdot} + u^T S_{k+1} \in \mathbb{R}^{1 \times \mdim}$.
    \EndFor
    \end{algorithmic}
    \caption{Streaming Multiplication by a \BLT Matrix \label{alg:streamingmatrixpower}}
\end{algorithm}


For a \BLT given via a rational generating function $r$ in the matrix representation of \cref{lem:rationaltoconstrec}, \cref{alg:streamingmatrixpower} in fact computes $\hat{Z} = M(r) Z$ in row-by-row streaming fashion:
\begin{lem}[Properties of \Cref{alg:streamingmatrixpower}]\label{lem:algstreamingmatrixpower}
    \Cref{alg:streamingmatrixpower} taskes as input a stream $Z_{0,\cdot},Z_{1,\cdot}, \cdots, Z_{\nd-1,\cdot} \in \mathbb{R}^{1 \times \mdim}$ and outputs a stream $\widehat{Z}_{0,\cdot},\widehat{Z}_{1,\cdot}, \cdots, \widehat{Z}_{\nd-1,\cdot} \in \mathbb{R}^{1 \times \mdim}$. 
    It takes parameters $u,v \in \mathbb{R}^{d \times 1}$ and $W \in \mathbb{R}^{d \times d}$ and $t \in \mathbb{R}$.
    At each iteration $k$, output satisfies 
    \begin{equation}
        \widehat{Z}_{k,\cdot} = \sum_{j=0}^{k} r_{k-j}  \cdot Z_{j,\cdot}  = t Z_{k,\cdot} + \sum_{j=0}^{k} u^T W^{k-j} v \cdot Z_{j,\cdot} , \label{eq:algstreamingmatrixpower}
    \end{equation} 
    where $r_k = u\tp W^k v + t \mathbb{I}[k=0]$.
    The space usage is $O(d \mdim + d^2)$ and the runtime per iteration is dominated by a matrix multiplication $W \cdot S$, where $S \in \mathbb{R}^{d \times \mdim}$.
\end{lem}
\begin{proof}
    By induction, for all $k$, we have
    \[
        S_{k+1} = \sum_{j=0}^k W^{k-j} v Z_{j,\cdot} 
    \]
    and, hence,
    \[
        t Z_{k,\cdot} + u^T S_{k+1} = t Z_{k,\cdot} + \sum_{j=0}^k u^T W^{k-j} v Z_{j,\cdot} = \sum_{j=0}^k r_{k-j} Z_{j,\cdot} = \widehat{Z}_{k,\cdot},
    \]
    where $\widehat{Z} = M(r,\nd) Z$.
    
    The state of the algorithm after iteration $k$ is given by $S_{k+1} \in \mathbb{R}^{d \times \mdim}$.
    The space usage of the algorithm is dominated by storing this state ($d \mdim$ registers) plus storing the parameters $u,v,W$ ($2d+d^2$ registers).
    More precisely, the only space that the algorithm requires is the state, plus whatever registers are required to perform the update $S_{k+1} = v Z_{k,\cdot} + W \cdot S_{k}$. (The output computation $\widehat{Z}_{k,\cdot} = t Z_{k,\cdot} + u^T S_{k+1}$ requires $\mdim$ registers to store the output.)
    The exact number of registers required for this update depends on the structure of $v$ and $W$. 
    If $W$ is a diagonal matrix or has the structure given in \cref{eq:companionmatrix}, then the update can be performed ``in place'' and we only require $\mdim$ extra registers to store the input $Z_{k,\cdot}$.
    In general, we can always perform this update using $d\mdim$ additional registers.
    
    The running time of one iteration consists of the update $S_{k+1} = v Z_{k,\cdot} + W \cdot S_{k}$ and the matrix-vector product $\widehat{Z}_{k,\cdot} = u^T S_{k+1}$.
    The matrix-matrix product will dominate over the matrix-vector products.
    The runtime of the matrix-matrix product depends on the structure of $W$. If $W$ is sparse (e.g., if it is diagonal), then this is $O(d\mdim)$ time.
    In general, we can perform this update in $O(d^2\mdim)$ time.
\end{proof}

\section{Factorizations via Rational Function Approximation}
\label{sec:rationalfunctionapprox}

In this section we prove \Cref{thm:main-inf}.
We follow the generating functions view presented in \Cref{sec:genfunc-framework} and use the corresponding sampling algorithm (\cref{alg:streamingmatrixpower}).
The proof is split into three steps. 
In \Cref{sec:genfuncmat-approx} we recap how a generating function $r$ yields a valid matrix factorization
 and prove that, if $r(x) \approx \sqrt{1-x}$, then this yields an approximately optimal factorization.
In \Cref{sec:rationalsqrt} we construct a low-degree rational function $r$ that appropriately approximates the square root.
Finally we assemble the parts of the proof in \Cref{sec:proof-main}.

\subsection{Reduction to Approximating the Square Root}\label{sec:genfuncmat-approx}

\Cref{lem:lbstruct} shows that the optimal matrix factorization is given by $\bfA^{(\nd)}=\bfB\bfC$ where $\bfB=\bfC=M(f,\nd)$ and $f(x)=1/\sqrt{1-x}$.
Suppose we have a generating function $r(x) \approx \sqrt{1-x}$. 
Let $b(x) \coloneqq \frac{r(x)}{1-x} \approx \frac{1}{\sqrt{1-x}} = f(x)$ and $c(x) \coloneqq \frac{1}{r(x)} \approx \frac{1}{\sqrt{1-x}}$. 
Since $ b(x) \cdot c(x) = \frac{r(x)}{1-x} \cdot \frac{1}{r(x)} = \frac{1}{1-x} = g(x)$, we always obtain a valid factorization $M(b,\nd) \cdot M(c,\nd) = \bfA^{(\nd)}$.
We can bound the objective value of this factorization in terms of the approximation $r(x) \approx \sqrt{1-x}$.

\begin{prop}[Approximating the square root approximates the matrix factorization objective]\label{prop:parseval}
    Let $\tau > 0$ and $\nd \in \mathbb{N}$.
    Let $r : C \to \mathbb{C}$, where $C \coloneqq \{ z \in \mathbb{C} : |z|<1 \}$ is the open unit disc in the complex plane centered at zero.
    Let \[\gamma_\tau \coloneqq \sup \left\{|r(x)-\sqrt{1-x}| : x \in \mathbb{C}, |x|=\exp(-\tau)\right\}.\]
    Let $b(x) \coloneqq \frac{r(x)}{1-x}$ and $c(x) \coloneqq \frac{1}{r(x)}$ and $f(x) = \frac{1}{\sqrt{1-x}}$.
    Then $M(b,\nd) \cdot M(c,\nd) = M(c,\nd) \cdot M(b,\nd) = A^{(\nd)} = M(f,\nd)^2$, where $M(\cdot,\nd)$ is defined in Equation \ref{eq:genfuncmatrix} and $A^{(\nd)}$ is defined in Equation \ref{eq:allonesM}.
    Moreover, if $\gamma_\tau \le \left(\frac{1-\exp(-2\tau)}{4}\right)^2$, then
    \begin{align*}
    \|M(b,\nd)\|_{2 \to \infty} &\le \|M(f,\nd)\|_{2 \to \infty} +  \frac{\exp(\tau \nd) \cdot \gamma_\tau}{\sqrt{\exp(2\tau)-1} },\\
    \|M(c,\nd)\|_{1 \to 2} &\le \|M(f,\nd)\|_{1 \to 2} + \frac{\exp(\tau \nd) \cdot \gamma_\tau}{\sqrt[4]{(\exp(2\tau)-1)^2 - 2^{7/2} \gamma_\tau \exp(4\tau)}}.
    \end{align*}
\end{prop}
We will end up setting $\tau=\frac{1}{2\nd}$. This ensures that $\frac{\exp(\tau \nd) \cdot \gamma_\tau}{\sqrt{\exp(2\tau)-1}} \le \sqrt{\nd} \cdot \exp(1/2) \cdot \gamma_\tau < 2 \sqrt{n} \gamma_\tau$ and, if $\nd \ge 6$ and $\gamma_\tau \le \frac{1}{32 \cdot \nd^2}$, then $\frac{\exp(\tau \nd) \cdot \gamma_\tau}{\sqrt[4]{(\exp(2\tau)-1)^2 - 2^{7/2} \gamma_\tau \exp(4\tau)}} \le \sqrt[4]{2} \cdot \sqrt{\nd} \cdot \exp(1/2) \cdot \gamma_\tau < 2 \sqrt{n} \gamma_\tau$.
\begin{proof}
Let $b(x) = \sum_{k=0}^\infty b_k x^k$, $c(x) = \sum_{k=0}^\infty c_k x^k$, and $f(x) = \sum_{k=0}^\infty f_k x^k$.
By linearity and the triangle inequality,
\begin{align*}
    \|M(b,\nd)\|_{2 \to \infty} &= \|M(f,\nd) + M(b-f,\nd)\|_{2 \to \infty} \le \|M(f,\nd)\|_{2 \to \infty} + \|M(b-f,\nd)\|_{2 \to \infty} , \\
    \|M(c,\nd)\|_{1 \to 2} &= \|M(f,\nd) + M(c-f,\nd)\|_{1 \to 2} \le \|M(f,\nd)\|_{1 \to 2} + \|M(c-f,\nd)\|_{1 \to 2} .
\end{align*}
Thus it suffices to bound $\|M(b-f,\nd)\|_{2 \to \infty} = \sqrt{\sum_{k=0}^{\nd-1} (b_k-f_k)^2}$ and $\|M(c-f,\nd)\|_{1 \to 2} = \sqrt{\sum_{k=0}^{\nd-1} (c_k-f_k)^2}$.
We bound these using a weighted version of Parseval's identity.

\begin{lem}[Weighted Parseval's Identity]
    Let $g,h : C \to \mathbb{C}$ be analytic, where $C \coloneqq \{ z \in \mathbb{C} : |z|<1 \}$ is the open unit disc in the complex plane centered at zero.
    Let $g(x) = \sum_{k=0}^\infty g_k x^k$ and $h(x) = \sum_{k=0}^\infty h_k x^k$ for $x \in C$, where $g_k,h_k \in \mathbb{C}$.
    Let $\tau > 0$. Then
    \begin{align*}
        \frac{1}{2\pi} \int_{-\pi}^{\pi} |g(\exp(\sqrt{-1}\theta-\tau)) - h(\exp(\sqrt{-1}\theta-\tau))|^2 \mathrm{d}\theta
        &= \sum_{k=0}^\infty |g_k-h_k|^2 \cdot \exp(-2 \tau k) \\
        &\ge \exp(-2 \tau (\nd-1)) \sum_{k=0}^{\nd-1} |g_k-h_k|^2 .
    \end{align*}
\end{lem}
\begin{proof}
    Let $w(x)=g(x)-h(x) = \sum_{k=0}^\infty w_k x^k$, where $w_k=g_k-h_k$. Then 
    \begin{align*}
         &\frac{1}{2\pi} \int_{-\pi}^{\pi} |g(\exp(\sqrt{-1}\theta-\tau)) - h(\exp(\sqrt{-1}\theta-\tau))|^2 \mathrm{d}\theta \\
         &= \frac{1}{2\pi} \int_{-\pi}^\pi |w(\exp(\sqrt{-1}\theta-\tau))|^2 \mathrm{d}\theta \\
         &= \frac{1}{2\pi} \int_{-\pi}^\pi w(\exp(\sqrt{-1}\theta-\tau)) \cdot \overline{w(\exp(\sqrt{-1}\theta-\tau))} \mathrm{d}\theta \\
         &= \frac{1}{2\pi} \int_{-\pi}^\pi \left( \sum_{k=0}^\infty w_k \cdot (\exp(\sqrt{-1}\theta-\tau))^k \right) \cdot \overline{\left( \sum_{\ell=0}^\infty w_\ell \cdot (\exp(\sqrt{-1}\theta-\tau))^\ell \right)} \mathrm{d}\theta \\
         &= \sum_{k,\ell=0}^\infty \frac{1}{2\pi} \int_{-\pi}^\pi  w_k \cdot (\exp(\sqrt{-1}\theta-\tau))^k \cdot \overline{w_\ell} \cdot (\exp(-\sqrt{-1}\theta-\tau))^\ell \mathrm{d}\theta \\
         &= \sum_{k,\ell=0}^\infty w_k \cdot \overline{w_\ell} \cdot \exp(-(k+\ell)\tau) \cdot \frac{1}{2\pi} \int_{-\pi}^\pi   \exp((k-\ell)\sqrt{-1}\theta) \mathrm{d}\theta \\
         &= \sum_{k,\ell=0}^\infty w_k \cdot \overline{w_\ell} \cdot \exp(-(k+\ell)\tau) \cdot \mathbb{I}[k-\ell=0] \\
         &= \sum_{k=0}^\infty |w_k|^2 \cdot \exp(-2k\tau) \\
         &= \sum_{k=0}^\infty |g_k-h_k|^2 \cdot \exp(-2\tau k).
    \end{align*}
\end{proof}

Recall $b(x) \coloneqq \frac{r(x)}{1-x}$ and $f(x) = \frac{1}{\sqrt{1-x}}$, whence $b(x)-f(x) = \frac{r(x)-\sqrt{1-x}}{1-x}$.
Now we have 
\begin{align*}
    \|M(b-f,\nd)\|_{2 \to \infty} &= \sqrt{\sum_{k=0}^{\nd-1} (b_k-f_k)^2} \\
    &\le \exp(\tau (\nd-1)) \sqrt{ \frac{1}{2\pi} \int_{-\pi}^{\pi} |b(\exp(\sqrt{-1}\theta-\tau)) - f(\exp(\sqrt{-1}\theta-\tau))|^2 \mathrm{d}\theta} \\
    &= \exp(\tau (\nd-1)) \sqrt{ \frac{1}{2\pi} \int_{-\pi}^{\pi} \frac{\left|r(\exp(\sqrt{-1}\theta-\tau)) - \sqrt{1-\exp(\sqrt{-1}\theta-\tau))}\right|^2}{\left|1-\exp(\sqrt{-1}\theta-\tau))\right|^2} \mathrm{d}\theta} \\
    &\le \exp(\tau (\nd-1)) \sqrt{ \frac{1}{2\pi} \int_{-\pi}^{\pi} \frac{\gamma_\tau^2}{\left|1-\exp(\sqrt{-1}\theta-\tau))\right|^2} \mathrm{d}\theta} \\
    &= \exp(\tau (\nd-1)) \cdot \gamma_\tau \cdot \sqrt{ \frac{\exp(\tau)}{4\pi} \int_{-\pi}^{\pi} \frac{1}{\cosh(\tau) + \cos(\theta)} \mathrm{d}\theta} \\
    &= \exp(\tau (\nd-1)) \cdot \gamma_\tau \cdot \sqrt{ \frac{\exp(\tau)}{2} \frac{1}{\sqrt{\cosh(\tau)^2-1}} } \\
    &= \frac{\exp(\tau (\nd-1)) \cdot \gamma_\tau}{\sqrt{ 1-\exp(-2\tau) }} = \frac{\exp(\tau \nd) \cdot \gamma_\tau}{\sqrt{ \exp(2\tau) -1 }} \le \frac{\exp(\tau \nd) \cdot \gamma_\tau}{\sqrt{ 2\tau }}.
\end{align*}
In the above we use the identity
\begin{align*}
    |1-\exp(\sqrt{-1}\theta-\tau)|^2 &= (1-e^{-\tau} \cos \theta)^2 + (e^{-\tau} \sin \theta)^2 = 1 - 2 e^{-\tau} \cos \theta + e^{-2\tau} (\cos^2 \theta + \sin^2 \theta) \\& = e^{-\tau} \cdot \left( e^\tau - 2 \cos \theta + e^{-\tau} \right) = 2 e^{-\tau} \cdot \left( \cosh \tau + \cos \theta \right)
\end{align*}
and the integral 
\begin{equation}
    \forall \sigma>1 ~~~~~ \frac{1}{2\pi} \int_{-\pi}^{\pi} \frac{1}{\sigma+\cos(\theta)} \mathrm{d} \theta = \frac{1}{2\pi} \int_{-\pi}^{\pi} \frac{1}{\sigma+\sin(\theta)} \mathrm{d} \theta = \frac{1}{\sqrt{\sigma^2-1}},
\end{equation}
which follows from the derivative $\frac{\mathrm{d}}{\mathrm{d}\theta} \left[ \frac{2}{\sqrt{\sigma^2-1}} \tan^{-1}\left(\frac{1 + \sigma \tan (\theta/2)}{\sqrt{\sigma^2-1}}\right) \right] = \frac{1}{\sigma+\sin(\theta)}$.

Similarly, recall $c(x) \coloneqq \frac{1}{r(x)}$ and $f(x) = \frac{1}{\sqrt{1-x}}$, whence $c(x) - f(x) = \frac{\sqrt{1-x}-r(x)}{r(x) \cdot \sqrt{1-x}}$.
Assuming $\gamma_\tau < \frac{(1-\exp(-2\tau))^2}{2^{7/2}}$, we have
\begin{small}
\begin{align*}
    &\|M(c-f,\nd)\|_{1 \to 2} \\
    &= \sqrt{\sum_{k=0}^{\nd-1} (c_k-f_k)^2} \\
    &\le \exp(\tau (\nd-1)) \sqrt{ \frac{1}{2\pi} \int_{-\pi}^{\pi} |c(\exp(\sqrt{-1}\theta-\tau)) - f(\exp(\sqrt{-1}\theta-\tau))|^2 \mathrm{d}\theta} \\
    &= \exp(\tau (\nd\!-\!1)) \!\! \sqrt{ \!\frac{1}{2\pi} \!\int_{-\pi}^{\pi}\! \frac{\left|r(\exp(\sqrt{-1}\theta-\tau)) - \sqrt{1-\exp(\sqrt{-1}\theta-\tau)}\right|^2}{\left| r(\exp(\sqrt{-1}\theta\!-\!\tau))  \!\cdot\! \sqrt{1\!-\!\exp(\sqrt{-1}\theta\!-\!\tau)} \right|^2}\! \mathrm{d}\theta} \\
    &= \exp(\tau (\nd\!-\!1)) \!\! \sqrt{ \!\frac{1}{2\pi} \!\int_{-\pi}^{\pi}\! \frac{\left|r(\exp(\sqrt{-1}\theta-\tau)) - \sqrt{1-\exp(\sqrt{-1}\theta-\tau)}\right|^2}{\left| 1\!-\!\exp(\!\sqrt{-1}\theta\!-\!\tau\!)\! +\! \left(r(\exp(\sqrt{-1}\theta\!-\!\tau)) \!-\! \sqrt{1\!-\!\exp(\sqrt{-1}\theta\!-\!\tau)}\right) \!\cdot\! \sqrt{1\!-\!\exp(\sqrt{-1}\theta\!-\!\tau)} \right|^2}\! \mathrm{d}\theta} \\
    &\le \exp(\tau (\nd\!-\!1)) \!\! \sqrt{ \!\frac{1}{2\pi} \!\int_{-\pi}^{\pi}\! \frac{\gamma_\tau^2}{\left(\left| 1\!-\!\exp(\!\sqrt{-1}\theta\!-\!\tau\!)\right|-\gamma_\tau \cdot \sqrt{|1-\exp(\sqrt{-1}\theta-\tau)|}\right)^2} \mathrm{d}\theta} \\
    &\le \exp(\tau (\nd\!-\!1)) \!\! \sqrt{ \!\frac{1}{2\pi} \!\int_{-\pi}^{\pi}\! \frac{\gamma_\tau^2}{\left| 1\!-\!\exp(\!\sqrt{-1}\theta\!-\!\tau\!)\right|^2 -2\sqrt{2} \gamma_\tau } \mathrm{d}\theta} \\
    &= \exp(\tau (\nd\!-\!1)) \! \cdot \! \gamma_\tau \!\cdot\! \sqrt{ \!\frac{\exp(\tau)}{4\pi} \!\int_{-\pi}^{\pi}\! \frac{1}{\cosh(\tau) - \sqrt{2} \exp(\tau) \gamma_\tau + \cos(\theta) } \mathrm{d}\theta} \\
    &= \exp(\tau (\nd\!-\!1)) \! \cdot \! \gamma_\tau \!\cdot\! \sqrt{ \!\frac{\exp(\tau)}{2} \frac{1}{\sqrt{(\cosh(\tau) - \sqrt{2} \exp(\tau) \gamma_\tau )^2-1}} } \\
    &= \frac{\exp(\tau \nd) \cdot \gamma_\tau}{\sqrt{2\exp(\tau)\sqrt{(\cosh(\tau)-\sqrt{2}\exp(\tau)\gamma_\tau)^2-1}}} \\
    &= \frac{\exp(\tau \nd) \cdot \gamma_\tau}{\sqrt[4]{{(\exp(2\tau)(1-2\sqrt{2}\gamma_\tau) + 1)^2-4\exp(2\tau)}}} \\
    &= \frac{\exp(\tau \nd) \cdot \gamma_\tau}{\sqrt[4]{1 -2(1+2^{3/2}\gamma_\tau)\exp(2\tau) + (1-2^{3/2}\gamma_\tau)^2\exp(4\tau)}} \\
    &= \frac{\exp(\tau \nd) \cdot \gamma_\tau}{\sqrt[4]{(\exp(2\tau)-1)^2 - 2^{5/2} \gamma_\tau \exp(2\tau) - 2^{5/2} \gamma_\tau \exp(4\tau) + 2^3 \gamma_\tau^2 \exp(4\tau)}} \\
    &\le \frac{\exp(\tau \nd) \cdot \gamma_\tau}{\sqrt[4]{(\exp(2\tau)-1)^2 - 2^{5/2} \gamma_\tau \exp(2\tau) ( 1 + \exp(2\tau) ) }} \\
    &\le \frac{\exp(\tau \nd) \cdot \gamma_\tau}{\sqrt[4]{(\exp(2\tau)-1)^2 - 2^{7/2} \gamma_\tau \exp(4\tau) }} .
\end{align*}
\end{small}
Combining the bounds yields the result.
\end{proof}

\subsection{Rational Approximation of $\sqrt{x}$} \label{sec:rationalsqrt}

\Cref{sec:genfunc-framework,sec:genfuncmat-approx} reduce the problem to finding a low-degree rational function that uniformly approximates $\sqrt{1-x}$ for $x \in \mathbb{C}$ with $|x|\le 1$.
Fortunately, such functions are known to exist.
In this subsection we construct the required rational functions. The proof is included for completeness and because it provides an explicit construction with explicit bounds.

For notational simplicity and consistency with the literature, we rescale so that we are approximating $\sqrt{x}$ instead of $\sqrt{1-x}$. To get an equivalent result, we must uniformly approximate $\sqrt{x}$ for $x \in \mathbb{C}$ with $|x-1/2|\le1/2$.

\citet{Newman64} gave an explicit rational function that approximates $\sqrt{x}$ for $x \in [0,1]$. Specifically, let 
\begin{equation}
    p(x) = \prod_{k=0}^{d-1} (x + \exp(-k/\sqrt{d})) ~~~\text{ and }~~~ r(x) = \frac{\sqrt{x} \cdot (p(\sqrt{x})-p(-\sqrt{x}))}{p(\sqrt{x})+p(-\sqrt{x})}.
\end{equation}
Then $r(x)$ is a rational function of degree $\lceil d/2 \rceil$ -- specifically, $r(x)$ is the ratio where the numerator is a polynomial in $x$ of degree $\lceil d/2 \rceil$ and the denominator is a polynomial in $x$ of degree $\lfloor d/2 \rfloor$. Furthermore, 
\begin{equation}
    \sup_{x \in [0,1]} |r(x)-\sqrt{x}| \le 3 \cdot \exp(-\sqrt{d}). \label{eq:newman}
\end{equation} 
The surprising aspect of this result is that the dependence on the degree is exponential rather than polynomial, as is the case if we restrict to polynomial approximations, rather than rational approximations.
Newman also showed that this result was optimal up to constant factors in the degree $d$ and the approximation error.
Subsequent work \cite{stahl1993best} obtained optimal constants for this approximation.

In order to apply \Cref{prop:parseval}, we need to extend Newman's result into the complex plane. We provide a proof adapted from that of \citet{GopalT19}.

\begin{thm}\label{thm:sqrtapprox}
    Let $d \in \mathbb{N}$ with $d \ge 2$. Then there exists a rational function $r$ of degree $d$ with real coefficients and real negative simple poles such that
    \begin{equation}
        |r(x) - \sqrt{x}| \le \big(4 + \frac4\pi \big) \cdot \exp\left(-\pi \sqrt{\frac12 \left\lfloor \frac{d-1}{2}\right\rfloor}\right) \le 6 \cdot \exp\left(-\frac{\pi}{2} \sqrt{d-2} \right)
    \end{equation}
    for all $x \in \mathbb{C}$ with $\Re(x) \ge 0$ and $|x| \le 1$.
    Specifically,
    \begin{equation}
        r(x) = \frac{2h}{\pi} \sum_{k=-d_-}^{d_+} \frac{x \cdot \exp(hk)}{x+\exp(2hk)}, ~~\text{ where }~~ d_+ = \left\lfloor\frac{d-1}{2}\right\rfloor,~~d_- = \left\lceil\frac{d-1}{2}\right\rceil,~~\text{ and }~~ h = \frac{\pi}{\sqrt{2 d_+}}.
    \end{equation}
\end{thm}
\begin{proof}
    We invoke \Cref{prop:sqrt-approx-complex}. Let $d_- = \lceil(d-1)/2\rceil$ and $d_+ = \lfloor(d-1)/2\rfloor$ and let $h \in (0,4.9)$ be determined later. Let $r(x) = r_{d_+,d_-,h}(x)$ be as in Equation \ref{eq:rational-approx-messy}. Clearly $r(x)$ is a rational function of degree $d_++d_-+1 = d$ with real coefficients and the poles are given by $-\exp(2hk) \in (-\infty,0)$ for $k\in\{-d_-,-d_-+1,\cdots,d_+\}$. It remains to simplify the approximation guarantee and set $h$.
    From Equation \ref{eq:rational-approx-messy}, for all $x \in \mathbb{C}$ with $\Re(x) \ge 0$ and $|x| \le 1$, we have
    \begin{align}
        |r(x)-\sqrt{x}| &\le 2 \sqrt{|x|}  \left( \frac{1}{\exp((1-c)\pi^2 / h) - 1} + \frac{1}{\exp((1+c)\pi^2 / h)-1} \right) \notag\\&~~~~~~~~~~ + \frac{2h}{\pi (\exp(h)-1)} \big( |x| \cdot {\exp(-hd_+)} + {\exp(-hd_-)} \big) \notag \\
        &\le 2 \left( \frac{1}{\exp((1-|c|)\pi^2 / h) - 1} + \frac{1}{\exp((1+|c|)\pi^2 / h)-1} \right) \tag{$|x| \le 1$ \& symmetry in $c$}\\&~~~~~~~~~~ + \frac{4h \exp(-hd_+)}{\pi (\exp(h)-1)} \tag{$d_- \ge d_+$} \\
        &\le 2 \left( \frac{1}{\exp(\pi^2 / 2h) - 1} + \frac{1}{\exp(\pi^2 / h)-1} \right) + \frac{4}{\pi} \cdot \exp(-hd_+) \cdot \frac{h}{\exp(h)-1} \tag{$|c| = |\frac1\pi\arg(x)| \le \frac12$} \\
        &\le 4 \cdot \exp(-\pi^2/2h) + \frac{4}{\pi} \cdot \exp(-hd_+) \label{eq:simplification-stuff} \\
        &= \big(4 + \frac4\pi \big) \cdot \exp\left(-\pi \sqrt{\frac{d_+}{2}}\right)  \label{eq:substitute-h} \\
        &= \big(4 + \frac4\pi \big) \cdot \exp\left(-\pi \sqrt{\frac12 \left\lfloor \frac{d-1}{2}\right\rfloor}\right) \notag \\
        &\le 6 \cdot \exp\left(-\frac\pi2 \sqrt{d-1} \right). \notag
    \end{align}
    The fourth inequality \cref{eq:simplification-stuff} follows from two simple facts:  
    First, $h \ge 0$ implies $\frac{h}{\exp(h)-1} \le 1$.
    Second, $\frac{1}{a-1}+\frac{1}{a^2-1} \le \frac{2}{a}$ when $a = \exp(\pi^2/2h) \ge 1+\sqrt{3}$ or, equivalently, $h \le \frac{\pi^2}{2 \log(1+\sqrt{3})} \approx 4.91$.
    The penultimate equality \cref{eq:substitute-h} follows by setting $h = \frac{\pi}{\sqrt{2d_+}} \le \frac{\pi}{\sqrt{2}} < 2.5$.
\end{proof}

\begin{cor}[Rescaling of \Cref{thm:sqrtapprox}]\label{cor:sqrt1approx}
    Let $d \ge 2$ be an integer. Then there exists a rational function $\widetilde{r}$ of degree $d$ with real coefficients such that all poles are simple and $> 1$ and we have
    \begin{equation}
        |\widetilde{r}(x) - \sqrt{1-x}| \le 8  \cdot \exp\left(-\frac{\pi}{2} \sqrt{d-2} \right)
    \end{equation}
    for all $x \in \mathbb{C}$ with $|x| \le 1$.
    Specifically,
    \begin{equation}
        \widetilde{r}(x) = \frac{2h\sqrt{2}}{\pi} \sum_{k=-d_-}^{d_+} \frac{(1-x) \cdot \exp(hk)}{1-x+2\cdot\exp(2hk)} = \frac{2h\sqrt{2}}{\pi} \sum_{k=-d_-}^{d_+} \exp(hk) - \frac{2 \cdot \exp(3hk)}{1+2\cdot\exp(2hk)-x}
        \label{eq:explicitrationalsqrt}
    \end{equation}
    where $d_+ = \lfloor(d-1)/2\rfloor$, $d_- = \lceil(d-1)/2\rceil$, and $h = \frac{\pi}{\sqrt{2d_+}}$.
\end{cor}
\begin{proof}
    Let $r$ be the rational function of degree $\le d$ promised by \Cref{thm:sqrtapprox}.
    Define $\widetilde{r}(x) = \sqrt{2} \cdot r(\tfrac{1-x}{2})$.
    Suppose $x \in \mathbb{C}$ with $|x|\le 1$.
    Then $\Re(\tfrac{1-x}{2}) = \frac12 -\frac12 \Re(x) \ge 0$, $|\tfrac{1-x}{2}| \le \tfrac{1+|x|}{2} \le 1$, and \[|\widetilde{r}(x)-\sqrt{1-x}|  = \sqrt{2} \cdot \left| r(\tfrac{1-x}{2})-\sqrt{\tfrac{1-x}{2}} \right| \le \sqrt{2} \cdot \left(4 + \frac4\pi\right) \cdot \exp\left(-\pi \sqrt{\frac12 \left\lceil \frac{d-1}{2}\right\rceil}\right) \le 8 \exp\left(-\frac{\pi}{2} \sqrt{d-1} \right).\]
    If $\widetilde{x}$ is a pole of $\widetilde{r}$, then $\tfrac{1-\widetilde{x}}{2}$ is a pole of $r$; since the poles of $r$ are negative we have $\tfrac{1-\widetilde{x}}{2} < 0$ and, hence, $\widetilde{x}>1$.
\end{proof}

The general form of the approximation promised by \Cref{thm:sqrtapprox} is given below.

\begin{prop}\label{prop:sqrt-approx-complex}
    Let $d_+,d_- \in \mathbb{N}$ and $h>0$. For $x \in \mathbb{C} \setminus [-\exp(2d_+h),-\exp(-2d_-h)]$, define 
    \begin{equation}
        r_{d_+,d_-,h}(x) \coloneqq \frac{2h}{\pi} \sum_{k=-d_-}^{d_+} \frac{x \cdot \exp(hk)}{x+\exp(2hk)}.
    \end{equation}
    Then, for all $x \in \mathbb{C}$ with $\Re(x) \ge 0$, we have
    \begin{align}
        |r_{d_+,d_-,h}(x)-\sqrt{x}| &\le 2 \sqrt{|x|}  \left( \frac{1}{\exp((1-c)\pi^2 / h) - 1} + \frac{1}{\exp((1+c)\pi^2 / h)-1} \right) \notag\\&~~~~~ +  \frac{2h}{\pi (\exp(h)-1)} \big( |x| \cdot {\exp(-hd_+)} + {\exp(-hd_-)} \big), \label{eq:rational-approx-messy}
    \end{align}
    where $c = \frac1\pi \arg(x) \in \left[-\frac12,\frac12\right]$.
\end{prop}
\begin{proof}
    Fix $x \in \mathbb{C}$ with $\Re(x) > 0$.\footnote{For the proof we assume $\Re(x)>0$. The case $\Re(x)=0$ follows by continuity.}
    Let $\arg(x) = \tan^{-1}\left(\frac{\Im(x)}{\Re(x)}\right) = c\pi \in \left(-\frac\pi2,\frac\pi2\right)$.
    
    We begin with the Cauchy distribution integral 
    \[\int_0^\infty \frac{1}{1+u^2} \mathrm{d}u = \int_0^\infty \left( \frac{\mathrm{d}}{\mathrm{d}u} \tan^{-1}(u) \right) \mathrm{d}u = \frac{\pi}{2}.\]
    We perform the variable substitutions $u = v/\sqrt{x}$ and $v=\exp(s)$ and rearrange to obtain
    \begin{equation}
        \sqrt{x} = \frac{2}{\pi} \int_0^\infty \frac{x}{x+v^2} \mathrm{d}v = \frac{2}{\pi} \int_{-\infty}^\infty \frac{x \cdot \exp(s)}{x+\exp(2s)} \mathrm{d}s . \label{eq:sqrt-integral}
    \end{equation}
    Now we make the approximation
    \begin{equation}
        \sqrt{x} = \frac{2}{\pi} \int_{-\infty}^\infty \frac{x \cdot \exp(s)}{x+\exp(2s)} \mathrm{d}s \approx \frac{2h}{\pi} \sum_{k \in \mathbb{Z}} \frac{x \cdot \exp(hk)}{x+\exp(2hk)} \approx \frac{2h}{\pi} \sum_{k = -d_-}^{d_+} \frac{x \cdot \exp(hk)}{x+\exp(2hk)} = r_{d_+,d_-,h}(x).\label{eq:sqrt-riemann-approx}
    \end{equation}
    Thus all that remains to complete the proof is to make the the two approximations precise. 
    Intuitively, the first approximation replaces the continuous integral with a discrete-but-infinite Riemann sum. If the function is smooth and $h$ is small, this approximation should be good.
    The second approximation truncates the infinite sum, which is a good approximation as long as the tails of the function decay rapidly and $hd_+$ and $hd_-$ are large.
    
    Define $f(z) \coloneqq \frac{2h}{\pi} \frac{x \cdot \exp(hz)}{x+\exp(2hz)}$.
    Let
    \begin{align}
        \widehat{f}(t) &\coloneqq \int_{-\infty}^\infty f(z) \cdot \exp(-2\pi z t \sqrt{-1}) \mathrm{d}z \notag \\
        &= \int_{-\infty}^\infty \frac{2h}{\pi}\frac{x \cdot \exp(hz - 2\pi z t \sqrt{-1})}{x + \exp(2hz)} \mathrm{d}z \notag \\
        &= \frac{2h}{\pi}\int_{-\infty}^\infty \frac{\exp(- 2\pi z t \sqrt{-1})}{\exp(-hz) + \exp(hz-\log x)} \mathrm{d}z \notag \\
        &= \frac{2h}{\pi}\int_{-\infty}^\infty \frac{\exp(- 2\pi (y+\frac{\log x}{2h}) t \sqrt{-1})}{\exp(-hy-\frac12 \log x) + \exp(hy-\frac12\log x)} \mathrm{d}y \tag{$z = y+\frac{\log x}{2h}$} \\
        &= \sqrt{x} \cdot \frac{2h}{\pi} \cdot \exp\left(-\pi t \frac{\log x}{h} \sqrt{-1}\right) \cdot \int_{-\infty}^\infty \frac{\exp(- 2\pi y t \sqrt{-1})}{\exp(-hy) + \exp(hy)} \mathrm{d}y \notag \\
        &=  \sqrt{x} \cdot \frac{2h}{\pi} \cdot \exp\left(-\pi t \frac{\log x}{h} \sqrt{-1}\right) \cdot \frac12 \int_{-\infty}^\infty \frac{\exp(- 2\pi y t \sqrt{-1})}{\cosh(hy)} \mathrm{d}y \notag \\
        &=  \sqrt{x} \cdot \frac{2h}{\pi} \cdot \exp\left(-\pi t \frac{\log x}{h} \sqrt{-1}\right) \cdot \frac{\pi}{2h \cdot \cosh(\pi^2 t / h)}  \label{eq:fourier} \\
        &=  \sqrt{x} \cdot \exp\left(-\pi t \frac{\log x}{h} \sqrt{-1}\right) \cdot \frac{1}{\cosh(\pi^2 t / h)} . \notag
    \end{align}
    The penultimate equality \eqref{eq:fourier} is the Fourier transform of the hyperbolic secant \cite{cosh_fourier}.
    The Poisson summation formula states that \begin{equation} \sum_{k \in \mathbb{Z}} f(k) = \sum_{t \in \mathbb{Z}} \widehat{f}(t).\end{equation}
    Per Equation \ref{eq:sqrt-integral}, we have \[\widehat{f}(0) = \int_{-\infty}^\infty f(z) \mathrm{d}z = \int_{-\infty}^\infty \frac{2h}{\pi} \frac{x \cdot \exp(hz)}{x+\exp(2hz)} \mathrm{d}z = \int_{-\infty}^\infty \frac{2}{\pi} \frac{x \cdot \exp(s)}{x+\exp(2s)} \mathrm{d}s = \sqrt{x} .\]
    Thus the Poisson summation formula allows us to analyze the first approximation in Equation \ref{eq:sqrt-riemann-approx}.
    Namely,
    \begin{align}
        &\left| \sqrt{x} - \frac{2h}{\pi} \sum_{k \in \mathbb{Z}} \frac{x \cdot \exp(hk)}{x+\exp(2hk)} \right|\notag\\ 
        &= \left| \widehat{f}(0) - \sum_{k \in \mathbb{Z}} f(k) \right| \notag \\
        &=\left| \widehat{f}(0) - \sum_{t \in \mathbb{Z}} \widehat{f}(t) \right| \notag \\
        &\le \sum_{t=1}^\infty |\widehat{f}(t)| +  |\widehat{f}(-t)| \notag \\
        &= \sum_{t=1}^\infty |\sqrt{x}| \cdot \left( \left|\exp\left(-\pi t \frac{\log x}{h} \sqrt{-1}\right)\right| + \left|\exp\left(\pi t \frac{\log x}{h} \sqrt{-1}\right)\right| \right) \cdot \frac{1}{\cosh(\pi^2 t / h)} \notag \\
        &= \sum_{t=1}^\infty |\sqrt{x}| \cdot \left( \exp\left(\Im\left(\pi t \frac{\log x}{h}\right)\right) + \exp\left(\Im\left(-\pi t \frac{\log x}{h} \right)\right) \right) \cdot \frac{1}{\cosh(\pi^2 t / h)} \tag{$|\exp(v\sqrt{-1})|=\exp(\Re(v\sqrt{-1}))=\exp(\Im(-v))$} \\
        &= \sqrt{|x|}  \sum_{t=1}^\infty 2\cosh\left(\pi \frac{t}{h} \Im(\log x) \right) \cdot \frac{1}{ \cosh(\pi^2 t / h)} \notag \\
        &= \sqrt{|x|}  \sum_{t=1}^\infty 2\cosh\left(\pi \frac{t}{h} \cdot c\pi \right) \cdot \frac{1}{ \cosh(\pi^2 t / h)} \tag{$\Im(\log x) = \arg(x) = c\pi$} \\
        &= 2\sqrt{|x|}  \sum_{t=1}^\infty \frac{\exp(c \pi^2 t / h) + \exp(-c \pi^2 t / h)}{ \exp(\pi^2 t / h) + \exp(-\pi^2 t / h)} \notag \\
        &\le  2\sqrt{|x|}  \sum_{t=1}^\infty \frac{\exp(c\pi^2 t / h) + \exp(-c\pi^2 t / h)}{ \exp(\pi^2 t / h) + 0} \notag \\
        &= 2\sqrt{|x|}  \sum_{t=1}^\infty \exp(-(1-c)\pi^2 t / h) + \exp(-(1+c)\pi^2 t / h) \notag \\
        &= 2\sqrt{|x|}  \left( \frac{1}{\exp((1-c)\pi^2 / h) - 1} + \frac{1}{\exp((1+c)\pi^2 / h)-1} \right) .\label{eq:geometric-series}
    \end{align}
    The final equality \eqref{eq:geometric-series} follows from the usual geometric series formula: If $a>0$, then $\sum_{t=1}^\infty \exp(-at) = \frac{\exp(-a)}{1-\exp(-a)} = \frac{1}{\exp(a)-1}$.
    
    Next we analyze the second approximation in Equation \ref{eq:sqrt-riemann-approx}, which requires bounding the tails of $f$.
    On one side, for $k \in \mathbb{N}$, we have 
    \begin{align}
        |f(k)| &= \left|\frac{2h}{\pi} \frac{x \cdot \exp(hk)}{x+\exp(2hk)}\right| \notag \\
        &= \frac{2h}{\pi} \frac{|x|}{|x \cdot \exp(-hk)+\exp(hk)|} \notag \\
        &= \frac{2h}{\pi} \frac{|x|}{\sqrt{\Re(x \cdot \exp(-hk)+\exp(hk))^2+\Im(x \cdot \exp(-hk)+\exp(hk))^2}} \notag \\
        &\le \frac{2h}{\pi} \frac{|x|}{|\Re(x \cdot \exp(-hk)+\exp(hk)|} \notag \\
        &= \frac{2h}{\pi} \frac{|x|}{|\Re(x) \cdot \exp(-hk)+\exp(hk)|} \tag{$hk \in \mathbb{R}$} \\
        &= \frac{2h}{\pi} \frac{|x|}{\Re(x) \cdot \exp(-hk)+\exp(hk)} \tag{$\Re(x)>0$} \\
        &\le \frac{2h}{\pi} \frac{|x|}{\exp(hk)}.\label{eq:fbound-positive}
    \end{align}
    On the other side, for $k \in \mathbb{N}$, we have
    \begin{align}
        |f(-k)| &= \left|\frac{2h}{\pi} \frac{x \cdot \exp(-hk)}{x+\exp(-2hk)}\right| \notag \\
        &= \frac{2h}{\pi} \frac{\exp(-hk)}{|1+\exp(-2hk)/x|} \notag \\
        &\le \frac{2h}{\pi} \frac{\exp(-hk)}{|\Re(1+\exp(-2hk)/x)|} \notag \\
        &= \frac{2h}{\pi} \frac{\exp(-hk)}{|1+\exp(-2hk)/\Re(x)|} \tag{$\Re(1/x)=1/\Re(x)$} \\
        &= \frac{2h}{\pi} \frac{\exp(-hk)}{1+\exp(-2hk)/\Re(x)} \tag{$\Re(x)>0$} \\
        &\le\frac{2h}{\pi}  \exp(-hk).\label{eq:fbound-negative}
    \end{align}
    Combining Equations \ref{eq:fbound-positive} and \ref{eq:fbound-negative}, we have
    \begin{align}
        \left|\frac{2h}{\pi} \sum_{k \in \mathbb{Z}} \frac{x \cdot \exp(hk)}{x+\exp(2hk)} - r_{d_+,d_-,h}(x)\right|
        &=\frac{2h}{\pi} \left|\sum_{k \in \mathbb{Z}} \frac{x \cdot \exp(hk)}{x+\exp(2hk)} - \sum_{k = -d_-}^{d_+} \frac{x \cdot \exp(hk)}{x+\exp(2hk)}\right| \notag \\ 
        &= \left|\sum_{k=d_++1}^\infty f(k) + \sum_{k=d_-+1}^\infty f(-k) \right| \notag \\
        &\le \sum_{k=d_++1}^\infty |f(k)| + \sum_{k=d_-+1}^\infty |f(-k)| \notag \\
        &\le  \frac{2h}{\pi} \sum_{k=d_++1}^\infty \frac{|x|}{\exp(hk)} + \frac{2h}{\pi} \sum_{k=d_-+1}^\infty \exp(-hk) \notag \\
        &= \frac{2h}{\pi} |x| \frac{\exp(-h(d_++1))}{1-\exp(-h)} + \frac{2h}{\pi}   \frac{\exp(-h(d_-+1))}{1-\exp(-h)} \\
        &= \frac{2h}{\pi (\exp(h)-1)} \big( |x| \cdot {\exp(-hd_+)} + {\exp(-hd_-)} \big). \label{eq:truncation-bound}
    \end{align}
    Combining Equations \ref{eq:geometric-series} and \ref{eq:truncation-bound} yields the result in Equation \ref{eq:rational-approx-messy}.
\end{proof}

\subsection{Putting Everything Together} \label{sec:proof-main}

Now we can assemble the tools we have developed to prove our main result (\Cref{thm:main-inf}).
In \Cref{sec:genfuncmat-approx} we connected near-optimal lower triangular Toeplitz matrix factorizations to approximating the square root function.
In \Cref{sec:genfuncmat,sec:algorithm} we connected low-degree rational functions to efficient streaming algorithms.
In \Cref{sec:rationalsqrt} we gave a low-degree rational function that approximates the square root function.
The combination of these three tools gives a near-optimal matrix factorization with a corresponding efficient streaming algorithm.

\begin{thm}[Main Theorem - Formal version of \Cref{thm:main-inf}]\label{thm:main-f}
    Let $\nd, d, \mdim \in \mathbb{N}$ satisfy $n \ge 5$ and $d \ge 2+\left(\frac{12+4\log n}{\pi}\right)^2$.
    There exists a lower triangular Toeplitz matrix factorization $B, C, A=BC \in \mathbb{R}^{\nd \times \nd}$ with the following properties.
    \begin{itemize}
        \item\textbf{Near-Optimality:} 
        Let $B^*=C^* = M(f,\nd)$ for $f(x)=1/\sqrt{1-x}$ be the optimal lower triangular Toeplitz matrix factorization of size $\nd$. Then
        \begin{align*}
            \|B\|_{2 \to \infty} &\le \|B^*\|_{2 \to \infty} + 16 \sqrt{n} \exp\left(-\frac\pi2 \sqrt{d-2}\right),\\
            \|C\|_{1 \to 2} &\le \|C^*\|_{1 \to 2} + 16 \sqrt{n} \exp\left(-\frac\pi2 \sqrt{d-2}\right).
        \end{align*}
        \item\textbf{Efficiency:}
        There is an algorithm that takes as input a stream $Z_{1,\cdot},Z_{2,\cdot},\cdots,Z_{\nd,\cdot} \in \mathbb{R}^{1 \times \mdim}$ and outputs a stream $\widehat{Z}_{1,\cdot},\widehat{Z}_{2,\cdot},\cdots,\widehat{Z}_{\nd,\cdot} \in \mathbb{R}^{1 \times \mdim}$ such that $\widehat{Z} = C^{-1} Z$ and the space and time (per iteration) required by the algorithm is $O(d\mdim)$.
    \end{itemize}
\end{thm}
To obtain \Cref{thm:main-inf} in the introduction from \Cref{thm:main-f}, we simply need to set $d=O(\log(n/\mu))^2$ such that $d \ge 2 + \left(\frac{12+4\log n}{\pi}\right)^2$ and $16 \sqrt{n} \exp\left(-\frac\pi2 \sqrt{d-2}\right) \le \frac{\mu}{4+\log(n)}$.
\begin{proof}
    \Cref{cor:sqrt1approx} gives us a rational function $r$ of degree $d$ with real coefficients such that 
    \[|{r}(x) - \sqrt{1-x}| \le 8  \cdot \exp\left(-\frac{\pi}{2} \sqrt{d-2} \right)\]
    for all $x \in \mathbb{C}$ with $|x| \le 1$.
    
    Let $b(x) \coloneqq \frac{r(x)}{1-x}$ and $c(x) \coloneqq \frac{1}{r(x)}$ and $f(x) \coloneqq \frac{1}{\sqrt{1-x}}$ and $g(x) \coloneqq 1/(1-x) = b(x)c(x)$.
    Let $B \coloneqq M(b,\nd), C \coloneqq M(c,\nd) \in \mathbb{R}^{\nd \times \nd}$ be as in \Cref{eq:genfuncmatrix}.
    By \Cref{lem:mgfmult}, this is a valid matrix factorization, as $B \cdot C = M(b,\nd) \cdot M(c,\nd) = M(g,\nd)$.
    
    Let $\tau > 0$ and \[\gamma_\tau \coloneqq \sup \left\{|r(x)-\sqrt{1-x}| : x \in \mathbb{C}, |x|=\exp(-\tau)\right\} \le 8  \cdot \exp\left(-\frac{\pi}{2} \sqrt{d-2} \right).\]
    By \Cref{prop:parseval}, if $\gamma_\tau \le \left(\frac{1-\exp(-2\tau)}{4}\right)^2$, then
    \begin{align*}
    \|M(b,\nd)\|_{2 \to \infty} &\le \|M(f,\nd)\|_{2 \to \infty} +  \frac{\exp(\tau \nd) \cdot \gamma_\tau}{\sqrt{\exp(2\tau)-1} },\\
    \|M(c,\nd)\|_{1 \to 2} &\le \|M(f,\nd)\|_{1 \to 2} + \frac{\exp(\tau \nd) \cdot \gamma_\tau}{\sqrt[4]{(\exp(2\tau)-1)^2 - 2^{7/2} \gamma_\tau \exp(4\tau)}}.
    \end{align*}
    Set $\tau=\frac{1}{2\nd}$. 
    Then $\exp(2\tau)-1 \ge 1/n$ and $\exp(\tau n) = \exp(1/2)$. 
    Thus \[\frac{\exp(\tau \nd) \cdot \gamma_\tau}{\sqrt{\exp(2\tau)-1}} \le \sqrt{\nd} \cdot \exp(1/2) \cdot \gamma_\tau < 2 \sqrt{n} \gamma_\tau\] and, if $\nd \ge 5$ and $\gamma_\tau \le \left(\frac{1}{6 \cdot \nd}\right)^2$, then \[\frac{\exp(\tau \nd) \cdot \gamma_\tau}{\sqrt[4]{(\exp(2\tau)-1)^2 - 2^{7/2} \gamma_\tau \exp(4\tau)}} \le \frac{\exp(1/2) \cdot \gamma_\tau}{\sqrt[4]{1/n^2 - 2^{7/2} \gamma_\tau \exp(2/n)}} \le \frac{\exp(1/2) \gamma_\tau}{\sqrt[4]{\frac{1}{n^2} - \frac{2^{7/2} \exp(2/5) }{ 36n^2} }}  < 2 \sqrt{n} \gamma_\tau.\]
    Assuming $n \ge 5$ and $d \ge 2 + \left(\frac{2}{\pi} \log(8 \cdot(6n)^2 )\right)^2$, we have $\gamma_\tau \le \left(\frac{1}{6n}\right)^2 \le  \left(\frac{1-\exp(-2\tau)}{4}\right)^2 $, which means all the above inequalities hold. Combining inequalities yields the first part of the result.
    
    Per \Cref{rem:degree}, we can write \[r(x) = t + \frac{p(x)}{q(x)}\] where $\deg(p)<d$ and $\deg(q)\le d$ and $t \in \mathbb{R}$.
    \Cref{lem:rationaltoconstrec} states that there exist $u,v \in \mathbb{R}^d$ and $W \in \mathbb{R}^{d \times d}$ such that we can express \[r(x) = \sum_{k=0}^\infty r_k x^k  ~~~\text{ for }~~~ r_k = u^T W^k v + t \mathbb{I}[k=0].\] Furthermore $W$ is sparse in the sense that each row has only two nonzero entries.
    Now we can feed $u,v,W,t$ as parameters to \Cref{alg:streamingmatrixpower}.
    \Cref{lem:algstreamingmatrixpower} tells us that \Cref{alg:streamingmatrixpower} receives the rows of $Z \in \mathbb{R}^{\nd \times \mdim}$ as a stream and returns the rows of $\widehat{Z} \in \mathbb{R}^{\nd \times \mdim}$ as a stream, where $\widehat{Z} = M(r,\nd) Z = C^{-1} Z$. (The fact that $M(r,\nd) = C^{-1}$ follows from \Cref{lem:mgfmult}.)
    The space required by the algorithm is $O(d\mdim)$ and the time per iteration is dominated by computing the matrix multiplication $WS$ for $S \in \mathbb{R}^{d \times \mdim}$. Since $W$ is sparse, this takes $O(d\mdim)$ time.
\end{proof}

\section{Factorizations via Direct Optimization} \label{sec:practical-optimization}

In this section, we begin (in \S\ref{sec:exact}) by showing the following result:
\begin{theorem}[Informal]\label{thm:closedforms}
Let $r$ be a rational generating function. Then, for the matrix factorization of $\bfA$ given by $\bfB = \LTT(r(x)/(1-x),\nd)$ and $\bfC = \LTT(1/r(x),\nd)$, we can compute in ``closed-form'' (more specifically, in time \maxerrtime) the sensitivity and error of this matrix factorization, namely $\|\bfB\|^2_{2 \to \infty}$ and $\|\bfC\|^2_{1 \to 2}$.
\end{theorem}
\begin{proof}
The result follows from \cref{lem:invogfs,lem:closedsensitivity,lem:closederr}.
\end{proof}

We remark that the $\log \nd$ dependence in \cref{thm:closedforms} comes (only) from the need to compute quantities like $\theta^k$ for $\theta \in \R$ and $k \le \nd$, which can be done via iterated squaring. For practical $\nd$ and modern hardware, this term can essentially be ignored, making the running time practically independent of $\nd$. If the roots of the numerator and denominator of $r$ are already known (as assumed in \cref{lem:invogfs}) then $\MaxErr$ can be computed in time $\calO(\nbuf^2 \log(n))$; and such roots can always be found in time $\calO(\poly(\nbuf))$ (omitting the dependence on the desired numerical accuracy).

In \cref{sec:optblt} we then show how \cref{thm:closedforms} can be used to directly optimize the parameters of a \BLT factorization to minimize $\MaxErr$ for a specific $\nd$.

\subsection{Fast and Exact Error Computation}\label{sec:exact}
In addition to enabling efficient streaming multiplication, the rational function structure of the matrix factorizations we consider allows us to compute the total error of these mechanisms efficiently for any $\nd$. 

\begin{lem}\label{lem:invogfs}
Let $r(x) = p(x)/q(x)$ be a rational function where $\deg(p) = \deg(q) = \nbuf$ with $q(0) = p(0) = 1$. Further, suppose $p$ and $q$ have pairwise distinct real roots, in particular there exist $\theta_i, \hth_i \in (0, 1]$ for $i \in [\nbuf]$ such that
\begin{align*}
    p(x) &= (1 - \hth_0 x) (1 - \hth_1 x) \cdots (1 - \hth_{d-1} x), \\
    q(x) &= (1 - \theta_0 x) (1 - \theta_1 x) \cdots (1 - \theta_{d-1} x)
\end{align*}
with $\theta_i \ne \theta_j$ and $\hat\theta_i \ne \hat\theta_j$ when $i \ne j$.
Then, there exist $\outp_i, \houtp_i \in \R$ for $i \in [\nbuf]$ such that 
\begin{equation}\label{eq:niceogf}
r(x) = 1 + x \left(\sum_{i=0}^{\nbuf-1} \frac{\omega_i}{1 - \theta_i x}\right)
\qquad \text{and} \qquad
  \rinv(x) \coloneqq \frac{1}{r(x)} = 1 + x \left(\sum_{i=0}^{\nbuf-1} \frac{\hat{\omega}_i}{1 - \hat{\theta}_i x}\right).
\end{equation}
These $\outp_i, \houtp_i$ can in fact be computed in closed form, \cref{eq:closedomega}. Further, the coefficients of the corresponding sequences can be computed directly as 
\begin{equation} \label{eq:closedcoefs}
\rc_i = 
\begin{cases}
1  & i = 0 \\
\sum_{j \in [\nbuf]} \outp_j \theta_j^{i-1} & i > 0.
\end{cases}
\quad \text{and} \quad
\rinv_i = 
\begin{cases}
1  & i = 0 \\
\sum_{j \in [\nbuf]} \houtp_j \hth_j^{i-1} & i > 0.
\end{cases}
\end{equation}
\end{lem}


Before giving a proof, we observe that \cref{eq:closedcoefs} yields a diagonalized companion matrix of the form of \cref{lem:rationaltoconstrec}, and using the representation of \cref{rem:degree} we can take
$t = 1 - \sum_{i=0}^{\nbuf-1} \outp_i/\theta_i$ so
\[
    \rc_k = t\mathbb{I}[k=0] + u\tp W^k v \coloneqq
    \begin{bmatrix}
    1 \\ 1 \\  \vdots \\ 1
    \end{bmatrix}\tp
    \begin{bmatrix}
        \theta_0  & 0        & 0      & 0\\
        0         &\theta_1  & 0      & \vdots \\
        \vdots    &  0       & \ddots &  0\\
        0         & 0        & \dots      & \theta_{\nbuf-1} \\
    \end{bmatrix}^k
    \begin{bmatrix}
    \outp_0/\theta_0 \\ \outp_1/\theta_1 \\ \vdots \\ \outp_{\nbuf-1}/\theta_{\nbuf-1}
    \end{bmatrix},
\]
and hence we can efficiently multiply by $M(r)$ using \cref{alg:streamingmatrixpower}; a similar construction with $(\houtp, \hth)$ enables efficient multiplication by $M(1/r)$ using only $\nbuf$ buffers. 

This construction is particularly useful for optimization, but assumes we already have factored $p$ and $q$; one can of course always find the roots of $p$ and $q$ if they are given in some other polynomial representation.\footnote{While we only need to consider real roots in our constructions and our optimization, the results of this section go through for complex roots as well.} Specifically, if one is instead given the matrix-power representation of $r$, \cref{prop:rationalreciprocal} of \cref{sec:bltinvmatrixparams} shows these can directly be converted to a parameterization of $1/r$.


\begin{proof}[Proof of \cref{lem:invogfs}]
\cref{eq:closedcoefs} follows from \cref{eq:niceogf} and the Taylor series
\[ 
\frac{\outp x}{1 - \theta x} = \sum_{i=1}^\infty \omega \theta^{i-1}  x^i. 
\]
It remains to show we can find $\outp_i$ to instantiate \cref{eq:niceogf} for $r(x)$ (the case of $\houtp_i$ for $\rinv(x)$ is completely symmetric). 
Letting
\[
q^{(-i)}(x) \coloneqq \frac{q(x)}{1 - \theta_i x} = \prod_{\substack{j \in [\nbuf]\\ j \neq i}} (1 - \theta_j x),
\]
we need $\outp_i$ that satisfy
\[
1 + x \left(\sum_{i=0}^{\nbuf-1} \frac{\omega_i}{1 - \theta_i x}\right)
= \frac{q(x) + \sum_{i=0}^{\nbuf-1} \outp_i x q^{(-i)}(x)}{q(x)}
= \frac{p(x)}{q(x)}.
\]
Since $p(0) = q(0) = 1$, $p(x) - q(x)$ is divisible by $x$, so $f(x) \coloneqq \frac{1}{x}(p(x) - q(x))$ is a polynomial of degree $\nbuf - 1$.  Hence, we wish to solve
\[
 \sum_{i=0}^{\nbuf-1} \outp_i q^{(-i)}(x) = f(x)
\] 
for $\outp_0, \dots, \outp_{\nbuf - 1}$. Equating the polynomial coefficients, this is a linear system of $\nbuf$ equations and $\nbuf$ unknowns. It remains to show that this system has full-rank or, equivalently, that the set of polynomials $\{q^{(-i)}\}_{i \in [\nbuf]}$ forms a basis for the vector space of polynomials of degree $\le \nbuf - 1$.

Let $z \coloneqq \prod_{i \in [\nbuf]} -\theta_i$ and $w_i \coloneqq \left(\prod_{j \ne i} (\theta_i\inv - \theta_j\inv)\right)\inv$ and define
\[\ell(x) \coloneqq \frac{q(x)}{z}
 = (x - \theta_0\inv) (x - \theta_1\inv) \cdots (x - \theta_{\nbuf-1}\inv)
 \quad \text{and} \quad
\ell^{(-i)} (x) \coloneqq \ell(x) \frac{w_i}{x - \theta_i\inv}.
\]
Then $\{\ell^{(-i)}\}_{i \in [\nbuf]}$ is exactly the Lagrange polynomial basis for interpolating points at $\theta_0\inv, \dots, \theta_{\nbuf-1}\inv$;
further, we have
\[ q^{(-i)}(x) 
   = \frac{ q(x)}{-\theta_i (x - \theta_i\inv)} 
  = \frac{-z}{\theta_i w_i} \ell^{(-i)}(x),\]
and since the $\theta_i$ are distinct and non-zero, $\frac{-\theta_i z}{w_i}$ is non-zero. It follows from Lagrange interpolation that
\[
f(x) 
  = \sum_{i \in [\nbuf]} f(\theta_i\inv) \ell^{(-i)}(x)
  = \sum_{i \in [\nbuf]} f(\theta_i\inv)\frac{-\theta_i w_i}{ z} q^{(-i)}(x)
\]
and so
\begin{equation}\label{eq:closedomega}
\omega_i = f(\theta_i\inv)\frac{-\theta_i w_i}{z}.
\end{equation}
\end{proof}

It will be useful to define the prefix sums of the geometric series by
\[
\gamma_n(\theta) \coloneqq 1 + \theta + \theta^2 + \dots + \theta^{n-1} = \sum_{i=0}^{n-1} \theta^i.
\]
with $\gamma_\infty(\theta) = 1 / (1 - \theta)$ and $\gamma_0(\theta) = 0$.
When $\abs{\theta} < 1$, we have of course
\begin{equation}\label{eq:geometricsum}
\gamma_n(\theta) = \sum_{i=0}^{n-1} \theta^i = \frac{1 -\theta^n}{1 - \theta}.
\end{equation}
(This use of $\gamma$ is unrelated to the $\gamma_2(A)$ factorization norm discussed earlier.)

\begin{lem}\label{lem:closedsensitivity}
Let $\bfC^{(\nd)} = \LTT(s(x), \nd)$ for a rational function 
\[
  \rinv(x) = 1 + x \left(\sum_{i=0}^{\nbuf-1} \frac{\houtp_i}{1 - \hth_i x}\right),
\]
for distinct $\hth_i$ as per the assumptions of \cref{lem:invogfs}. Then, there exists a closed form expression for the sensitivity of $\bfC^{(\nd)}$, $\sens(\bfC^{(\nd)}) = \|\bfC^{(\nd)}\|_{1\to2} $ allowing computation in time and space 
$\calO(\nbuf^2 \log \nd)$.
\end{lem}
\begin{proof}
The squared sensitivity can be computed as
\begin{align*}
\sens^2(\bfC) = \sum_{i=0}^{\nd - 1} \rinv_i^2 
  &= 1 + \sum_{i=0}^{n-2} \left(\sum_{j \in [\nbuf]} \houtp_j \hth^i_j\right)^2
     && \text{by \cref{eq:closedcoefs}} \\
  &= 1 + \sum_{i=0}^{n-2} 
    \left(\sum_{j \in [\nbuf]}\sum_{k \in [\nbuf]} \houtp_j\houtp_k \hth^i_j\hth^i_k\right)\\
  &= 1 +  \sum_{j \in [\nbuf]}\sum_{k \in [\nbuf]} \left(\sum_{i=0}^{n-2} 
    \houtp_j\houtp_k (\hth_j\hth_k)^i\right)\\
  &= 1 +  \sum_{j \in [\nbuf]}\sum_{k \in [\nbuf]} 
      \houtp_j\houtp_k \gamma_{n-1}(\hth_j\hth_k),
\end{align*}
and taking a square root completes the proof.
\end{proof}

\begin{lem}\label{lem:closederr}
Let $\bfB^{(\nd)} = \bfA^{(\nd)} (\bfC^{(\nd)})\inv = \bfA^{(\nd)} \LTT(r(x), \nd) = \LTT(r(x)/(1-x), \nd)$ for a rational function
\[
r(x) = 1 + x \left(\sum_{i=0}^{\nbuf-1} \frac{\omega_i}{1 - \theta_i x}\right)
\]
for distinct $\theta_i$ as per the assumptions of \cref{lem:invogfs}.
Then, there exists a closed form expression for the variance induced in the $n$th prefix sum,
$\|\bfB^{(n)}\|^2_{2 \to \infty}$ 
allowing computation in time and space 
$\calO(\nbuf^2 \log n)$.
\end{lem}
\begin{proof}
In order to make our notation more compact, we will compute sensitivity for an arbitrary row index $\nd$ of $\LTT(r(x)/(1-x), \nd+1)$.  Since multiplication by $(1-x)\inv$ corresponds to taking prefix sums, we know $\bfB = \LTT(t)$ where $t_0 = r_0 = 1$ and for any $n > 0$,
\begin{equation}\label{eq:dtocoefsums2}
t_n \coloneqq \sum_{i=0}^{n} r_i 
    = 1 + \sum_{j \in [\nbuf]} \sum_{i=0}^{n-1} \outp_j \bftheta^i_j
    = 1 + \sum_{j \in [\nbuf]} \outp_j \gamma_{n}(\bftheta_j)
\end{equation}
and we calculate
\begin{equation}\label{eq:tsquared2}
    t_{n}^2 
    =  \left( 1 + \sum_{j \in [\nbuf]} \outp_j\gamma_{n}(\bftheta_j)\right)^2 
    = 1 + 2\sum_{j \in [\nbuf]} \outp_j \gamma_{n}(\bftheta_j) 
        + \sum_{j \in [\nbuf]}\sum_{k \in [\nbuf]} \outp_j \gamma_{n}(\bftheta_j)\outp_k \gamma_{n}(\bftheta_k).
\end{equation}
It is thus sufficient to compute
\begin{align}
   \norm{\bfB\idx{n}{:} }_2^2 
    &=\sum_{i=0}^{n-1} t_i^2 \label{eq:maxerr_is_at_n} \\
    &= 1 + \sum_{i=1}^{n-1} t^2_{i} \notag \\
    &= 1 + \sum_{i=1}^{n-1}  \left( 1 + 2\sum_{j \in [\nbuf]} \outp_j \gamma_{i}(\bftheta_j) 
        + \sum_{j \in [\nbuf]}\sum_{k \in [\nbuf]} \outp_j \gamma_{i}(\bftheta_j) \outp_k \gamma_{i}(\bftheta_k)  \right).  \notag
\end{align}
where we used \cref{eq:tsquared2} for the last step.
Bringing the sum over $i$ inside the sums over the buffers $j$ and $k$, it will be sufficient to consider the following terms. For any $j \in [b]$,
\begin{align}
    \Gamma_j 
     &\coloneqq \sum_{i=1}^{n-1} \outp_j\gamma_{i}(\bftheta_j) 
     = \sum_{i=1}^{n-1}  \frac{\outp_j(1 -\bftheta_j^{i})}{1-\bftheta_j} ~~ \text{assuming $0 \le \bftheta < 1$}\notag \\
     &= \frac{\outp_j}{1 - \bftheta_j}\left((n - 1) - \sum_{i=1}^{n-1} \bftheta_j^i\right) \notag \\
     &= \frac{\outp_j}{1 - \bftheta_j}\left(n - \gamma_n(1, \bftheta_j))\right) \label{eq:ierr_j}
\end{align}
Similarly, for any $j \in [\nbuf], k \in [\nbuf]$,  let
\begin{align}
    \Gamma_{j,k} \coloneqq \sum_{i=1}^{n-1}& \outp_j \gamma_{i}( \bftheta_j) \outp_k \gamma_{i}(\bftheta_k) \notag \\
     &= \sum_{i=1}^{n-1}  \frac{\outp_j(1 -\bftheta_j^{i})}{1-\bftheta_j}
      \frac{\outp_k(1 -\bftheta_k^{i})}{1-\bftheta_k} && \text{since $0 \le \bftheta < 1$} \notag \\
     &= \frac{\outp_j\outp_k}{(1 - \bftheta_j)(1 - \bftheta_k)} 
       \sum_{i=1}^{n-1}
       \big(1 - \bftheta_j^i - \bftheta_k^i + (\bftheta_j \bftheta_k)^i\big) \notag\\     
    &= \frac{\outp_j\outp_k}{(1 - \bftheta_j)(1 - \bftheta_k)} 
       \sum_{i=0}^{n-1}
       \big(1 - \bftheta_j^i - \bftheta_k^i + (\bftheta_j \bftheta_k)^i\big) \notag\\
    &= \frac{\outp_j\outp_k}{(1 - \bftheta_j)(1 - \bftheta_k)}        
       \big(n - \gamma_n(1, \bftheta_j) - \gamma_n(1, \bftheta_k) + \gamma_n(1, \bftheta_j \bftheta_k)\big). \label{eq:ierr_jk} 
\end{align}
Putting everything together we have
\begin{equation} \label{eq:iter_err}
      \norm{\bfB\idx{n}{:} }_2^2 = \sum_{i=0}^{\nd - 1} t_i^2 = n + 2\sum_{j \in [\nbuf]} \Gamma_j + \sum_{j \in [\nbuf]}\sum_{k \in [\nbuf]} \Gamma_{j,k}.
\end{equation}
Since \cref{eq:iter_err} computes the $L_2$-norm squared for the row indexed by $n$, and this quantity is non-decreasing in $n$, it follows that
\[
\|\bfB^{(n)}\|_{2 \to \infty} = \sqrt{\norm{\bfB\idx{n-1}{:} }_2^2},
\]
which can be computed using \cref{eq:iter_err}.
\end{proof}

\begin{rem}\label{rem:growth}
Rearranging \cref{eq:iter_err} and in particular expanding the terms $\Gamma_j$ (\cref{eq:ierr_j}) and $\Gamma_{j, k}$ (\cref{eq:ierr_jk}) reveals there exists constants $\alpha_0, \alpha_1 \in \R$ and $\tilde{\omega}_i \in \R, \tilde{\theta}_i \in (0, 1]$ for $i \in [d + d^2]$ such that
\[
\|\bfB^{(n)}\|_{2 \to \infty} = \alpha_0 + \alpha_1 \cdot \nd + \sum_{i \in [d + d^2]} \tilde{\omega} \tilde{\theta}^{\nd -1}
\]
where in particular
\[
\alpha_1
 = 1 + 2\sum_{j \in [\nbuf]} \frac{\outp_j}{1 - \bftheta_j} 
      + \sum_{j \in [\nbuf]}\sum_{k \in [\nbuf]} \frac{\outp_j\outp_k}{(1 - \bftheta_j)(1 - \bftheta_k)}.
\]
We must have $\alpha_1 \ge 0$ as for a valid factorization the norm must remain positive for all $\nd$, and $\lim_{n \rightarrow \infty} \tilde{\omega}_i \tilde{\theta}^{\nd -1} = 0$ (or another constant $\tilde{\omega}_i$ if $\tilde{\theta}_i = 1$). For all the \BLTs we consider (and we believe all useful \BLTs), we find $\alpha_1 > 0$. Since sensitivity, \cref{lem:closedsensitivity}, is non-decreasing in $\nd$, for any such \BLT, eventually $\MaxErr$ must grow linearly. Hence, one must either increase $\nbuf$ with $\nd$ or optimize the \BLT for a specific $\nd$ in order to hope for good performance. In particular, by optimizing for a specific $\nd$ (\S\ref{sec:optblt}), we must keep $\alpha$ small enough that the $\alpha_1 \nd$ term remains relatively small (say, of the same order as $\alpha_0)$. However, for larger $\nd$ this term can quickly explode; we see exactly this behavior in \cref{fig:bltcompare} (Left column). 
\end{rem}

\subsection{Direct Optimization of \BLTs}\label{sec:optblt}
Inspection of \cref{lem:invogfs,lem:closedsensitivity,lem:closederr} reveals that the function $(\theta, \hth) \rightarrow \MaxErr$ is in fact differentiable. With this result, together with the automatic differentiation capabilities of JAX \citep{jax2018github}, we can in fact directly optimize for rational functions of a given degree $\nbuf$ which minimize $\MaxErr$ for a given number of steps $\nd$.

Some care is needed in the implementation. In particular, we find:
\begin{itemize}
    \item Using \texttt{float64} precision is necessary, particularly when optimizing for $\nd > 10^6$.
    \item As $\nd$ becomes large, $\max_{i \in [\nbuf]} \theta_i$ tends to 1.  \cref{eq:geometricsum,eq:ierr_j,eq:ierr_jk} can be numerically unstable in such situations; this can be remedied by switching to a Taylor-series approximation taken around $\theta=1$ for $\theta$ sufficiently near (or equal to) 1.
    \item The optimization is more stable (particularly for larger $\nd$) if we add a barrier function
    \[\ell(\theta, \omega) = 10^{-7}
         \sum_{i \in [\nbuf]} -\log(\theta_i) -\log(\outp_i)
    \]
    to ensure these parameters remain strictly positive. 
    \item For large $\nd$ (e.g $\nd >10^7$), larger numbers of buffers $\nbuf$ can make the optimization less stable. For example, with our current implementation we find for $\nd = 10^8$, optimizing for more than $\nbuf=6$ buffers actually slightly decreases performance (while theoretically it should only help).
\end{itemize}
With the above setup, we find the L-BFGS optimizer works well. Convergence generally takes less than a second running on CPUs on a modern workstration\footnote{There is approximately 5-15 seconds of overhead for JAX to just-in-time compile the loss function; this is only incurred once even if mechanisms are optimized for many different $\nd$.}, even from a naive initialization (initialization from \cref{sec:rationalfunctionapprox} would likely require even fewer iterations). Numerical results for these \OptBLT mechanisms are given in \cref{fig:main_blt_results,fig:bltcompare}.

Proving the convergence properties of this approach and deriving more robust optimization algorithms are interesting directions for future work.

\subsection{Derivation of the Inverse \BLT Parameterization}\label{sec:bltinvmatrixparams}
We now show that the matrix-power representation of a rational function $r(x)$ (see \Cref{lem:rationaltoconstrec}) can be converted to a matrix-power representation for its reciprocal $1/r(x)$. This is useful for our algorithm.
Note that the dimension of the representation increases by one.
\begin{prop}[Representation of the reciprocal of a rational generating function]\label{prop:rationalreciprocal}
    Let $u,v\in\mathbb{R}^d$ and $W \in \mathbb{R}^{d \times d}$. Assume $\langle u , v \rangle = 1$.\footnote{This assumption is only made for simplicity and holds without loss of generality, as we can always rescale the function by $r(0) = u^Tv$: That is, let $\widehat{r}(x) \coloneqq r(x)/r(0)$, apply the result, and rescale back $1/r(x) = r(0) \cdot 1/\widehat{r}(x)$.}
    Define $r(x) = \sum_{k=0}^\infty u^T W^k v \cdot x^k$.
    Then
    \begin{align}
        \frac{1}{r(x)} &= 1 - \sum_{\ell=1}^\infty u^T W (W - v u^T W)^{\ell-1} v \cdot x^\ell \label{eq:prop:rationalreciprocal1}\\
        &= \sum_{k=0}^\infty \widetilde{u}^T \widetilde{W}_0^k \widetilde{v} \cdot x^k \label{eq:prop:rationalreciprocal2}
    \end{align}
    and
    \begin{equation}
        \frac{1}{r(x) \cdot (1-x)} = \sum_{k=0}^\infty \widetilde{u}^T \widetilde{W}_1^k \widetilde{v} \cdot x^k
    \end{equation}
    where, for both $\beta \in \{0,1\}$,
    \begin{align}
        \widetilde{W}_\beta &= \left( \begin{array}{cc}
            \beta & 0 \cdots 0 \\
            v & W - v u^T W
        \end{array} \right) \in \mathbb{R}^{(d+1) \times (d+1)}, \\
        \widetilde{u} &= \left( \begin{array}{c}
            1 \\
            - W^T u
        \end{array} \right) \in \mathbb{R}^{d+1}, \\
        \widetilde{v} &= \left( \begin{array}{c}
            1 \\
            0 \\
            \vdots \\
            0
        \end{array} \right) \in \mathbb{R}^{d+1}.
    \end{align}
\end{prop}
\begin{proof}
    To prove the first part of the result \eqref{eq:prop:rationalreciprocal1} it suffices to show that
    \[
        r(x) \cdot \left( 1 - \sum_{\ell=1}^\infty u^T W (W - v u^T W)^{\ell-1} v \cdot x^\ell \right) = 1.
    \]
    We have
    \begin{align*}
        & \!\! r(x) \cdot \left( 1 - \sum_{\ell=1}^\infty u^T W (W - v u^T W)^{\ell-1} v \cdot x^\ell \right) \\
        &=   \sum_{k=0}^\infty u^T W^k v \cdot x^k \cdot \left( 1 - \sum_{\ell=1}^\infty u^T W (W - v u^T W)^{\ell-1} v \cdot x^\ell \right) \\
        &= \sum_{i=0}^\infty \left( u^T W^i v - \sum_{\ell=1}^i u^T W^{i-\ell} v \cdot u^T W (W - v u^T W)^{\ell-1} v\right) \cdot x^i \\
        &= \sum_{i=0}^\infty \left( u^T W^i v - u^T W (W - v u^T W)^{i-1} v - \sum_{\ell=1}^{i-1} u^T W^{i-\ell} v \cdot u^T W (W - v u^T W)^{\ell-1} v\right) \cdot x^i \\
        &= u^T v + \sum_{i=1}^\infty u^T \left( W^i - W (W - v u^T W)^{i-1} - \sum_{\ell=1}^{i-1} W^{i-\ell} v \cdot u^T W (W - v u^T W)^{\ell-1} \right) v \cdot x^i \\
        &= 1 + \sum_{i=1}^\infty u^T \left( W^{i-1} - (W-Wvu^T)^{i-1} - \sum_{\ell=1}^{i-1} W^{i-\ell} v \cdot u^T (W - W v u^T)^{\ell-1} \right) W v \cdot x^i ,
    \end{align*}
    where the final equality uses the fact \[ W (W - v u^T W)^{\ell-1} = W ((I-vu^T)W)^{\ell-1} = (W(I-vu^T))^{\ell-1}W = (W - W v u^T)^{\ell-1} W .\]
    Now it suffices to show that, for all $i \ge 1$, we have
    \begin{equation}
        W^{i-1} = (W-Wvu^T)^{i-1} + \sum_{\ell=1}^{i-1} W^{i-\ell} v \cdot u^T (W - W v u^T)^{\ell-1}. \label{eq:proof:prop:rationalreciprocal1}
    \end{equation}
    We prove \Cref{eq:proof:prop:rationalreciprocal1} by induction on $i$. For $i=1$, the equation is trivially true (both sides are the identity).
    Assuming \Cref{eq:proof:prop:rationalreciprocal1} holds for a given $i \ge 1$, we have
    \begin{align*}
        W^i &= W \cdot W^{i-1} = ((W-Wvu^T)+Wvu^T) \cdot W^i  \\
        &= ((W-Wvu^T)+Wvu^T) \cdot (W-Wvu^T)^{i-1} + W \cdot \sum_{\ell=1}^{i-1} W^{i-\ell} v \cdot u^T (W - W v u^T)^{\ell-1}  \\
        &= (W-Wvu^T)^i + Wvu^T \cdot (W-Wvu^T)^{i-1} + \sum_{\ell=1}^{i-1} W^{i+1-\ell} v \cdot u^T (W - W v u^T)^{\ell-1}  \\
        &= (W-Wvu^T)^i + \sum_{\ell=1}^{i} W^{i+1-\ell} v \cdot u^T (W - W v u^T)^{\ell-1}  \\ &~~~~~ - W^{i+1-i} v \cdot u^T (W-Wvu^T)^{i-1} + Wvu^T \cdot (W-Wvu^T)^{i-1} \\
        &= (W-Wvu^T)^i + \sum_{\ell=1}^{i} W^{i+1-\ell} v \cdot u^T (W - W v u^T)^{\ell-1},
    \end{align*}
    which establishes that \Cref{eq:proof:prop:rationalreciprocal1} holds for the next value of $i$.
    
    To prove \Cref{eq:prop:rationalreciprocal2}, we must show that $\widetilde{u}^T \widetilde{W}_0^0 \widetilde{v} = 1$ and $\widetilde{u}^T \widetilde{W}_0^k \widetilde{v} = - u^T W (W-vu^TW)^{k-1} v$ for $k \ge 1$.
    We immediately have $\widetilde{u}^T \widetilde{v} = 1$.
    By induction we can show that, for all $k \ge 1$, we have 
    \[
        \widetilde{W}_0^k = \left( \begin{array}{cc}
            0 & 0 \cdots 0 \\
            (W-vu^TW)^{k-1} v & (W-vu^TW)^k
        \end{array} \right) \in \mathbb{R}^{(d+1) \times (d+1)}.
    \]
    From this it follows that $\widetilde{u}^T \widetilde{W}_0^k \widetilde{v} = - u^T W (W-vu^TW)^{k-1} v$ for all $k \ge 1$, as required.
    
    More generally, for all $\beta \in \mathbb{R}$ and $k \ge 1$,
    \[
        \widetilde{W}_\beta^k = \left( \begin{array}{cc}
            \beta^k & 0 \cdots 0 \\
            \sum_{\ell=0}^{k-1} \beta^{k-1-\ell} (W-vu^TW)^{\ell} v & (W-vu^TW)^k
        \end{array} \right) \in \mathbb{R}^{(d+1) \times (d+1)}.
    \]
    Setting $\beta=1$ still gives $\widetilde{u}^T \widetilde{W}_1^0 \widetilde{v} = 1$ and now \[\widetilde{u}^T \widetilde{W}_1^k \widetilde{v} = 1 - \sum_{\ell=0}^{k-1} u^T W (W-vu^TW)^{\ell} v = 1+ \sum_{\ell=0}^{k-1} \widetilde{u}^T \widetilde{W}_0^{\ell+1} \widetilde{v} = \sum_{j=0}^{k} \widetilde{u}^T \widetilde{W}_0^{j} \widetilde{v}\] for $k \ge 1$.
    
    We have $\frac{1}{r(x)} = \sum_{k=0}^\infty \widetilde{u}^T \widetilde{W}_0^k \widetilde{v} \cdot x^k$ and, hence, \[\frac{1}{r(x)\cdot(1-x)} = \left( \sum_{k=0}^\infty \widetilde{u}^T \widetilde{W}_0^k \widetilde{v} \cdot x^k\right) \cdot \left(\sum_{i=0}^\infty x^i\right) = \sum_{\ell=0}^\infty x^\ell \sum_{j=0}^\ell \widetilde{u}^T \widetilde{W}_0^j \widetilde{v} = \sum_{\ell=0}^\infty \widetilde{u}^T \widetilde{W}_1^\ell \widetilde{v} \cdot x^\ell,\]
    which completes the proof.
\end{proof}

\section{Generalizing the Binary Tree Approach} \label{sec:generalized-binary-tree}

In this section we prove \Cref{thm:main2-inf}. Compared to \Cref{thm:main-inf} this attains a weaker approximation to the optimal matrix factorization, but is asymptotically better in terms of space usage.

Our starting point is the binary tree mechanism of \citet{Dwork-continual,CSS11-continual}.
The binary tree mechanism can be viewed as a recursive construction of a matrix factorization; see \Cref{eq:bintree} for a precise statement.
A recursion of depth $\ell$ yields a matrix factorization of size $\nd=2^\ell$ and an algorithm running in time and space $O(\ell)$.
The binary tree mechanism is within a constant factor of optimal; we will improve the constant factor to be arbitrarily close to optimal.
Our approach is to combine a recursive construction with the factorization we already developed to prove \Cref{thm:main-inf}.
We will start with a matrix factorization of size $\nd_1$ and, after a recursion of depth $\ell$, we obtain a matrix factorization of size $\nd_1^\ell$. 
The space required is $O(\ell \cdot \log^2 \nd_1)$. Appropriately setting the parameters yields  \Cref{thm:main2-inf}.



    First we define the non-cyclic shift matrix $S^{(\nd)}\in\{0,1\}^{\nd \times \nd}$ by $S_{i,j}=1 \iff i=j+1$ for all $i,j\in [\nd]$.
    If we multply the non-cyclic shift matrix with the lower-triangular all-ones matrix, it has the effect of zeroing out the diagonal and producing a strictly lower-triangular matrix.
    For example,
    \[
    S^{(6)} =
    \left(\begin{array}{cccccc}
    0 & 0 & 0 & 0 & 0 & 0\\
    1 & 0 & 0 & 0 & 0 & 0\\
    0 & 1 & 0 & 0 & 0 & 0\\
    0 & 0 & 1 & 0 & 0 & 0\\
    0 & 0 & 0 & 1 & 0 & 0\\
    0 & 0 & 0 & 0 & 1 & 0
    \end{array}\right)
    ~~~\text{ and }~~~
    A^{(6)} S^{(6)} =  S^{(6)} A^{(6)} =
    \left(\begin{array}{cccccc}
    0 & 0 & 0 & 0 & 0 & 0\\
    1 & 0 & 0 & 0 & 0 & 0\\
    1 & 1 & 0 & 0 & 0 & 0\\
    1 & 1 & 1 & 0 & 0 & 0\\
    1 & 1 & 1 & 1 & 0 & 0\\
    1 & 1 & 1 & 1 & 1 & 0
    \end{array}\right).
    \]
    The non-cyclic shift matrix corresponds to the generating function $h(x)=x$ -- i.e., $S^{(\nd)}=M(h,\nd)$.
    
    Given matrices $M \in \mathbb{R}^{n \times m}$ and $M' \in \mathbb{R}^{n' \times m'}$, define the Kronecker product 
    \iffalse 
    \begin{equation}
        M \otimes M'  = \left(
            \begin{array}{cccc}
             M_{1,1} \cdot M' & M_{1,2} \cdot M' & \cdots & M_{1,m} \cdot M' \\
             M_{2,1} \cdot M' & M_{2,2} \cdot M' & \cdots & M_{2,m} \cdot M' \\
             \vdots & \vdots & \ddots & \vdots \\
             M_{n,1} \cdot M' & M_{n,2} \cdot M' & \cdots & M_{n,m} \cdot M' \\
            \end{array}
        \right) \in \mathbb{R}^{(n \cdot n') \times (m \cdot m')} \label{eq:kronecker}
    \end{equation}
    by $(M \otimes M')_{i,j} = M_{i_1,j_1} \cdot M'_{i_2,j_2}$ where $i = (i_1-1) \cdot n' + i_2 \in [n\cdot n']$ and $j = (j_1-1) \cdot m' + j_2 \in [m\cdot m']$.
    \else 
        \begin{equation}
        M \otimes M'  = \left(
            \begin{array}{cccc}
             M_{0,0} \cdot M' & M_{0,1} \cdot M' & \cdots & M_{0,m-1} \cdot M' \\
             M_{1,0} \cdot M' & M_{1,1} \cdot M' & \cdots & M_{1,m-1} \cdot M' \\
             \vdots & \vdots & \ddots & \vdots \\
             M_{n-1,0} \cdot M' & M_{n-1,1} \cdot M' & \cdots & M_{n-1,m-1} \cdot M' \\
            \end{array}
        \right) \in \mathbb{R}^{(n \cdot n') \times (m \cdot m')} \label{eq:kronecker}
    \end{equation}
    by $(M \otimes M')_{i,j} = M_{i_1,j_1} \cdot M'_{i_2,j_2}$ where $i = i_1 \cdot n' + i_2 \in [n\cdot n']$ and $j = j_1 \cdot m' + j_2 \in [m\cdot m']$.
    \fi
    The Kronecker product has the ``mixed-product property'' with matrix multiplication: We have $(A \otimes B) (A' \otimes B') = (AA') \otimes (BB')$ whenever the matrix dimensions match so that $AA'$ and $BB'$ are well-defined.

\subsection{Recursive Matrix Factorization}   
Without further ado, we present the basis of our recursive construction. 
Suppose we are given two matrix factorizations ${A}^{(\nd_1)} = B_1 C_2$ and ${A}^{(\nd_2)} = B_2 C_2$. We will combine these into one matrix factorization for ${A}^{(\nd_1 \cdot \nd_2)}$.
\begin{defn}[Combining Matrix Factorizations]\label{defn:combc}
    For $B_1 \in \mathbb{R}^{\nd_1 \times \nd_1'}$ and $B_2 \in \mathbb{R}^{\nd_2 \times \nd_2'}$, define
    \begin{equation}
        \comb{B_1}{B_2} \coloneqq \big( ~\underbrace{~I^{(\nd_1)} \otimes B_2 ~}_{\in \mathbb{R}^{(\nd_1\nd_2) \times (\nd_1\nd_2')}}~~\underbrace{~ S^{(\nd_1)} B_1 \otimes \mathbf{1}^{(\nd_2)} ~}_{\in \mathbb{R}^{(\nd_1\nd_2) \times \nd_1'}}~ \big) \in \mathbb{R}^{(\nd_1\nd_2) \times (\nd_1\nd_2'+\nd_1')}, \label{eq:comb}
    \end{equation}
    where $I^{(\nd_1)}$ is the $\nd_1 \times \nd_1$ identity matrix and $\mathbf{1}^{(\nd_2)}$ is the all-ones column vector of length $\nd_2$.
    For $C_1 \in \mathbb{R}^{\nd_1' \times \nd_1}$ and $C_2 \in \mathbb{R}^{\nd_2' \times \nd_2}$, define
    \begin{equation}
        \comc{C_1}{C_2} \coloneqq \left( \begin{array}{c} I^{(\nd_1)} \otimes C_2 \\ C_1 \otimes (\mathbf{1}^{(\nd_2)})^T \end{array} \right) \in \mathbb{R}^{(\nd_1\nd_2'+\nd_1')\times(\nd_1\nd_2)},
    \end{equation}
    where $I^{(\nd_1)}$ is the $\nd_1 \times \nd_1$ identitiy matrix and $(\mathbf{1}^{(\nd_2)})^T$ is the all-ones row vector of length $\nd_2$
\end{defn}
\begin{lem}[Properties of \Cref{defn:combc}]\label{lem:combc}
    Let  $B_1,C_1^T \in \mathbb{R}^{\nd_1 \times \nd_1'}$ and $B_2,C_2^T \in \mathbb{R}^{\nd_2 \times \nd_2'}$.
    Suppose $B_1C_1=A^{(\nd_1)}$ and $B_2C_2=A^{(\nd_2)}$, where $A^{(\nd)}$ is the all-ones lower triangular matrix defined in \Cref{eq:allonesM}.
    Then \[(\comb{B_1}{B_2}) (\comc{C_1}{C_2}) = A^{(\nd_1\nd_2)}.\]
    Furthermore, 
    \begin{align*}
        \|\comb{B_1}{B_2}\|_{2\to\infty}^2 &= \|S^{(\nd_1)}B_1\|_{2\to\infty}^2 + \|B_2\|_{2\to\infty}^2 \le \|B_1\|_{2\to\infty}^2 + \|B_2\|_{2\to\infty}^2, \\
        \|\comc{C_1}{C_2}\|_{1\to2}^2 &= \|C_1\|_{1\to2}^2 + \|C_2\|_{1\to2}^2 .
    \end{align*}
\end{lem}
\begin{proof}
    We have
    \begin{align*}
        (\comb{B_1}{B_2}) (\comc{C_1}{C_2})
        &= \big( I^{(\nd_1)} \otimes B_2 ~~ S^{(\nd_1)} B_1 \otimes \mathbf{1}^{(\nd_2)}  \big) \left( \begin{array}{c} I^{(\nd_1)} \otimes C_2 \\ C_1 \otimes (\mathbf{1}^{(\nd_2)})^T \end{array} \right) \\
        &= (I^{(\nd_1)} \otimes B_2) (I^{(\nd_1)} \otimes C_2) + (S^{(\nd_1)} B_1 \otimes \mathbf{1}^{(\nd_2)}) (C_1 \otimes (\mathbf{1}^{(\nd_2)})^T) \\
        &= (I^{(\nd_1)} I^{(\nd_1)}) \otimes (B_2 C_2) + (S^{(\nd_1)} B_1 C_1) \otimes (\mathbf{1}^{(\nd_2)} (\mathbf{1}^{(\nd_2)})^T) \\
        &= I^{(\nd_1)} \otimes A^{(\nd_2)} + (S^{(\nd_1)} A^{(\nd_2)}) \otimes (\mathbf{1}^{(\nd_2)} (\mathbf{1}^{(\nd_2)})^T) \\
        &= A^{(\nd_1\nd_2)}.
    \end{align*}
    The last equality is best understood by looking at an example for $\nd_1=3$ and $\nd_2=2$:\footnote{If you don't like this proof by example, imagine some ellipses ($\cdots,\vdots,\ddots$) inserted into the matrices and it will look like a general proof. Formally, we can do an index-by-index case analysis.}
    \begin{align*}
        &I^{(3)} \otimes A^{(2)} + \left( S^{(3)} A^{(3)} \right) \otimes \left( \mathbf{1}^{(2)} (\mathbf{1}^{(2)})^T \right) \\
        &= 
        \left(\begin{array}{ccc}
        1 & 0 & 0 \\
        0 & 1 & 0 \\
        0 & 0 & 1
        \end{array}\right) 
        \otimes 
        \left(\begin{array}{cc}
        1 & 0 \\
        1 & 1
        \end{array}\right)
        +
        \left(\begin{array}{ccc}
        0 & 0 & 0 \\
        1 & 0 & 0 \\
        1 & 1 & 0
        \end{array}\right) 
        \otimes 
        \left(\begin{array}{cc}
        1 & 1 \\
        1 & 1
        \end{array}\right) \\    
        &=
        \left(\begin{array}{cccccc}
        1 & 0 & 0 & 0 & 0 & 0\\
        1 & 1 & 0 & 0 & 0 & 0\\
        0 & 0 & 1 & 0 & 0 & 0\\
        0 & 0 & 1 & 1 & 0 & 0\\
        0 & 0 & 0 & 0 & 1 & 0\\
        0 & 0 & 0 & 0 & 1 & 1
        \end{array}\right)
        + 
        \left(\begin{array}{cccccc}
        0 & 0 & 0 & 0 & 0 & 0\\
        0 & 0 & 0 & 0 & 0 & 0\\
        1 & 1 & 0 & 0 & 0 & 0\\
        1 & 1 & 0 & 0 & 0 & 0\\
        1 & 1 & 1 & 1 & 0 & 0\\
        1 & 1 & 1 & 1 & 0 & 0\\
        \end{array}\right) \\
        &=
        \left(\begin{array}{cccccc}
        1 & 0 & 0 & 0 & 0 & 0\\
        1 & 1 & 0 & 0 & 0 & 0\\
        1 & 1 & 1 & 0 & 0 & 0\\
        1 & 1 & 1 & 1 & 0 & 0\\
        1 & 1 & 1 & 1 & 1 & 0\\
        1 & 1 & 1 & 1 & 1 & 1
        \end{array}\right) \\
        &= A^{(6)}.
    \end{align*}
    
    Now we have
    \begin{align*}
        \|\comb{B_1}{B_2}\|_{2\to\infty}^2
        &= \max_{i \in [\nd_1\nd_2]} \sum_{j \in [\nd_1\nd_2'+\nd_1']} (\comb{B_1}{B_2})_{i,j}^2 \\
        &= \max_{i \in [\nd_1\nd_2]} \sum_{j' \in [\nd_1\nd_2']} (I^{(\nd_1)} \otimes B_2)_{i,j'}^2 + \sum_{j'' \in [\nd_1]} (S^{(\nd_1)} B_1 \otimes \mathbf{1}^{(\nd_2)})_{i,j''}^2 \\
        &= \max_{i_1 \in [\nd_1], i_2 \in [\nd_2]} \sum_{j_1' \in [\nd_1], j_2' \in [\nd_2']} (I^{(\nd_1)}_{i_1,j_1} \cdot (B_2)_{i_2,j_2'})^2 + \sum_{j'' \in [\nd_1]} ((S^{(\nd_1)} B_1)_{i_1,j''} \cdot \mathbf{1}^{(\nd_2)})_{j''})^2 \\
        &= \max_{i_1 \in [\nd_1], i_2 \in [\nd_2]} \sum_{ j_2' \in [\nd_2']} (B_2)_{i_2,j_2'}^2 + \sum_{j'' \in [\nd_1]} (S^{(\nd_1)} B_1)_{i_1,j''} ^2 \\
        &= \max_{i_1 \in [\nd_1]} \sum_{ j_2' \in [\nd_2']} (B_2)_{i_2,j_2'}^2 + \max_{ i_2 \in [\nd_2]} \sum_{j'' \in [\nd_1]} (S^{(\nd_1)} B_1)_{i_1,j''} ^2 \\
        &= \|B_2\|_{2\to\infty}^2 + \|S^{(\nd_1)}B_1\|_{2\to\infty}^2.
    \end{align*}
    Similarly, $\|\comc{C_1}{C_2}\|_{1\to2}^2 = \|C_1\|_{1\to2}^2 + \|C_2\|_{1\to2}^2$.
    Finally, \[\|S^{(\nd_1)}B_1\|_{2\to\infty}^2 = \max_{i \in [\nd_1]} \sum_{j \in [\nd_1']} (S^{(\nd_1)}B_1)_{i,j}^2  = \max_{i \in [\nd_1-1]} \sum_{j \in [\nd_1']} (B_1)_{i,j}^2 \le \|B_1\|_{2\to\infty}^2.\]
\end{proof}

\Cref{lem:combc} can be applied recursively. We start with some base factorization $B_1 C_1 = A^{(\nd_1)}$ and inductively define a factorization satisfying $B_\ell C_\ell = A^{(\nd_1^\ell)}$ and \[\|B_\ell\|_{2\to\infty} \cdot \|C_\ell\|_{1\to2} \le \ell \cdot \|B_1\|_{2\to\infty} \cdot \|C_1\|_{1\to2}.\]

\begin{prop}\label{prop:recursivemf}
    Let $B_1, C_1^T \in \mathbb{R}^{\nd_1 \times \nd'_1}$. 
    Following \Cref{defn:combc}, for $\ell \ge 2$, define 
    \begin{equation}
        B_\ell = \comb{B_1}{B_{\ell-1}} 
        ~~~ \text{ and } ~~~
        C_\ell = \comc{C_1}{C_{\ell-1}}.
    \end{equation}
    Then, for all $\ell \ge 1$, we have $B_\ell, C_\ell^T \in \mathbb{R}^{\nd_\ell \times \nd'_\ell}$, where $\nd_\ell \coloneqq \nd_1^\ell$ and \[\nd'_\ell \coloneqq \nd_1' \cdot \sum_{k=0}^{\ell-1} \nd_1^k = \nd_1' \cdot \frac{\nd_1^\ell-1}{\nd_1-1}\]
    If $B_1 C_1 = A^{(\nd_1)}$, then $B_\ell C_\ell = A^{(\nd_\ell)}$ for all $\ell \ge 1$.
    Furthermore,
    \begin{align*}
        \|B_\ell\|_{2\to\infty} &\le \sqrt{\ell} \cdot \|B_1\|_{2\to\infty},\\
        \|C_\ell\|_{1\to2} &= \sqrt{\ell} \cdot \|C_1\|_{1\to2}.
    \end{align*}
\end{prop}
\begin{proof}
    Perform induction on $\ell$ using \Cref{lem:combc} with $B_2=B_{\ell-1}$ and $C_2=C_{\ell-1}$.
\end{proof}

\subsection{Recursive Algorithm}\label{sec:recursivealg}

Now we show that the recursive matrix factorization in \Cref{prop:recursivemf} corresponds to an efficient algorithm.
The recursive construction naturally corresponds to a recursive algorithm.
That is, for streaming multiplication by $B_\ell$ we must call a base algorithm for $B_1$ and recursively call the algorithm for $B_{\ell-1}$.

Recall the streaming multiplication setting: We receive the rows of $Z \in \mathbb{R}^{\nd'_\ell \times \mdim}$ one by one as a stream and we must output the rows of $B_\ell \cdot Z \in \mathbb{R}^{\nd_\ell \times \mdim}$ as a stream.
Note that our matrix $B_\ell$ is not a lower triangular matrix; in fact, it is not even square. 
This makes the problem setting not entirely straightforward, as we do not have a one-to-one correspondence between streaming inputs and outputs.

In our application, the input $Z$ is independent random noise. In particular, the order of the rows is unimportant. 
The time when we read each row is also not important as the noise can be generated as needed. The constraint is simply that we can only read each row once. (We can store the rows for later use, but that requires memory, so we want to avoid this.) 
In other words, the input is a stream of independent random noise, which we can read at any time, but we cannot ``rewind'' it.

\begin{algorithm}
    \begin{algorithmic}
    \State \textbf{Parameters:} Base streaming algorithm $\mathcal{A}_1$. Recursion depth $\ell \ge 1$.
    \State \textbf{Streaming Input:} Row vectors $Z_{0,\cdot},Z_{1,\cdot}, \cdots, Z_{\nd'_\ell-1,\cdot} \in \mathbb{R}^{1 \times \mdim}$.
    \State \textbf{Streaming Output:} Row vectors $\widetilde{Z}_{0,\cdot},\widetilde{Z}_{1,\cdot}, \cdots, \widetilde{Z}_{\nd_\ell-1,\cdot} \in \mathbb{R}^{1 \times \mdim}$ such that
    $\widetilde{Z} = B_\ell \cdot Z$.
    \If{$\ell=1$}
        \State Run $\mathcal{A}_1$ and directly output what it outputs.
    \Else
        \State Start a copy of $\mathcal{A}_1$. (Copy I)
        \State Let $z'_0=(\mathbf{0}^{(m)})^T \in \mathbb{R}^{1 \times \mdim}$.
        \For{$i = 0 \cdots \nd_1-1$}
            \State Start a copy of $\mathcal{A}_{\ell-1}$. (Copy II) \Comment{Recursion $\ell \mapsto \ell-1$.}
            \For{$j = 0 \cdots \nd_{\ell-1}-1$}
                    \State Get $z''_{i,j} \in \mathbb{R}^{1 \times \mdim}$ the next output of $\mathcal{A}_{\ell-1}$ (copy II).
                    \State Output $\widetilde{Z}_{\nd_{\ell-1}i+j+1,\cdot}=z'_i+z''_{i,j} \in \mathbb{R}^{1 \times \mdim}$.
                    \State Delete $z''_{i,j}$.
            \EndFor
            \State Terminate $\mathcal{A}_{\ell-1}$ (copy II) and delete $z'_i$. \Comment{Free up memory.}
            \State Get $z'_{i+1} \in \mathbb{R}^{1 \times \mdim}$ the next output of $\mathcal{A}_1$ (copy I).
        \EndFor
        \State Terminate $\mathcal{A}_1$ (copy I).
    \EndIf
    \end{algorithmic}
    \caption{Recursive Streaming Algorithm corresponding to \Cref{prop:recursivemf} \label{alg:streamingrecursive}}
\end{algorithm}

Our recursive streaming algorithm is presented as \Cref{alg:streamingrecursive}.
We will instantiate the base streaming algorithm $\mathcal{A}_1$ with \Cref{alg:streamingmatrixpower}.
The recursive algorithm runs many copies of $\mathcal{A}_1$.
The streaming input is simply split among these copies. For simplicity, we do not belabor the precise order in which the input stream is read.

\begin{lem}[Properties of \Cref{alg:streamingrecursive}]\label{lem:alg-combc}
    For $\ell \ge 1$, let $B_\ell,C_\ell^T \in \mathbb{R}^{\nd_\ell \times \nd_\ell'}$ be as in \Cref{prop:recursivemf}.
    Let $\mathcal{A}_1$ be a streaming algorithm that takes as input a stream $Z_{0,\cdot},Z_{1,\cdot},\cdots,Z_{\nd_{\ell-1}'-1,\cdot} \in \mathbb{R}^{1 \times \mdim}$ and outputs a stream $\widehat{Z}_{0,\cdot},\widehat{Z}_{1,\cdot},\cdots,\widehat{Z}_{\nd_{\ell-1}-1,\cdot} \in \mathbb{R}^{1 \times \mdim}$ such that $\widehat{Z} = B_1 Z$.
    Let $\mathcal{A}_\ell$ denote \Cref{alg:streamingrecursive} instantiated with the base algorithm $\mathcal{A}_1$ and recursion depth $\ell \ge 1$.
    
    Then $\mathcal{A}_\ell$ is a streaming algorithm that takes as input a stream $Z_{0,\cdot},Z_{1,\cdot},\cdots,Z_{\nd'_\ell-1,\cdot} \in \mathbb{R}^{1 \times \mdim}$ (in an arbitrary -- but fixed -- order) and outputs a stream $\widetilde{Z}_{0,\cdot},\widetilde{Z}_{1,\cdot},\cdots,\widetilde{Z}_{\nd_\ell-1,\cdot} \in \mathbb{R}^{1 \times \mdim}$ such that $\widetilde{Z} = B_1 Z$.

    Moreover, the space required by $\mathcal{A}_\ell$ is at most $\ell$ times the space required by $\mathcal{A}_1$. And the time required by $\mathcal{A}_\ell$ is $\frac{\nd_1^{\ell+1}-1}{\nd_1-1} \le O(\nd_\ell)$ times the time required by $\mathcal{A}_1$.
\end{lem}
\begin{proof}
    We may inductively assume that the lemma holds for $\ell-1$ in place of $\ell$. That is, we assume $\mathcal{A}_{\ell-1}$ is a streaming algorithm that takes as input a stream $Z_{0,\cdot},Z_{1,\cdot},\cdots,Z_{\nd'_\ell-1,\cdot} \in \mathbb{R}^{1 \times \mdim}$ (in an arbitrary -- but fixed -- order) and outputs a stream $\check{Z}_{0,\cdot},\check{Z}_{1,\cdot},\cdots,\check{Z}_{\nd_{\ell-1}-1,\cdot} \in \mathbb{R}^{1 \times \mdim}$ such that $\check{Z} = B_{\ell-1} Z$.
    Let 
    \[
    Z' = \left( \begin{array}{c} z'_0 \\ z'_0 \\ \vdots \\ z'_0 \\ z'_1 \\ z'_1 \\ \vdots \\ z'_1 \\ \vdots \\ z'_{\nd_1-1} \\ z'_{\nd_1-1} \\ \vdots \\ z'_{\nd_1-1}  \end{array} \right) = \left( \begin{array}{c} z'_0 \\ z'_1 \\ \vdots \\  z'_{\nd_1-1}  \end{array} \right) \otimes \mathbf{1}^{(\nd_{\ell-1})} \in \mathbb{R}^{\nd_1\nd_{\ell-1}\times\mdim}
    \]
    and
    \[
    Z'' = \left(\begin{array}{c} z''_{0,0} \\ z''_{0,1} \\ \vdots \\ z''_{0,\nd_{\ell-1}} \\ z''_{1,0} \\ z''_{1,1} \\ \vdots \\ z''_{1,\nd_{\ell-1}} \\ \vdots \\ z''_{\nd_1-1,0} \\ z''_{\nd_1-1,1} \\ \vdots \\ z''_{\nd_1-1,\nd_{\ell-1}} \\\end{array}\right)  \in \mathbb{R}^{\nd_1\nd_{\ell-1}\times\mdim}.
    \]
    We have $\widetilde{Z} = Z' + Z''$. We want to show that 
    \[\widetilde{Z} = B_\ell Z = (\comb{B_1}{B_{\ell}}) Z = \big( I^{(\nd_1)} \otimes B_2 ~~ S^{(\nd_1)} B_1 \otimes \mathbf{1}^{(\nd_2)}  \big) \left( \begin{array}{c} Z_{0:\nd_1\nd'_{\ell-1},\cdot} \\ Z_{\nd_1\nd'_{\ell-1}:\nd'_\ell,\cdot} \end{array} \right),\]
    where $Z_{\nd_1\nd'_{\ell-1}:\nd'_\ell,\cdot} \in \mathbb{R}^{\nd'_1 \times \mdim}$ and $Z_{0:\nd_1\nd'_{\ell-1},\cdot} \in \mathbb{R}^{\nd_1 \nd'_{\ell-1} \times \mdim}$ are a partition of the rows of $Z \in \mathbb{R}^{\nd'_{\ell} \times \mdim}$. (Note that $\nd'_\ell = \nd_1\nd'_{\ell-1}+\nd'_1$. The notation $Z_{i:j,\cdot}$ denotes the submatrix of $Z$ formed by rows $i, i+1, \cdots, j-1$ and all columns.)
    Thus it suffices to show that $Z' = (S^{(\nd_1)} B_1 \otimes \mathbf{1}^{(\nd_2)}) Z_{\nd_1\nd'_{\ell-1}:\nd'_\ell,\cdot}$ and $Z'' = (I^{(\nd_1)} \otimes B_2) Z_{0:\nd_1\nd'_{\ell-1},\cdot}$.
    By our inductive assumption on $\mathcal{A}_{\ell-1}$, for all $i \in [\nd_1]$, \[\left(\begin{array}{c} z''_{i,0} \\ z''_{i,1} \\ \vdots \\ z''_{i,\nd_{\ell-1}}\end{array}\right) = B_\ell \cdot Z_{i\nd_{\ell-1}':(i+1)\nd_{\ell-1}',\cdot}.\]
    Thus 
    \[
    Z'' = \left(\begin{array}{c} B_\ell \cdot Z_{0:\nd_{\ell-1}',\cdot} \\ B_\ell \cdot Z_{\nd_{\ell-1}':2\nd_{\ell-1}',\cdot} \\ \vdots \\ B_\ell \cdot Z_{(\nd_1-1)\nd_{\ell-1}':\nd_1\nd_{\ell-1}',\cdot} \end{array}\right)  = (I^{(\nd_1)} \otimes B_{\ell-1}) \cdot Z_{0:\nd_1\nd_{\ell-1},\cdot} ,
    \]
    as required.
    By our assumption about the base algorithm $\mathcal{A}_1$, we have 
    \[
        \left( \begin{array}{c} z'_1 \\ z'_2 \\ \vdots \\  z'_{\nd_1}  \end{array} \right) = B_1 \cdot Z_{\nd_1\nd'_{\ell-1}:\nd_\ell',\cdot} \in \mathbb{R}^{\nd_1 \times \mdim}.
    \]
    (Note that $z'_{\nd_1}$ is never actually used in the algorithm.)
    With $z'_0=(\mathbf{0}^{(m)})^T$, the non-cyclic shift $S^{(\nd_1)}$ gives
    \[
        \left( \begin{array}{c} z'_0 \\ z'_1 \\ \vdots \\  z'_{\nd_1-1}  \end{array} \right) = S^{(\nd_1)} \cdot B_1 \cdot Z_{\nd_1\nd'_{\ell-1}:\nd_\ell',\cdot} \in \mathbb{R}^{\nd_1 \times \mdim}.
    \]
    Thus \[Z' = \left( \begin{array}{c} z'_0 \\ z'_1 \\ \vdots \\  z'_{\nd_1-1}  \end{array} \right) \otimes \mathbf{1}^{(\nd_{\ell-1})} = ( S^{(\nd_1)} \cdot B_1 \cdot Z_{\nd_1\nd'_{\ell-1}:\nd_\ell',\cdot}) \otimes  \mathbf{1}^{(\nd_{\ell-1})} = \big( ( S^{(\nd_1)} \cdot B_1 ) \otimes  \mathbf{1}^{(\nd_{\ell-1})} \big) \cdot Z_{\nd_1\nd'_{\ell-1}:\nd_\ell',\cdot},\]
    which establishes the correctness of the algorithm.
    
    The space usage of $\mathcal{A}_\ell$ is is comprised of that of $\mathcal{A}_1$ plus that of $\mathcal{A}_{\ell-1}$. (Note that, although many copies of $\mathcal{A}_{\ell-1}$ are run, they are run sequentially, so the space usage does not add up.) Thus, by induction, the space usage of $\mathcal{A}_\ell$ is $\ell$ times that of $\mathcal{A}_1$.
    
    The runtime of $\mathcal{A}_\ell$ is comprised of that of $\mathcal{A}_1$ plus $\nd_1$ times that of $\mathcal{A}_{\ell-1}$. By induction, this gives a runtime of $1 + \nd_1 + \nd_1^2 + \cdots + \nd_1^{\ell} = \frac{\nd_1^{\ell+1}-1}{\nd_1-1} = \frac{\nd'_\ell}{\nd'_1}$ times that of $\mathcal{A}_1$, as required.
\end{proof}

\subsection{Combining Recursion with \BLTs}

The algorithm in \Cref{sec:recursivealg} requires a base factorization and a corresponding algorithm.
We can instantiate this with the \RABLT rational function approximation algorithm from \Cref{sec:genfunc-framework,sec:rationalfunctionapprox}.
This combination yields the following result.

\begin{prop}[Instantiating the recursive construction with the rational function approximation]\label{prop:instantiate-recursive}
    Let $\mdim, \ell \ge 1$, $\nd_1 \ge 5$, and $d \ge 2+\left(\frac{12+4\log \nd_1}{\pi}\right)^2$ be integers.
    Let $\nd=\nd_1^\ell$ and $\nd' = \frac{\nd_1^{\ell+1}-\nd_1}{\nd_1-1}$.
    Then there exist matrices $B,C^T \in \mathbb{R}^{\nd \times \nd'}$ and a streaming algorithm $\mathcal{A}$ satisfying the following.
    \begin{itemize}
        \item \textbf{Valid Matrix Factorization:}
        $BC=A^{(\nd)}$.
        \item \textbf{Near-Optimality:}
        \[\MaxErr(B,C) \le \ell \cdot \left(\sqrt{\mathsf{OptLTToe}(\nd_1)} + 16\sqrt{\nd_1}\exp\left(-\frac{\pi}{2}\sqrt{d-2}\right)\right)^2,\]
        where $\mathsf{OptLTToe}(\nd_1) = 1+\sum_{k=1}^{\nd_1-1} \left( 2^{-2k} {2k \choose k} \right)^2 \le 1 + \frac{0.57722+\log(\nd_1)}{\pi}$.
        \item \textbf{Valid Algorithm:}
        The algorithm $\mathcal{A}$ takes as input a stream $Z_{0,\cdot},Z_{1,\cdot},\cdots,Z_{\nd'-1,\cdot} \in \mathbb{R}^{1 \times \mdim}$ in some arbitrary-but-fixed order (but only reading each $Z_{i,\cdot}$ once) and outputs a stream $\widetilde{Z}_{0,\cdot},\widetilde{Z}_{1,\cdot},\cdots,\widetilde{Z}_{\nd-1,\cdot} \in \mathbb{R}^{1 \times \mdim}$ such that $\widetilde{Z} = B \cdot Z$.
        \item \textbf{Efficient Algorithm:}
        The algorithm $\mathcal{A}$ uses space $O(\ell d \mdim)$ and has total runtime $O(\nd d \mdim)$.
    \end{itemize}
\end{prop}
\begin{proof}
    \Cref{thm:main-f} provides a base factorization $B_1, C_1 \in \mathbb{R}^{\nd_1 \times \nd_1}$ and a base algorithm $\mathcal{A}_1$ (\Cref{alg:streamingmatrixpower}) for computing $B_1Z$ in a streaming manner.
    We are guaranteed that
    \begin{align*}
            \|B_1\|_{2 \to \infty} &\le \|B_1^*\|_{2 \to \infty} + 16 \sqrt{\nd_1} \exp\left(-\frac\pi2 \sqrt{d-2}\right),\\
            \|C_1\|_{1 \to 2} &\le \|C_1^*\|_{1 \to 2} + 16 \sqrt{\nd_1} \exp\left(-\frac\pi2 \sqrt{d-2}\right),
    \end{align*}
    where $B_1^*=C_1^*=M(f,\nd_1)$ for $f(x)=1/\sqrt{1-x}$.
    By \Cref{lem:lbstruct}, $B_1^*,C_1^*$ is the optimal matrix factorization. In particular, \[\|B_1^*\|_{2 \to \infty} = \|C_1^*\|_{1 \to 2} = \sqrt{\mathsf{OptLTToe}(\nd_1)} = \sqrt{1+\sum_{k=1}^{\nd_1-1} \left( 2^{-2k} {2k \choose k} \right)^2} \le \sqrt{1 + \frac{\gamma+\log(\nd_1)}{\pi}},\]
    where $\gamma \le 0.57722$ is the Euler-Mascheroni constant.
    The algorithm $\mathcal{A}_1$ runs in space and time (per iteration) $O(d\mdim)$.
     
    Now we apply the recursive algorithm with this base construction.
    By \Cref{prop:recursivemf}, we obtain a matrix factorization $B_\ell,C_\ell\in\mathbb{R}^{\nd \times \nd'}$ such that $B_\ell C_\ell = A^{(n)}$,
    \[\|B_\ell\|_{2\to\infty} \le \sqrt{\ell} \|B_1\|_{2\to\infty} \le \sqrt{\ell}\left( \sqrt{\mathsf{OptLTToe}(\nd_1)} + 16 \sqrt{\nd_1} \exp\left(-\frac\pi2 \sqrt{d-2}\right) \right) ,\]
    and
    \[ \|C_\ell\|_{1\to2} = \sqrt{\ell} \|C_1\|_{1\to2} \le \sqrt{\ell}\left(\sqrt{\mathsf{OptLTToe}(\nd_1)} + 16 \sqrt{\nd_1} \exp\left(-\frac\pi2 \sqrt{d-2}\right)\right).\]

    By \Cref{lem:alg-combc}, \Cref{alg:streamingrecursive} is a valid streaming algorithm for computing $\widetilde{Z} = B_\ell Z$ and runs in space $O(\ell d \mdim)$ and total time $O(n d \mdim)$, as required. 
\end{proof}

Setting parameters in \Cref{prop:instantiate-recursive} yields our result.

\begin{thm}[Formal version of \Cref{thm:main2-inf}]\label{thm:main2-f}
    Let $\nd, \nd_1, \mdim \ge 1$ be integers.
    Then there exist matrices $B,C^T \in \mathbb{R}^{\nd \times \nd'}$ (for some $\nd'\le O(\nd)$) and a streaming algorithm $\mathcal{A}$ satisfying the following.
    \begin{itemize}
        \item \textbf{Valid Matrix Factorization:}
        $BC=A^{(\nd)}$.
        \item \textbf{Near-Optimality:}
        \[\MaxErr(B,C) \le \mathsf{Opt}(n) \cdot \left(1 + O\left(\frac{1}{\log(\nd_1)}\right)\right) + O(\log(\nd_1)),\]
        where $\mathsf{Opt}(n) = \inf\{\MaxErr(B_*,C_*) : B_*C_*=A^{(n)}\} = \frac{\log n}{\pi} \pm O(1)$ is the optimal value of the matrix factorization objective over all possible factorizations.
        \item \textbf{Valid Algorithm:}
        The algorithm $\mathcal{A}$ takes as input a stream $Z_{0,\cdot},Z_{1,\cdot},\cdots,Z_{\nd'-1,\cdot} \in \mathbb{R}^{1 \times \mdim}$ in some arbitrary-but-fixed order (but only reading each $Z_{i,\cdot}$ once) and outputs a stream $\widetilde{Z}_{0,\cdot},\widetilde{Z}_{1,\cdot},\cdots,\widetilde{Z}_{\nd-1,\cdot} \in \mathbb{R}^{1 \times \mdim}$ such that $\widetilde{Z} = B \cdot Z$.
        \item \textbf{Efficient Algorithm:}
        The algorithm $\mathcal{A}$ uses space $O(\log(n)\log(n_1)\mdim)$ and has total runtime $O(\nd \log(n_1)^2 \mdim)$.
    \end{itemize}
\end{thm}
To obtain \Cref{thm:main2-inf} from \Cref{thm:main2-f} we simply need to set $\nd_1$ to a value that is superconstant (so that $1/O(\log(\nd_1))=o(1)$) but not too large (so that  $O(\log(n)\log(n_1)\mdim)=\widetilde{O}(\log(n)\mdim)$). For example, we can set $\nd_1=\Theta(\log(\nd))$.
\begin{proof}
    Let $\nd_1 \ge 5$, $d = \left\lceil 2+\left(\frac{12+4\log \nd_1}{\pi}\right)^2 \right\rceil$, $\ell = \lceil \log(\nd)/\log(\nd_1)\rceil$, and $\nd' = \frac{\nd_1^{\ell+1}-\nd_1}{\nd_1-1}$ be integers.
    
    Then $\nd \le \nd_1^\ell < \nd \cdot \nd_1$. We will provide a matrix factorization of size $\nd_1^\ell$, which can be truncated to one of size $\nd$.

    \Cref{prop:instantiate-recursive} provides the matrix factorization and algorithm. It gives the near-optimality guarantee
    \[\|B_\ell\|_{2\to\infty} \cdot \|C_\ell\|_{1\to2} \le \ell \cdot \left(\sqrt{\mathsf{OptLTToe}(\nd_1)} + 16\sqrt{\nd_1}\exp\left(-\frac{\pi}{2}\sqrt{d-2}\right)\right)^2,\]
    where $\mathsf{OptLTToe}(\nd_1) = 1+\sum_{k=1}^{\nd_1-1} \left( 2^{-2k} {2k \choose k} \right)^2 \le 1 + \frac{\gamma+\log(\nd_1)}{\pi}$.
    It only remains for us to simplify this guarantee.    
        
    We assume $d \ge 2 + \left(\frac{2\log(16) + 3\log(\nd_1)}{\pi}\right)^2$ so that $16 \sqrt{\nd_1} \exp\left(-\frac\pi2 \sqrt{d-2}\right) \le 1/\nd_1$.

    Thus
    \begin{align*}
        \|B_\ell\|_{2\to\infty} \|C_\ell\|_{1\to2} 
        &\le \ell \cdot \left(\sqrt{\mathsf{OptLTToe}(\nd_1)} + \frac{1}{\nd_1} \right)^2 \\
        &= \ell \cdot \mathsf{OptLTToe}(\nd_1) + \frac{\ell}{\nd_1} \cdot \left(2\sqrt{\mathsf{OptLTToe}(\nd_1)} + \frac{1}{\nd_1}\right) \\
        &\le \ell \cdot \left( 1 + \frac{\gamma + \log(\nd_1)}{\pi} \right) + \frac{\ell}{\nd_1} \cdot  \left(2\sqrt{\nd_1} + \frac{1}{\nd_1} \right ) \\
        &=  \frac{\ell \log(\nd_1)}{\pi}  + \ell \cdot (1+\gamma/\pi) + \ell \cdot \left(\frac{2}{\sqrt{\nd_1}} + \frac{1}{\nd_1^2}\right ) \\
        &=  \frac{\log(\nd_1^\ell)}{\pi}  + \ell \cdot \left(1+\frac\gamma\pi + \frac{2}{\sqrt{\nd_1}} + \frac{1}{\nd_1^2}  \right) \\
        &\le  \frac{\log(\nd)+\log(\nd_1)}{\pi}  + \ell \cdot \left(1+\frac\gamma\pi + \frac{2}{\sqrt{5}} + \frac{1}{5^2}  \right) \\
        &\le  \frac{\log(\nd)}{\pi}  + 3\ell + \frac{\log(\nd_1)}{\pi}.
    \end{align*}
    Now we invoke the lower bound of \citet{matouvsek2020factorization} or \citet{mathias1993hadamard} (see \Cref{eq:sasho}), which tells us \[\mathsf{Opt}(\nd) \ge \frac{\log(\nd)}{\pi} -1,\]
    where $\mathsf{Opt}(\nd)$ is the optimal value of the objective function over all factorizations of size $\nd$.
    Thus
    \begin{align*}
        \|B_\ell\|_{2\to\infty} \|C_\ell\|_{1\to2} 
        &\le \mathsf{Opt}(\nd) + 1 + 3\ell + \frac{\log(\nd_1)}{\pi} \\
        &\le \mathsf{Opt}(\nd) + 4 + 3 \frac{\log(\nd)}{\log(\nd_1)} + \frac{\log(\nd_1)}{\pi} \\
        &\le \mathsf{Opt}(\nd) + 4 + 3 \frac{1+\pi \mathsf{Opt}(\nd)}{\log(\nd_1)} + \frac{\log(\nd_1)}{\pi} \\
        &= \mathsf{Opt}(\nd) \cdot\left(1 + \frac{3\pi}{\log(\nd_1)}\right) + 4 + \frac{3}{\log(\nd_1)} + \frac{\log(\nd_1)}{\pi} .
    \end{align*}
    By \Cref{prop:instantiate-recursive}, the space usage of the algorithm is $O(\ell d \mdim) = O\left(\frac{\log(\nd)}{\log(\nd_1)} (\log(\nd_1))^2 \mdim\right) = O(\log(\nd)\log(\nd_1)\mdim)$ and the total runtime is $O(\nd d \mdim) = O(\nd \log(\nd_1)^2 \mdim)$.
\end{proof}

\section{Numerical Lower Bound on Optimal Performance}\label{sec:sdp_lower}
In this section, we develop \emph{numerical} lower bounds on the optimal matrix factorizations of the all-ones lower triangular matrix $A^{(n)}$ for various classes of matrices. More precisely, given the sequence length $n$, we write down a semidefinite program that provides a lower bound on the lowest achievable error over a specific class of correlation matrices (for example, Toeplitz matrices or triangular matrices). 

In general, we consider mechanisms whose matrices can be written as 
$C=\Jmat\br{c}=\sum_i \Jmat_i c_i$ where $c \in \R^m$ denotes the parameters of the mechanism (to be optimized) and $\Jmat: \R^m \mapsto \R^{n \times n}$ is an linear operator and $\Jmat_i \in \R^{n \times n}$ are ``basis'' matrices. Any linear subspace of the space of $n \times n$ matrices (including lower triangular matrices and Toeplitz matrices) can be expressed in this form.

All results in this section derive from the following fundamental result which is a straightforward consequence of weak duality:

\begin{theorem}\label{thm:SDPlower}
Let $\Amatrix$ denote the $n \times n$ lower triangular matrix with all entries equal to $1$. 
For any $c \in \R^m$, define the objective for the correlated mechanism with correlation matrix $\Jmat\br{c}$ by
\[F\br{c}=\MaxErr(\invb{\Jmat\br{c}}\Amatrix,\Jmat\br{c}) = \norm{\Jmat\br{c}}_{1 \to 2}^2\norm{\invb{\Jmat\br{c}}\Amatrix}_{2 \to \infty}^2.\]
Further, for any set $\mathfrak{C} \subseteq \R^m$, we have 
\begin{align}\min_{c \in \mathfrak{C}} F\br{c} \geq \frac{1}{\displaystyle\min_{\substack{\Gamma_0, \ldots, \Gamma_N \in \Sym^n_+\\ \sum_i \Amatrix_i^T \Gamma_i \Amatrix_i=1}} \max_{\norm{\Jmat\br{c}}_{1\mapsto 2}=1}\inner{\tran{\Jmat\br{c}}\Jmat\br{c}}{\sum_i \Gamma_i} } \label{eq:tightlb}
\end{align}
where $\Sym^n_+$ denotes the space of symmetric positive semidefinite $n \times n$ matrices, $\Amatrix_i$ denotes the $i$-th column of $\Amatrix$ and $\inner{A}{B}=\text{trace}\br{AB}$ denotes the standard inner product between symmetric matrices $A, B \in \Sym^N$.

Furthermore, if the inner maximization can be relaxed to a concave maximization problem, i.e, there exists a compact convex set $\mathcal{C} \subset \Sym^N$ such that
\[\max_{\substack{c \in \mathfrak{C} \\ \norm{\Jmat\br{c}}_{1\mapsto 2}=1}}\inner{\tran{\Jmat\br{c}}\Jmat\br{c}}{\Gamma} = \max_{S \in \mathcal{C}} \inner{S}{\Gamma} \quad \forall \Gamma \in \Sym^N_+\]
then the inequality in \cref{eq:tightlb} is an equality. 
\end{theorem}
\begin{proof}
Since the objective is invariant to scaling $c$ by any positive constant, i.e., $F\br{\alpha c}=F\br{c}$ for all $\alpha \in \R, \alpha > 0$, we can set $\alpha=\frac{1}{\norm{\Jmat{c}}_{1 \to 2}}$ so that the optimization problem becomes
\begin{align}\min_{\substack{c \in \mathfrak{C} \\ \norm{\Jmat\br{c}}_{1 \to 2}=1}} \norm{\invb{\Jmat\br{c}}\Amatrix}_{2\to \inf}^2=\min_{\substack{c \in \mathfrak{C} \\ \norm{\Jmat\br{c}}_{1 \to 2}=1}}\max_i\norm{\invb{\tran{\Jmat\br{c}}}\Amatrix_i}^2 \label{eq:proof_lb_a}\end{align}

From Lemma \ref{lem:matrix_frac_lb}, we have that
\[\tran{\evec}\invb{M}\evec = \frac{1}{\min_{\substack{\Gamma \in \Sym^N_+ \\ \tran{\evec}\Gamma\evec}} \inner{\Gamma}{M}}\]
for any $M \in \Sym^n_++$. Applying this to $M=\tran{\Jmat\br{c}}\Jmat\br{c}$ and $\evec=\Amatrix_i $ (for $i=0, \ldots, n$) and plugging this into \eqref{eq:proof_lb_a}, we obtain
\[\max_i\norm{\invb{\tran{\Jmat\br{c}}}\Amatrix_i}^2 = \max_i \max_{\substack{\Gamma_i \succeq 0 \\ \Amatrix_i^T \Gamma_i \Amatrix_i=1}} \frac{1}{\inner{\tran{\Jmat\br{c}}\Jmat\br{c}}{\Gamma_i}}=\frac{1}{\min_i \min_{\substack{\Gamma_i \succeq 0 \\ \Amatrix_i^T \Gamma_i \Amatrix_i=1}} \inner{\tran{\Jmat\br{c}}\Jmat\br{c}}{\Gamma_i}}\]
The denominator can be rewritten as
\begin{align*}
&\min_{\nu \in \Delta^n}\sum_i \nu_i\br{\min_{\substack{\Gamma_i \succeq 0 \\ \Amatrix_i^T \Gamma_i \Amatrix_i=1}} \inner{\tran{\Jmat\br{c}}\Jmat\br{c}}{\Gamma_i}}=\min_{\nu \in \Delta^n} \sum_i \min_{\substack{\Gamma_i \succeq 0 \\ \Amatrix_i^T \Gamma_i \Amatrix_i=\nu_i}} \inner{\tran{\Jmat\br{c}}\Jmat\br{c}}{\Gamma_i}\nonumber \\
&=\min_{\substack{\Gamma_0, \ldots, \Gamma_N \in \Sym^n_+\\ \sum_i \Amatrix_i^T \Gamma_i \Amatrix_i=1}} \inner{\tran{\Jmat\br{c}}\Jmat\br{c}}{\sum_i \Gamma_i}    
\end{align*}
where $\Delta^n$ is the simplex in $n$ dimensions. Thus, the optimization problem \eqref{eq:proof_lb_a} can be rewritten as
\[\min_{\norm{\Jmat\br{c}}_{1 \to 2} \leq 1} \frac{1}{\min_{\substack{\Gamma_0, \ldots, \Gamma_N \in \Sym^n_+\\ \sum_i \Amatrix_i^T \Gamma_i \Amatrix_i=1}} \inner{\tran{\Jmat\br{c}}\Jmat\br{c}}{\sum_i \Gamma_i}} = \frac{1}{\displaystyle\max_{\norm{\Jmat\br{c}}_{1 \to 2} \leq 1}\min_{\substack{\Gamma_0, \ldots, \Gamma_n \in \Sym^n_+\\ \sum_i \Amatrix_i^T \Gamma_i \Amatrix_i=1}} \inner{\tran{\Jmat\br{c}}\Jmat\br{c}}{\sum_i \Gamma_i}}\]
By weak duality, we have that the optimal value is bounded below by
\[\frac{1}{\displaystyle\min_{\substack{\Gamma_0, \ldots, \Gamma_n \in \Sym^n_+\\ \sum_i \Amatrix_i^T \Gamma_i \Amatrix_i=1}}\max_{\norm{\Jmat\br{c}}_{1 \to 2} \leq 1} \inner{\tran{\Jmat\br{c}}\Jmat\br{c}}{\sum_i \Gamma_i}}\]
which gives \cref{eq:tightlb}.

Furthermore, if the inner maximization can be rewritten as a concave maximization problem, the overall problem is a convex-concave problem and strong duality holds by the Von-Neumann minimax theorem \citep{v1928theorie}, so that the order of the min and max can be interchanged without changing the optimal value. Hence, the lower bound is tight.

\end{proof}

\begin{cor}
Let $\mathcal{J} = \{\Jmat\br{c}: c \in \mathfrak{C}\}$. A tight lower bound in \cref{thm:SDPlower} is achieved in the following 3 cases:
\begin{itemize}
    \item $\mathcal{J}$ is the set of all lower triangular matrices.
    \item $\mathcal{J}$ is the set of all lower triangular Toeplitz matrices.
    \item $\mathcal{J}$ is the set of all lower triangular Toeplitz matrices formed from a degree $1$ constant recurrent sequence.
\end{itemize}
\end{cor}
\begin{proof}
1. If $\mathcal{J}$ is the set of all lower triangular matrices, for any $\mathbb{J} \in \mathcal{J}$, letting $\Delta=\tran{\mathbb{J}}\mathbb{J}$, we have that 
\[\max_{\norm{\mathbb{J}}_{1 \to 2} \leq 1, \mathbb{J}\in\mathcal{J}} \inner{\tran{\mathbb{J}}\mathbb{J}}{\Gamma} \leq \max_{\substack{\Delta \succeq 0 \\ \text{diag}\br{\Delta}\leq 1}} \inner{\Delta}{\Gamma}\]
since for any feasible $\mathbb{J}$, $\Delta=\tran{\mathbb{J}}\mathbb{J}$ satisfies the constraints listed in the optimization problem on the right hand side (note that $\text{diag}\br{\Delta}$ refers to the vector of diagonal elements of the matrix $\Delta$). 

Conversely, given any feasible solution $\Delta$ satisfying the constraints, one can compute its upper triangular Cholesky factorization (the regular Cholesky factorization multiplied by the matrix with anti-diagonals equal to $1$), to obtain $\mathbb{J}$ that satisfies the constraint $\norm{\mathbb{J}}_{1 \to 2} \leq 1$ and $\Delta=\tran{\mathbb{J}}\mathbb{J}$. Hence the two optimization problems have equal optimal values, and hence the lower bound in theorem \ref{thm:SDPlower} is tight.

2. If $\mathcal{J}$ is the set of all $n \times n$ lower triangular Toeplitz matrices, let $\mathbb{J}\br{c}$ denote the lower triangular Toeplitz matrix whose first column is $c$. Then, we have 
\[\norm{\Jmat\br{c}}_{1\to 2} = \norm{c}, \mathbb{J}\br{c}=\sum_i c_i H_i\]
where $H_i$ is the matrix with $0$s everywhere except the $i$-th subdiagonal equal to $1$. Thus
\[\max_{\norm{\mathbb{J}\br{c}}_{1 \to 2} \leq 1, \mathbb{J}\br{c}\in\mathcal{J}} \inner{\tran{\mathbb{J}}\mathbb{J}}{\Gamma} = \max_{\norm{c} = 1} \inner{\tran{\mathbb{J}\br{c}}\mathbb{J}\br{c}}{\Gamma}\]
and the right hand side is a standard trust region subproblem with a quadratic objective and a single quadratic constraint, that can be solved exactly via an SDP relaxation by the S-lemma \cite{polik2007survey}, i.e, there exists an affine map $Q: \Sym^n \to \Sym^n$ such that the right hand side equals $\max_{C \succeq 0, \text{trace}\br{C}=1} \inner{Q\br{\Gamma}}{ C}$. Hence, the lower bound in \cref{thm:SDPlower} is tight.

3. If $\mathcal{J}$ is the set of all $n \times n$ lower triangular Toeplitz matrices formed from a degree $1$ constant recurrent sequence, we have that $\mathbb{J}=\mathbb{J}\br{c}$ where $c_i = \alpha \theta^i$. where $\alpha, \theta \in \R$ are arbitrary real numbers. Then, the problem
\[\max_{\norm{c} = 1} \inner{\tran{\mathbb{J}\br{c}}\mathbb{J}\br{c}}{\Gamma}=\max_{c} \frac{\inner{\tran{\mathbb{J}\br{c}}\mathbb{J}\br{c}}{\Gamma}}{\tran{c}c}\]
can be written as
\[\max_{\theta, \alpha} 
\frac{\alpha^2 p\br{\theta;\Gamma}}{\alpha^2 q\br{\theta}} = \max_{\theta}\frac{p\br{\theta;\Gamma}}{q\br{\theta}}\]
where $p, q$ are polynomials in $\theta$ and the coefficients of $q$ are independent of $\Gamma$ while the coefficients of $p$ are affine functions of $\Gamma$. Using the fact that \[p\br{\theta}-\beta q\br{\theta} \leq 0 \iff p\br{\theta}-\beta q\br{\theta} \text{ is a sum of squares polynomial}\]
we can write the maximization problem as being equal to
\[\min_{p\br{\theta;\Gamma}-\beta q\br{\theta} \text{ is a sum of squares polynomial}} \beta\]
Thus, the inner maximization can be written as a concave maximization problem of the form required for lower bound tightness in Theorem \ref{thm:SDPlower} \cite{parrilo2012chapter}.

\end{proof}
\begin{figure}
    \centering
    \includegraphics[width=.7\textwidth]{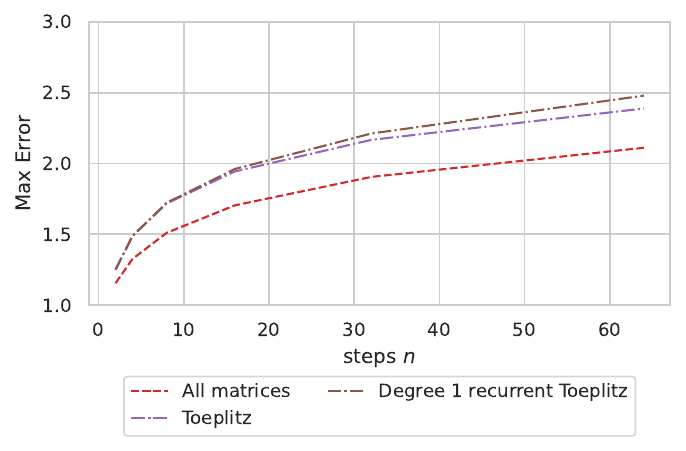}
    \caption{Semidefinite programming based lower bounds on optimal performance for various classes of matrices based on the results from theorem \ref{thm:SDPlower}}
    \label{fig:sdp_lower_bounds}
\end{figure}

\begin{lem}\label{lem:matrix_frac_lb}[Lower bound on matrix fractional function]
For any $\evec \in \R^n$ and $M \in \Sym^n_{++}$, we have
\[\tran{\evec}\invb{M}\evec = \frac{1}{\min_{\substack{\Gamma \in \Sym^n_+ \\ \tran{\evec}\Gamma\evec = 1}} \inner{\Gamma}{M}}\]
\end{lem}
\begin{proof}
Let $\evec \in \R^n$ be an arbitrary vector. 
Now, consider any positive definite matrix $M \in \Sym^{n}_+$ and eigendecomposition $M=\sum_i \lambda_i u_i\tran{u_i}$ for some orthogonal basis $\{u_i\}$. Further, expressed in this basis, suppose that $\evec=\sum_i \sigma_i u_i$. 

Then, we have
\begin{align*}
\tran{\evec}{M^{-1}}\evec = \sum_i \frac{1}{\lambda_i} \br{\tran{u_i}\One}^2 = \sum_i \frac{\sigma_i^2}{\lambda_i}
\end{align*}
Let $\inner{M}{\Gamma}$ denote the inner product in $\Sym^n$, i.e., $\inner{M}{\Gamma}=\text{trace}\br{M\Gamma}$.
Further, for any matrix $\Gamma \in \Sym^n_+$, we have
\begin{align*}
& \br{\tran{\evec}{M^{-1}}\evec}{\inner{M}{\Gamma}} = \br{\sum_i \frac{\sigma_i^2}{\lambda_i}}\br{\sum_i \lambda_i \tran{u_i} \Gamma u_i} =\br{\sum_i \frac{\sigma_i^2}{\lambda_i}}\br{\sum_i \lambda_i \norm{\Gamma^{1/2}u_i}^2}\\ & \geq \br{\sum_i \sigma_i\norm{\Gamma^{1/2}u_i}}^2 \geq \norm{\sum_i \sigma_i \Gamma^{1/2} u_i}^2 = \evec^T \Gamma \evec
\end{align*}
where the first inequality is an application of the Cauchy-Schwartz inequality and the second inequality is an application of the triangle inequality. 

Thus, for any $M, \Gamma \in \Sym^n_+$, we have
\[\tran{\evec}{M^{-1}}\evec \geq \frac{\evec^T \Gamma \evec}{\inner{M}{\Gamma}}\]
Further, choosing $\Gamma=\invb{M}\evec\tran{\evec}\invb{M}$, the inequality is replaced with equality. Hence, we have 
\[\tran{\evec}{M^{-1}}\evec = \max_{\Gamma \succeq 0} \frac{\evec^T \Gamma \evec}{\inner{M}{\Gamma}} = \max_{\substack{\Gamma \succeq 0 \\ \evec^T \Gamma \evec=1}} \frac{1}{\inner{M}{\Gamma}}\]
\end{proof}

\section{\BLT with $d=1$ Buffers}
\label{sec:warmup}


In this section, we show that it is possible to achieve $\MaxErr(B,C) = O(n^{1/6})$ error for a degree $d=1$ constant recurrent sequence, without appealing to the sophisticated machinery of rational function approximation that build the crux of the paper. Unfortunately, this approach does not naturally extend to degrees $d>1$. In comparison, adding independent noise (that is when $d=0$) results in an $\ell_\infty$ error of $O(n)$.

We consider the factorization $A=BC$, where $B=AC^{-1}$ and  $\bfC$ and $\bfC^{-1}$ are parameterized as follows. 
Let $ a,\lambda \in [-1,1]$ with $|\lambda-a^2|\leq 1 $.
Define
\begin{equation}
    \bfC= \left(\begin{array}{cccccc}
    1 & 0 & 0 & 0&\cdots & 0\\
    a^2 & 1 & 0 & 0&\cdots & 0\\
    a^2\lambda & a^2 & 1 & 0 &\cdots & 0 \\
    a^2\lambda^2 & a^2\lambda & a^2 & 1 &\cdots & 0\\
    \vdots &\vdots &\vdots &\vdots &\vdots &\vdots
    \end{array}\right),\ \ \  
    \bfC^{-1}=\left(\begin{array}{cccccc}
    1 & 0 & 0 & 0&\cdots & 0\\
    -a^2 & 1 & 0 & 0&\cdots & 0\\
    -a^2(\lambda-a^2) & -a^2 & 1 & 0 &\cdots & 0 \\
    -a^2(\lambda-a^2)^2 & -a^2(\lambda-a^2) & -a^2 & 1 &\cdots & 0\\
    \vdots &\vdots &\vdots &\vdots &\vdots &\vdots
    \end{array}\right).
    \label{eq:2}
\end{equation}
That is, $C$ corresponds to the generating function $c(x) = 1 + \frac{a^2x}{1-\lambda x}$, which clearly has degree $1$.
It can be checked that the expression for $C^{-1}$ is correct (via \cref{lem:mgfmult}) as $1/c(x) = 1 - \frac{a^2x}{1-(\lambda-a^2)x}$.

\begin{thm} 
Setting $\lambda=1-\frac{1}{n^{2/3}}$, and $a^2=\frac{1}{n^{1/3}}\left(1-\frac{1}{n^{1/3}}\right)$ in \cref{eq:2} gives 
\begin{align}
    \|C\|_{1\to2}^2 &\le O(1),\\
    \|B\|_{2\to\infty}=\|A C^{-1}\|_{2\to\infty} &\le O(n^{1/6})
\end{align}
\end{thm}

\begin{proof}

The maximum squared column norm of $\bfC$ is
\begin{equation}
    \|C\|_{1\to2}^2 = 1+ a^4\left(1+\lambda^2+\lambda^4+\cdots\lambda^{2(n-2)}\right)\leq 1+\frac{a^4}{1-\lambda^2}.\label{eq:3}
\end{equation}
Since $1-\lambda^2 = 2 n^{-2/3} - n^{-4/3} = \Theta(n^{-2/3})$ and $a^4 = n^{-2/3}(1-n^{-1/3})^2 = \Theta(n^{-2/3})$, we have $\|C\|_{1\to2}=\Theta(1)$.

Let $\bfb=[b_{n-1},\cdots,b_0]$ be the last row of $\bfB=\bfA\bfC^{-1}$. This is the longest row, so $\|B\|_{2\to\infty}^2 = \|b\|_2^2$. 
In the following, we bound each element of the vector $\bfb$. We have $b_0=1$ and, for any $i\in\{1,\cdots n-1\}$,
\begin{equation} 
    b_i= 1-a^2\left(1+(\lambda-a^2)+\cdots+(\lambda-a^2)^{i-2}\right)=1-\frac{a^2(1-(\lambda-a^2)^{i-1})}{1+a^2-\lambda}.
    \label{eq:5}
\end{equation}
Thus
\begin{align}
    \|B\|_{2\to\infty}^2 &=1 + \sum\limits_{i=1}^{n-1} \left(1-\frac{a^2(1-(\lambda-a^2)^{i-1}}{1+a^2-\lambda}\right)^2 \notag\\
    &= 1 + \sum\limits_{i=1}^{n-1} \left(\frac{(1-\lambda)+a^2(\lambda-a^2)^{i-1}}{1+a^2-\lambda}\right)^2 \notag\\
    &= 1 + \frac{1}{(1+a^2-\lambda)^2} \sum_{i=1}^{n-1} (1-\lambda)^2 + 2(1-\lambda)a^2(\lambda-a^2)^{i-1} + a^4(\lambda-a^2)^{2(i-1)} \notag\\
    &= 1 + \frac{(n-1)(1-\lambda)^2 + 2(1-\lambda)a^2\frac{1-(\lambda-a^2)^{n-1}}{1+a^2-\lambda}+a^4\frac{1-(\lambda-a^2)^{2(n-1)}}{1-(\lambda-a^2)^2}}{(1+a^2-\lambda)^2}\\
    &\le 1 + \frac{(n-1)(1-\lambda)^2 + 2(1-\lambda)a^2\frac{1}{1+a^2-\lambda}+a^4\frac{1}{1-(\lambda-a^2)^2}}{(1+a^2-\lambda)^2}.\label{eq:d1-b}
\end{align}
Now $a^2 = n^{-1/3}-n^{-2/3}$, $1-\lambda = -n^{-2/3}$, $1-\lambda+a^2 = n^{-1/3}$ and $1-(\lambda-a^2)^2 = 2n^{-1/3}-n^{-2/3}=\Theta(n^{-1/3})$.
Substituting these into \cref{eq:d1-b} gives
\begin{align*}
 \|B\|_{2\to\infty}^2 &\le 1 + \frac{(n-1)n^{-4/3} - 2 n^{-2/3} (n^{-1/3}-n^{-2/3})/n^{-1/3} + (n^{-1/3}-n^{-2/3})^2/(2n^{-1/3}-n^{-2/3})}{n^{-2/3}} \\
 &= 1 + \frac{n-1}{n^{2/3}} - 2 + \frac{2}{n^{1/3}} + \frac{n^{-2/3}-2n^{-1}+n^{-4/3}}{2n^{-1} - n^{-4/3}} \\
 &= 1 + \frac{n-1}{n^{2/3}} - 2 + \frac{2}{n^{1/3}} + \frac{n^{1/3}-2+n^{-1/3}}{2 - n^{-1/3}} \\
 &= \Theta(n^{1/3})
\end{align*}
as required.
\end{proof}

\section{Conjecture}

We give the following conjecture which improves the results in \cref{sec:rationalfunctionapprox}. If true, this would imply that \BLTs can work with $d = O(\log n)$ space complexity. 

\begin{conjecture}
    For all $n \in \mathbb{N}$ there exists a rational function $r : \mathbb{C} \to \mathbb{C}$ of degree $d=O(\log n)$ such that the following hold.
    \[
        \frac{1}{2\pi} \int_{-\pi}^\pi \left| \frac{1}{r(x)} - \frac{1}{\sqrt{1-x}} \right|^2 \mathrm{d}\theta \le \frac1n ~~~ \text{ where } ~~~ x = \exp(\sqrt{-1} \theta - 1/n).
    \]
    \[
        \frac{1}{2\pi} \int_{-\pi}^\pi \left| \frac{r(x)}{1-x} - \frac{1}{\sqrt{1-x}} \right|^2 \mathrm{d}\theta \le \frac1n ~~~ \text{ where } ~~~ x = \exp(\sqrt{-1} \theta - 1/n).
    \]
\end{conjecture}
Compared to what we prove in \cref{sec:genfuncmat-approx,sec:rationalsqrt}, this conjecture is quantitatively stronger in that $d=O(\log n)$ instead of $d = O(\log^2 n)$ but qualitatively weaker in that (i) we prove an approximation guarantee for $|x|\le\exp(-1/n)$ rather than $|x|\le1$ and, (ii) rather than a uniform bound on $|r(x)-\sqrt{1-x}|$, we bound the integrals directly.

\section*{Acknowledgements}
\addcontentsline{toc}{section}{Acknowledgements}
We thank Andreas Terzis, Shuang Song, and Arun Ganesh for comments on the draft.
We thank Aleksandar Nikolov and Jalaj Upadhyay for helpful discussions about prior work.
We thank Mark Bun for making us aware of the results of \citet{Newman64}.

\addcontentsline{toc}{section}{References}
\printbibliography  
\appendix

\end{document}